\documentclass[11pt]{article}

\usepackage{amsmath}
\usepackage{amsthm}
\usepackage{graphicx} % support the \includegraphics command and options
\usepackage{array} % for better arrays (eg matrices) in maths
\usepackage{bbm}
\usepackage{thm-restate}

\usepackage{fancybox}

\newenvironment{fminipage}%
  {\begin{Sbox}\begin{minipage}}%
  {\end{minipage}\end{Sbox}\fbox{\TheSbox}}

\newenvironment{algbox}[0]{\vskip 0.2in
\noindent 
\begin{fminipage}{6.3in}
}{
\end{fminipage}
\vskip 0.2in
}

\DeclareMathOperator*{\poly}{poly}
\DeclareMathOperator{\diag}{diag}

\newcommand{\norm}[1]{\left\lVert#1\right\rVert}
\def\abs#1{\left|#1  \right|}

\usepackage{amsmath, amssymb, amsfonts, verbatim}
\usepackage{hyphenat,epsfig,subfigure,multirow}
\usepackage{multicol}
\usepackage{nicefrac}
\usepackage{paralist}
\usepackage{enumitem}

\usepackage[ruled]{algorithm2e}

\DeclareFontFamily{U}{mathx}{\hyphenchar\font45}
\DeclareFontShape{U}{mathx}{m}{n}{
      <5> <6> <7> <8> <9> <10>
      <10.95> <12> <14.4> <17.28> <20.74> <24.88>
      mathx10
      }{}
\DeclareSymbolFont{mathx}{U}{mathx}{m}{n}
\DeclareMathSymbol{\bigtimes}{1}{mathx}{"91}

\usepackage{tcolorbox}
\tcbuselibrary{skins,breakable}
\tcbset{enhanced jigsaw}

\usepackage[normalem]{ulem}
\usepackage[compact]{titlesec}

\definecolor{DarkRed}{rgb}{0.5,0.1,0.1}
\definecolor{DarkBlue}{rgb}{0.1,0.1,0.5}

\usepackage{nameref}
\definecolor{ForestGreen}{rgb}{0.1333,0.5451,0.1333}
\definecolor{Red}{rgb}{0.9,0,0}
\usepackage[linktocpage=true,
	pagebackref=true,colorlinks,
	linkcolor=DarkRed,citecolor=ForestGreen,
	bookmarks,bookmarksopen,bookmarksnumbered]
	{hyperref}
\usepackage[noabbrev,nameinlink]{cleveref}
\crefname{property}{property}{Property}
\creflabelformat{property}{(#1)#2#3}
\crefname{equation}{eq}{Eq}
\creflabelformat{equation}{(#1)#2#3}

\usepackage{bm}
\usepackage{url}
\usepackage{xspace}
\usepackage[mathscr]{euscript}

\usepackage{tikz}
\usetikzlibrary{arrows}
\usetikzlibrary{arrows.meta}
\usetikzlibrary{shapes}
\usetikzlibrary{backgrounds}
\usetikzlibrary{positioning}
\usetikzlibrary{decorations.markings}
\usetikzlibrary{patterns}
\usetikzlibrary{calc}
\usetikzlibrary{fit}
\usetikzlibrary{decorations.pathmorphing}

\usepackage{mdframed}

\usepackage[noend]{algpseudocode}
\makeatletter
\def\BState{\State\hskip-\ALG@thistlm}
\makeatother

\usepackage{cite}
\usepackage{enumitem}

\usepackage[margin=1in]{geometry}

\newtheorem{theorem}{Theorem}[section]
\newtheorem{example}{Example}[section]
\newtheorem{lemma}{Lemma}[section]

\newtheorem{corollary}[lemma]{Corollary}

\newtheorem{fact}[lemma]{Fact}

\newtheorem{definition}[lemma]{Definition}

\DeclareMathOperator{\sign}{sgn}

\newtheorem*{theorem*}{Theorem}
\newtheorem*{claim*}{Claim}
\newtheorem*{proposition*}{Proposition}
\newtheorem*{lemma*}{Lemma}
\newtheorem*{problem*}{Problem}

\crefname{lemma}{Lemma}{Lemmas}
\crefname{claim}{Claim}{Claims}

\newtheoremstyle{restate}{}{}{\itshape}{}{\bfseries}{~(restated).}{.5em}{\thmnote{#3}}
\theoremstyle{restate}

\allowdisplaybreaks

\title{
Entrywise Approximation for Matrix Inversion and Linear Systems\footnote{A preliminary version of this work that focuses only on the cubic time algorithm has appeared in \url{https://arxiv.org/pdf/2409.10022} \cite{CGNPY24}.}
}
\author{
Mehrdad Ghadiri\footnote{mehrdadg@mit.edu} 
\and Hoai-An Nguyen\footnote{hnnguyen@andrew.cmu.edu}
\and Junzhao Yang\footnote{junzhaoy@andrew.cmu.edu}
}
\date{}

\newcommand\ee{\boldsymbol{\mathit{e}}}
\newcommand\xx{\boldsymbol{\mathit{x}}}
\newcommand\xxtil{\widetilde{\boldsymbol{\mathit{x}}}}
\newcommand\xxbar{\overline{\boldsymbol{\mathit{x}}}}
\newcommand\xxhat{\widehat{\boldsymbol{\mathit{x}}}}
\newcommand\bb{\boldsymbol{\mathit{b}}}

\newcommand\bbbar{\overline{\boldsymbol{\mathit{b}}}}
\newcommand\bbhat{\widehat{\boldsymbol{\mathit{b}}}}

\newcommand\DD{\boldsymbol{\mathit{D}}}
\renewcommand\AA{\boldsymbol{\mathit{A}}}
\newcommand\BB{\boldsymbol{\mathit{B}}}
\newcommand\CC{\boldsymbol{\mathit{C}}}

\newcommand\NN{\boldsymbol{\mathit{N}}}
\newcommand\TT{\boldsymbol{\mathit{T}}}
\renewcommand\SS{\boldsymbol{\mathit{S}}}

\newcommand\II{\boldsymbol{\mathit{I}}}
\newcommand\HH{\boldsymbol{\mathit{H}}}
\newcommand\LL{\boldsymbol{\mathit{L}}}
\newcommand\MM{\boldsymbol{\mathit{M}}}
\newcommand\UU{\boldsymbol{\mathit{U}}}
\newcommand\VV{\boldsymbol{\mathit{V}}}
\newcommand\XX{\boldsymbol{\mathit{X}}}
\newcommand\YY{\boldsymbol{\mathit{Y}}}
\newcommand\yy{\boldsymbol{\mathit{y}}}
\newcommand\vv{\boldsymbol{\mathit{v}}}

\newcommand\uu{\boldsymbol{\mathit{u}}}
\newcommand\pp{\boldsymbol{\mathit{p}}}

\newcommand\vvhat{\widehat{\boldsymbol{\mathit{v}}}}

\newcommand\ZZ{\boldsymbol{\mathit{Z}}}
\newcommand\pr{\textbf{Pr}}

\renewcommand{\sc}{\textsc{Sc}}
\newcommand\N{\mathbb{N}}
\renewcommand\P{\mathbb{P}}
\newcommand\E{\mathbb{E}}
\newcommand{\vertiii}[1]{{\left\vert\kern-0.25ex\left\vert\kern-0.25ex\left\vert #1 
    \right\vert\kern-0.25ex\right\vert\kern-0.25ex\right\vert}}

\newcommand{\approxbar}{\overset{\underline{\hspace{0.6em}}}{\approx}}
\newcommand{\Otil}{\widetilde{O}}

\newcommand{\R}{\mathbb{R}}
\newcommand{\Z}{\mathbb{Z}}
\newcommand{\Q}{\mathbb{Q}}
\newcommand{\eps}{\epsilon}
\newcommand{\dis}{\mathrm{dis}}

\newcommand{\Mbegin}{\left[\begin{matrix}}
\newcommand{\Mend}{\end{matrix}\right]}

\newcommand{\callalg}[1]{\hyperref[alg:#1]{\textsc{#1}}}
\newcommand{\labelalg}[1]{\label{alg:#1}}

\begin{document}

\pagenumbering{roman}

\maketitle

\begin{abstract}
We study matrix inversion and solving linear systems on diagonally dominant matrices. These are associated with random walk quantities such as hitting times and escape probabilities in graphs. Such quantities can be exponentially small, even on undirected unit-weighted graphs. However, their nonnegativity suggests that they can be approximated entrywise, leading to a stronger notion of approximation than vector norm–based error. 

Under this notion of error, existing Laplacian solvers and fast matrix multiplication approaches require $\Omega(mn^2)$ and $\Omega(n^{\omega+1})$ bit operations, respectively, where $m$ is the number of nonzero entries in the matrix, $n$ is its size, and $\omega$ is the matrix multiplication exponent.

We present algorithms that compute entrywise $\exp(\epsilon)$-approximate inverses of row diagonally dominant $L$-matrices (RDDL) in two settings:  
(1) when the matrix entries are given in floating-point representation;  
(2) when they are given in fixed-point representation.

For floating-point inputs, we present a cubic-time algorithm and show that it has an optimal running time under the all-pairs shortest paths (APSP) conjecture.

For fixed-point inputs, we present several algorithms for solving linear systems and inverting RDDL and SDDM matrices (the latter being symmetric RDDL matrices), all with high probability.

Omitting logarithmic factors:  
\begin{itemize}
  \item For SDDM matrices, we provide an algorithm for solving a linear system with entrywise approximation guarantees using $\widetilde{O}(m\sqrt{n})$ bit operations, and another for computing an entrywise approximate inverse using $\widetilde{O}(mn)$ bit operations.
  \item For RDDL matrices, we present an algorithm for solving a linear system with entrywise approximation guarantees using $\widetilde{O}(mn^{1+o(1)})$ bit operations, and two algorithms for computing an entrywise approximate inverse: one using $\widetilde{O}(n^{\omega+0.5})$ bit operations, and the other using $\widetilde{O}(mn^{1.5+o(1)})$ bit operations.
\end{itemize}
\end{abstract}

\clearpage
\tableofcontents
\clearpage

\pagenumbering{arabic}

\section{Introduction}

\newcounter{algcounter}

Over the past few decades, there has been significant effort to design more efficient algorithms for structured linear systems \cite{ST04:journal,PV21,FFG22}. The most prominent example of this line of work is the development of near-linear time solvers for Laplacian systems, which have enabled fast algorithms for a wide range of graph problems---from computing various probabilities associated with random walks / Markov chains~\cite{CKPPSV16,CKPPRSV17} to solving network flow problems \cite{CKMST11,M13,LRS13} and beyond~\cite{CMTV17}.

In this paper, we study linear system solving and matrix inversion for diagonally dominant matrices, where $\LL_{ii} \geq \sum_{j\neq i} \abs{\LL_{ij}}$. 
More specifically, we consider two classes: row diagonally dominant 
L-matrices (RDDL), which are principal submatrices of directed Laplacian matrices, and symmetric diagonally dominant 
M-matrices (SDDM), which are principal submatrices of undirected Laplacian matrices.

An entrywise $\exp(\eps)$-approximate solution $\widehat{\xx}$ to a linear system $\LL \xx = \bb$ with invertible $\LL$ satisfies $e^{-\epsilon} (\LL^{-1} \bb)_i \leq \widehat{\xx}_i \leq  e^{\epsilon} (\LL^{-1} \bb)_i$, for all $i$. 
We denote this by $\widehat{\xx} \approxbar_{\epsilon} \LL^{-1} \bb$. 
An entrywise approximate inverse of $\LL$ is a matrix $\ZZ$ with $e^{-\epsilon} (\LL^{-1})_{ij} \leq \ZZ_{ij} \leq e^{\epsilon} (\LL^{-1})_{ij}$, for all $i,j$, which we denote by $\LL^{-1} \approxbar_{\epsilon} \ZZ$. 

We present algorithms that compute entrywise $\exp(\epsilon)$-approximate inverses of RDDL and SDDM matrices in two settings:  
(1) when the matrix entries are given in floating-point representation;  
(2) when they are given in fixed-point representation, i.e., as integers in $[-U,U]$. 
The main difference between these two settings is that with fixed-point representation, we can represent numbers that are exponentially large or small in bit length, whereas with floating-point representation, we can (approximately) represent numbers that are \emph{doubly} exponentially large or small.
In our computational model, we measure running time by the number of bit operations (i.e., the word-RAM model), as opposed to counting arithmetic operations over exact arithmetic (i.e., the real-RAM model).

For floating-point input, we present a cubic-time algorithm for computing an entrywise approximate inverse of an RDDL matrix, using a slightly modified version of existing techniques. Surprisingly, we show that this cubic-time algorithm is optimal (up to logarithmic factors), assuming the all-pairs shortest paths (APSP) conjecture \cite{WW10:journal}, which states that computing the distances between all pairs of vertices in a graph in truly subcubic time is impossible.

For fixed-point input, we present algorithms for computing entrywise approximate solutions to linear systems, as well as for inverting RDDL and SDDM matrices. The running times of our algorithms are summarized in \Cref{tab:fixed-point-results}---in this paper, we use \emph{running time} to refer to the number of bit operations required by an algorithm. For comparison, direct application of nearly linear-time Laplacian solvers \cite{ST04:journal} and fast matrix multiplication techniques requires $\Omega(mn^2)$ and $\Omega(n^{\omega+1})$ running time, respectively, where $n$ is the dimension of the matrix, $m$ is the number of nonzero entries, and $\omega$ is the matrix multiplication exponent.

\begin{table}[h]
    \centering
    \begin{tabular}{c|c|c}
         & Solving a linear system & Matrix inversion \\
         \hline
         SDDM & $mn^{0.5} $ (\Cref{thm:entrywise-SDDM-solver-fast}) & $mn$ (\Cref{thm:quadratic-sddm-inversion}) \\
         \hline
         RDDL & $mn^{1+o(1)}$ (\Cref{thm:entrywise-RDDL-solver}) & $\min\{n^{\omega+0.5}, mn^{1.5+o(1)}\}$ (\Cref{thm:subcube-inverse-of-dense,thm:subcube-inverse-of-sparse})
    \end{tabular}
    \caption{Summary of the running time of our algorithms for diagonally dominant matrices with entries represented in fixed-point numbers. In this table, we are hiding the logarithmic factors in condition number, error and probability parameters, and absolute value of input numbers.}
    \label{tab:fixed-point-results}
\end{table}

The high running times stem from the presence of exponentially small entries in the solution, as we explain next. Nearly-linear time Laplacian solvers work by sparsifying the matrix to build a preconditioner, then applying a preconditioned iterative method (such as Chebyshev iteration) that costs $\Otil(m)$ arithmetic operations per iteration. To recover exponentially small entries, however, $\Omega(n)$ iterations are needed, each using numbers with $\Omega(n)$ bit length, leading to a total running time of $\Omega(mn^2)$. Matrix multiplication–based methods invert matrices in $\Otil(n^\omega)$ arithmetic operations, but similarly require numbers with $\Omega(n)$ bit length to recover such small entries.

A motivation for our study is the computation of probabilities associated with random walks on graphs. For a graph with vertex set $V$ and vertices $s, t, p \in V$, the escape probability of $s$ to $t$ without hitting $p$, denoted by $\P(s, t, p)$, is defined as the probability that a random walk starting at $s$ reaches $t$ before visiting $p$.
The graph illustrated in \Cref{fig:bad-example} is an example where
$\P(s, t, p)$ is exponentially small in the number of vertices.

\begin{figure}[t]
    \centering
    \begin{tikzpicture}
    % Define nodes
    \node (hub) at (0, 1) [circle, draw, fill=red!20] {$p$};
    \node (n1) at (-3, 0) [circle, draw, fill=blue!20] {$s$};
    \node (n2) at (-1.5, 0) [circle, draw, fill=blue!20] {};
    \node (n3) at (0, 0) [circle, draw, fill=blue!20] {};
    \node (n4) at (1.5, 0) [circle, draw, fill=blue!20] {};
    \node (n5) at (3, 0) [circle, draw, fill=blue!20] {$t$};
    \node (n6) at (0.75, 0) [] {$\cdots$};

    % Draw path edges
    \draw (n1) -- (n2);
    \draw (n2) -- (n3);
    \draw (n4) -- (n5);
    \draw (n3) -- (0.4, 0);
    \draw (1.1, 0) -- (n4);

    % Draw connections to the hub
    \draw (hub) -- (n1);
    \draw (hub) -- (n2);
    \draw (hub) -- (n3);
    \draw (hub) -- (n4);
    \draw (hub) -- (n5);
\end{tikzpicture}  
    \caption{Exponentially small escape probability.}
    \label{fig:bad-example}
\end{figure}

Although we use examples involving probabilities associated with random walks on graphs to motivate entrywise approximation guarantees, we emphasize that these quantities are closely related to solution of linear systems and the entries of the inverse of the corresponding diagonally dominant matrices.

\subsection{Our Results and Paper Outline}

In this section, we first state our results and then outline the paper.
We begin by presenting the following result, which provides a cubic-time algorithm for the case where the matrix entries are given in floating-point representation. We use a standard recursive approach with some modifications that enable us to work with floating-point input.

\makeatletter
\newcommand{\IfRestatedTF}[2]{\ifthmt@thisistheone #1\else #2\fi}
\makeatother
\begin{restatable}{theorem}{recursionTheorem}
    \label{thm:main-recursion}
    \IfRestatedTF { (Informal) 
    There exists an algorithm that, for any invertible RDDL matrix $\MM \in \mathbb{R}^{n \times n}$ with entries given in $L$-bit floating-point representation, and $\epsilon \in(0,1)$, computes a matrix $\ZZ$ such that $\ZZ \approxbar_{\epsilon} \NN^{-1}$ using $\Otil\left(n^3 \cdot \left(L + \log \frac{1}{\epsilon}\right)\right)$ bit operations.
    }
    {
    Let $\MM \in \mathbb{R}^{n \times n}$ be an RDDL matrix, and let $\vv \in \mathbb{R}_{\le 0}^n$ with at least one entry of $\vv$ strictly smaller than zero. Suppose the entries of $\MM$ and $\vv$ are given in $L$-bit floating-point representation. Define $\NN \in \mathbb{R}^{n \times n}$ such that for $i \neq j$, $\NN_{ij} = \MM_{ij}$, and for $i \in [n]$, $\NN_{ii} = -\left(\vv_i + \sum_{j \in [n] \setminus \{i\}} \MM_{ij}\right)$. If $\NN$ is invertible, then $\callalg{RecursiveInversion}(\MM, \vv, \epsilon)$ returns a matrix $\ZZ$ such that $\ZZ \approxbar_{\epsilon} \NN^{-1}$ using $\Otil\left(n^3 \cdot \left(L + \log \frac{1}{\epsilon}\right)\right)$ bit operations.}
\end{restatable}

To store a number $c$ around $2^{-n}$ using fixed-point numbers, we require $\Omega(n)$ time and bits of memory. In contrast, with floating-point numbers, we can store a value within a multiplicative factor of $e^{\epsilon}$ of $c$ using only $O(\log n + \log(1/\epsilon))$ time and bits of memory --- that is, $O(\log n)$ bits to store the exponent and $O(\log(1/\epsilon))$ bits to store the most significant bits of the number (i.e., the mantissa) for the required precision / approximation.
Therefore, when working with very large or very small numbers, floating-point representation is more efficient. However, the stability analysis of floating-point arithmetic is often quite complex. 

In \Cref{thm:main-recursion}, the $L$ bits refer to both the bits of precision and the bits required for the exponent of the floating-point numbers. Note that our notion of an \emph{approximate} inverse in \Cref{thm:main-recursion} is stronger than the standard guarantees in numerical linear algebra, which typically bound the error in a norm-wise sense rather than entrywise. We also emphasize that the running time of \Cref{thm:main-recursion} does not depend on the condition number (i.e., the ratio of the largest to the smallest singular value) of the matrix.

One might expect that near-linear-time algorithms for solving Laplacian systems (and, more generally, diagonally dominant systems), or approaches based on fast matrix multiplication, would achieve a better running time than \Cref{thm:main-recursion}. However, as we will discuss extensively in \Cref{sec:discussion-and-related}, due to norm-wise error bounds of these approaches, they require significantly more bit operations than \Cref{thm:main-recursion} to compute entrywise approximations.
Specifically, approaches that directly apply near-linear-time solvers or fast matrix multiplication require roughly $mn^2$ and $n^{\omega+1}$ bit operations, respectively.

In fact, we show that the all-pairs shortest paths (APSP) problem can be reduced to inverting an RDDL matrix with entries given in floating-point representation. Therefore, under the APSP conjecture \cite{WW10:journal}, our cubic-time algorithm is optimal.

\begin{restatable}{theorem}{APSPreduction}
\label{thm:RDDLInversionToAPSP}
\IfRestatedTF { (Informal)
Let $c,\delta>0$.
Assuming the APSP conjecture, there does not exist any algorithm that computes an entrywise $0.7$-approximation of an RDDL matrix $\LL \in \R^{n\times n}$ with integer entries in $[1,M]$ in time $\Otil(n^{3-\delta} \log^c M)$.
}
{
Let $n\geq 3$ and $G=([n],E,w)$ be a (directed or undirected) graph with integer edge weights in $[1,M]$. Then computing the distances $d_w(u,v)$ for all pairs $u,v\in [n]$ can be performed by computing an entrywise $0.7$-approximation of the inverse of an RDDL matrix $\MM$ where the entries of $\MM$ are represented with $O(\log n + \log M)$ bits in floating-point representation. More precisely, $\MM\in \R^{n\times n}$ is a matrix defined as
\begin{align*}
    \MM_{ij} =
    \begin{cases}
        1 & i = j, \\
        0 & (i, j) \notin E, \\
        -2^{-n^2 M w_{ij}} & (i,j) \in E.
    \end{cases}
\end{align*}
Matrix $\MM$ is symmetric if $G$ is undirected.
}
\end{restatable}

Although there is prior work on the hardness of solving linear systems or linear programs, to our knowledge, these results consider exact arithmetic or fixed-point arithmetic and rule out the existence of near-linear-time or quadratic-time algorithms \cite{KZ17:journal,DKGZ22,CGHPW23,KWZ20}.

We now turn to our results for matrices with entries in fixed-point arithmetic. Our algorithms are designed using two general frameworks: the \callalg{ThresholdDecay} framework and the \callalg{ShortcutterFramework}---see the technical overview (\Cref{sec:technical-overview}) for a detailed discussion of these frameworks. 

We prove in \Cref{thm:ThresDec-correctness} and \Cref{thm:shortcutter-correctness}, respectively, that these frameworks yield correct solutions to linear systems and matrix inversion, provided the approximation conditions specified in their steps are satisfied. This allows us to decouple the correctness of our algorithms from the bit complexity (i.e., running time under the word-RAM model) analysis and to develop different algorithms with varying running times by providing different implementations for the framework steps. 

The next result is obtained using the \callalg{ThresholdDecay} framework, with its steps implemented using the almost-linear-time solver of \cite{CKPPRSV17} for row diagonally dominant matrices. We do not use \cite{CKKPPRS18}, as its bit complexity under finite-precision arithmetic has not been analyzed.

\begin{restatable}{theorem}{rddlQuadraticSolve} \label{thm:entrywise-RDDL-solver}
    \IfRestatedTF{
    There exists an algorithm that, for any $\epsilon, \delta \in (0,1)$ and $U \ge 2$, given an $n \times n$ invertible RDDL matrix $\LL$ with $m$ nonzero integer entries in $[-U, U]$ and a nonnegative vector $\bb$ with integer entries in $[0, U]$, produces a vector $\widetilde{\xx}$ such that $\widetilde{\xx} \approxbar_{\epsilon} \LL^{-1} \bb$, with probability at least $1 - \delta$ and using $\Otil\big(mn^{1+o(1)} \log^{O(1)}(U \kappa \epsilon^{-1} \delta^{-1})\big)$ bit operations, where $\kappa$ is the condition number of $\LL$.
    }{
    For \callalg{RDDLSolver}, given the inputs $\LL, \bb, T, \eps, \kappa$ under the assumption of \Cref{thm:ThresDec-correctness} and $\delta\in (0,1)$, if $\LL$ has $m$ nonzero entries, with probability $1-\delta$, \callalg{RDDLSolver} produces a vector $\xxtil$ and a set $A \subseteq [n]$ such that $\xxtil$ has support $A$, with the following properties hold:  
    \begin{itemize}
        \item For any $i \in [n]$ with $(\LL^{-1} \bb)_i \geq (nU)^{-(T+1)} \norm{\bb}_1$, we have $i \in A$. 
        \item Our output is an entrywise $\exp(\eps)$-approximation, $\xxtil_A \approxbar_{\epsilon} (\LL^{-1} \bb)_A$.
        \item Specifically, when $T = n$, we have $A = [n]$ and $\xxtil \approxbar_{\epsilon} \LL^{-1} \bb$.
        \item It runs in $\Otil((m+n^{1+o(1)}) \cdot T \log^{O(1)}(U\kappa\epsilon^{-1} \delta^{-1}))$ bit operations.
    \end{itemize} 
    }
\end{restatable}

\Cref{thm:entrywise-RDDL-solver} is primarily useful for matrices with polynomially bounded condition numbers due to its dependence on $\log(\kappa)$. Interestingly, we show that for symmetric matrices, the entire inverse can be recovered in a similar running time of approximately $mn$. More specifically, for SDDM matrices---that is, symmetric RDDL matrices (see \Cref{def:rddl})---we provide the following result using the \callalg{ThresholdDecay} framework, with its steps implemented using the nearly-linear-time solver of \cite{ST04:journal} for symmetric diagonally dominant matrices.

\begin{restatable}{theorem}{quadSDDMInv}
\label{thm:quadratic-sddm-inversion}
\IfRestatedTF{
There exists an algorithm that, for any $\epsilon, \delta \in (0,1)$ and $U \ge 2$, given an $n \times n$ invertible SDDM matrix $\LL$ with $m$ nonzero integer entries in $[-U, U]$, produces a matrix $\ZZ$ such that $\ZZ \approxbar_{\epsilon} \LL^{-1}$ with probability at least $1 - \delta$, using $\Otil\big(mn \log^2(U \epsilon^{-1} \delta^{-1})\big)$ bit operations.
}
     {
    For \callalg{InvertSDDM}, given $\eps \in (0, 1)$, an invertible SDDM matrix $\LL$ with $m$ nonzero entries, the algorithm runs in
     $\Otil( mn \log^2 (U\eps^{-1} \delta^{-1}))$ bit operations and outputs a matrix $\ZZ$ such that $\ZZ \approxbar_{\epsilon} \LL^{-1}$.
     }
\end{restatable}

Note that the running time of \Cref{thm:quadratic-sddm-inversion} is optimal for sparse matrices (i.e., $m = O(n)$) since the algorithm requires about $n^2$ time---which is also the time needed to write all $n^2$ entries of the output. 

Next, we show that for solving a single linear system with an SDDM matrix, we can improve the running time to approximately $m\sqrt{n}$. This result also uses the \callalg{ThresholdDecay} framework, though in a somewhat more involved manner, which we will discuss in \Cref{sec:technical-overview}.

\begin{restatable}{theorem}{mrootnthm} \label{thm:entrywise-SDDM-solver-fast}
\IfRestatedTF{
There exists an algorithm that, for any $\epsilon, \delta \in (0,1)$ and $U \ge 2$, given an $n \times n$ invertible SDDM matrix $\LL$ with $m$ nonzero integer entries in $[-U, U]$ and a nonnegative vector $\bb$ with integer entries in $[0, U]$, produces a vector $\widetilde{\xx}$ such that $\widetilde{\xx} \approxbar_{\epsilon} \LL^{-1} \bb$, with probability at least $1 - \delta$ and using $\Otil( m \sqrt{n} \log^2 (U\eps^{-1} \delta^{-1}))$ bit operations.
     }{
Let $\LL, \bb, \eps$ be inputs satisfying the assumptions of \Cref{thm:ThresDec-correctness}. Let $\LL$ be an SDDM matrix with $m$ nonzero integer entries in $[-U,U]$. Then \callalg{SDDMSolverFast} runs in
     $\Otil( m \sqrt{n} \log^2 (U\eps^{-1} \delta^{-1}))$ bit operations and outputs vector $\xxtil$ such that $\xxtil \approxbar_{\epsilon} \LL^{-1} \bb$.
     }
\end{restatable}

We now turn to the inversion of RDDL matrices. Our algorithms for this purpose are designed using the \callalg{ShortcutterFramework}. Our first result on this implements the steps of the \callalg{ShortcutterFramework} using fast matrix multiplication.
In the following $\omega<2.372$ is the matrix multiplication exponent \cite{ADWXXZ25}.

\begin{restatable}{theorem}{subcubeTheorem}
\label{thm:subcube-inverse-of-dense}
There exists an algorithm such that for $n, U > 3$, $\epsilon, \delta \in (0,1)$, and any invertible RDDL matrix $\LL \in \mathbb{R}^{n \times n}$ with integer entries in $[-U, U]$, the algorithm, with probability at least $1 - \delta$, computes a matrix $\ZZ$ such that $\ZZ \approxbar_{\epsilon} \LL^{-1}$ using $\Otil\big(n^{\omega + 0.5} \log(U \epsilon^{-1}) \log(\epsilon^{-1} \delta^{-1})\big)$ bit operations in expectation.
\end{restatable}

\Cref{thm:subcube-inverse-of-dense} is well-suited for dense matrices. In the case where $m = \mathcal{O}(n)$, the next result offers a running time of approximately $n^{2.5}$, compared to the $n^{2.872}$ running time of \Cref{thm:subcube-inverse-of-dense}. However, remarkably, the running time of \Cref{thm:subcube-inverse-of-dense} does not depend on the condition number of $\LL$. This is a direct consequence of a random matrix theory result that we prove in \Cref{lemma:cond-number-of-submatrix}, which states that if we remove each row and column $i$ of the matrix $\LL$ independently with probability $1/\sqrt{n}$, then the condition number of the resulting matrix is at most about $\exp(\sqrt{n})$, even if the condition number of $\LL$ is $\exp(n)$.

Our next result uses the \callalg{ShortcutterFramework} in conjunction with \Cref{thm:entrywise-RDDL-solver} to compute the inverse of a sparse RDDL matrix.

\begin{restatable}{theorem}{subcubeTheoremSparse}
\label{thm:subcube-inverse-of-sparse}
There exists an algorithm such that for $n, U > 3$, $\epsilon, \delta \in (0,1)$, and any invertible RDDL matrix $\LL \in \mathbb{R}^{n \times n}$ with $m$ nonzero integer entries in $[-U, U]$, the algorithm, with probability at least $1 - \delta$, computes a matrix $\ZZ$ such that $\ZZ \approxbar_{\epsilon} \LL^{-1}$ using $\Otil\big(m n^{1.5 + o(1)} \log^{O(1)}(U \kappa\delta^{-1}\epsilon^{-1})\big)$ bit operations.
\end{restatable}

\paragraph{Paper organization.}
The rest of the paper is organized as follows. In \Cref{sec:discussion-and-related}, we discuss other approaches for computing escape probabilities and explain why they result in a higher number of bit operations when multiplicative error guarantees are required. 

We provide a technical overview of our frameworks and results in \Cref{sec:technical-overview}. The notation and preliminaries needed for the remainder of the paper are presented in \Cref{sec:prelim}. 
We describe the \callalg{ThresholdDecay} framework and present the proofs of \Cref{thm:entrywise-RDDL-solver,thm:quadratic-sddm-inversion,thm:entrywise-SDDM-solver-fast} in \Cref{sec:quadratic-sparse}. The \callalg{ShortcutterFramework}, along with the proofs of \Cref{thm:subcube-inverse-of-dense,thm:subcube-inverse-of-sparse}, is given in \Cref{sec:subcube-dense}. 

The proofs of \Cref{thm:main-recursion,thm:RDDLInversionToAPSP} are presented in \Cref{sec:recursive-SC}. 
In \Cref{sec:probabilities}, we discuss the relationship between RDDL matrices, their inverses, and certain graph-theoretic quantities such as escape probabilities and hitting times. While parts of \Cref{sec:probabilities} have appeared previously in the literature \cite{CKPPSV16}, we offer a simplified and accessible exposition.

If the reader is unfamiliar with diagonally dominant matrices and their connection to random walks, we recommend reading \Cref{sec:probabilities} after \Cref{sec:prelim} before proceeding further. \Cref{sec:repeated-square} discusses a simple algorithm based on repeated squaring that applies to graphs with edge weights in a polynomial range. We believe this algorithm may be of particular practical interest. Finally, we conclude in \Cref{sec:conclusion} with several open problems related to our results.

\subsection{Discussion and Related Work}
\label{sec:discussion-and-related}

In this section, we first discuss several some applications of entrywise approximation for solutions of diagonally dominant linear systems, particularly in the context of random walks on graphs. We then review existing approaches in the literature for computing entrywise approximate solutions, highlighting their high running times for both linear system solving and matrix inversion.

As illustrated in \Cref{lemma:RDDL-escape}, the escape probability can be computed by solving a linear system. More specifically, given the Laplacian matrix $\LL$ associated with a graph $G$ with vertex set $V$, and $s, t, p \in V$,
\[
\P(s, t, p) = \frac{(\LL_{-p, -p}^{-1})_{st}}{\LL_{tt} (\LL_{-p, -p}^{-1})_{tt}},
\]
where $\LL_{-p, -p}$ denotes the matrix obtained by removing the row and column corresponding to $p$.

Therefore, various approaches—including near-linear time Laplacian solvers, fast matrix multiplication, and iterative methods—can be used to compute the escape probability. However, we argue that all of these approaches incur high running times (measured in bit operations) when computing (exponentially) small escape probabilities, as illustrated in \Cref{fig:bad-example}.

Another application of our work is the computation of hitting time (i.e., the expected number of steps for a random walk to travel from vertex $s$ to vertex $t$) in directed graphs. While the hitting time is always polynomially bounded \cite{L86} in undirected graphs, it can be exponentially large in directed graphs. The graph in \Cref{fig:bad-example-hitting} provides an example where the hitting time from $s$ to $t$ is exponentially large. By \eqref{eq:hitting-time-recurrence}, the hitting time from $s$ to $t$ is given by $(\LL_{-t, -t}^{-1} \boldsymbol{1})_s$, where $\boldsymbol{1}$ is the all-ones vector. In this case, norm-wise error bounds are inadequate for recovering small hitting times; see \Cref{ex:norm-wise-bad}.

\begin{figure}[h]
    \centering
    \begin{tikzpicture}
    % Define nodes
    \node (hub) at (0, 1) [circle, draw, fill=blue!20] {};
    \node (n1) at (-3, 0) [circle, draw, fill=blue!20] {$s$};
    \node (n2) at (-1.5, 0) [circle, draw, fill=blue!20] {};
    \node (n3) at (0, 0) [circle, draw, fill=blue!20] {};
    \node (n4) at (1.5, 0) [circle, draw, fill=blue!20] {};
    \node (n5) at (3, 0) [circle, draw, fill=blue!20] {$t$};
    \node (n6) at (0.75, 0) [] {$\cdots$};

    % Draw path edges
    \draw[-{Stealth[scale=1.5]}] (n1) -- (n2);
    \draw[-{Stealth[scale=1.5]}] (n2) -- (n3);
    \draw[-{Stealth[scale=1.5]}] (n4) -- (n5);
    \draw (n3) -- (0.4, 0);
    \draw (1.1, 0) -- (n4);

    % Draw connections to the hub
    \draw[-{Stealth[scale=1.5]}] (hub) -- (n1);
    \draw[-{Stealth[scale=1.5]}] (n2) -- (hub);
    \draw[-{Stealth[scale=1.5]}] (n3) -- (hub);
    \draw[-{Stealth[scale=1.5]}] (n4) -- (hub);
    \draw[-{Stealth[scale=1.5]}] (n5) -- (hub);
\end{tikzpicture}  
    \caption{Exponentially large hitting time.}
    \label{fig:bad-example-hitting}
\end{figure}

Another example is the entries of the personalized $\textsc{PageRank}$ matrix which is defined as $\beta(\II - (1-\beta) \AA)^{-1}$, where $\AA$ is the matrix corresponding to a random walk and $\beta\in[0,1]$. The matrix $\II - (1-\beta) \AA$ can be considered as the diagonally dominant matrix associated with a random walk where at each step the walk ends up at a dummy vertex $p$ with probability $\beta$. Therefore the entries of $(\II - (1-\beta) \AA)^{-1}$ can be interpreted in connection with escape probabilities (see \Cref{sec:probabilities}).

\paragraph{Near-linear-time Laplacian solvers.} 
Spielman and Teng \cite{ST04:journal} have shown that for a symmetric diagonally dominant matrix $\LL$, one can approximately solve the linear system $\LL \xx = \bb$ with $O(m \log^c n \log \epsilon^{-1})$ arithmetic operations on numbers with $O(\log(\kappa(\LL)) \log^c n \log \epsilon^{-1})$ bits of precision. Here, $\kappa(\LL)$ is the condition number of $\LL$ and $c$ is a constant. More specifically, their algorithm outputs $\widetilde{\xx}$ such that $\norm{\widetilde{\xx} - \LL^{\dagger} \bb}_{\LL} \leq \epsilon \cdot \norm{\LL^{\dagger} \bb}_{\LL}$, where $\LL^{\dagger}$ is the pseudo-inverse of $\LL$.
By taking $\epsilon'$ to be smaller than $\epsilon/(\kappa(\LL))^2$ and orthogonalizing against the all-$1$s vector on each connected component,
we get $\norm{\widetilde{\xx} - \LL^{\dagger} \bb}_{2} \leq \epsilon' \cdot \norm{\LL^{\dagger} \bb}_{2}$.
This bounds the norm-wise error of the computed vector of escape probabilities. The next example shows why norm-wise error bounds are not suited for computing small escape probabilities.

\begin{example}
\label{ex:norm-wise-bad}
Let $\vv=(10^{15},1)$, $\widetilde{\vv}=(10^{15}(1+\epsilon),0)$, and $\epsilon=10^{-5}$. Note that the second entry of $\vv$ and $\widetilde{\vv}$ are very different from each other.
Then
\begin{align}
\frac{\norm{\widetilde{\vv} - \vv}_1}{\norm{\vv}_1} = \frac{\epsilon\cdot 10^{15}+1}{10^{15}+1} = \frac{10^{10}+1}{10^{15}+1} \approx 10^{-5}.
\end{align}
This means that $\widetilde{\vv}$ approximates $\vv$ in a norm-wise manner, but entry-wise $\vv$ and $\widetilde{\vv}$ are very different vectors.
Therefore a bound on the norm does not necessarily provide bounds for the entries. For the norm bound to give guarantees for the entries, we would need to have $\epsilon \approx \frac{\vv_{\min}}{\vv_{\max}}$, where $\vv_{\min}$ and $\vv_{\max}$ are the smallest and largest entry of $\vv$ in terms of absolute value, respectively.
\end{example}

Now note that for the graph in \Cref{fig:bad-example}, we have to take the error parameter $\epsilon$ to be exponentially small (in $n$) to be able to find a multiplicative approximation to the smallest escape probability. This means that the total number of bit operations for the Spielman-Teng algorithm will be $\Otil(m n^2)$. Later works that improve the constant $c$ \cite{KMST10,KMP11,LS13,PS14,KMP14,CKMPPRX14,JS21,KS16}, make the algorithm simpler \cite{KOSZ13,LS13,KS16} or parallelize it \cite{BGKMPT11,PS14,KLPSS16}, and algorithms for directed Laplacians
\cite{CKPPSV16,CKPPRSV17,CKKPPRS18,PS22} all have the same dependencies on $\log(1/\epsilon)$ and similar norm-wise guarantees.

\paragraph{Fast-matrix-multiplication.}

Strassen \cite{S69} has shown that two $n$-by-$n$ matrices can be multiplied with fewer than $n^3$ arithmetic operations. Currently, the best bound for the number of arithmetic operations for fast-matrix-multiplication is $O(n^{\omega})$, where $\omega < 2.372$ \cite{ADWXXZ25} which is based on techniques from \cite{CW90}. It is also well-known that matrix inversion can be reduced to polylogarithmic matrix multiplications. Therefore a linear system can be solved in $O(n^{\omega})$ arithmetic operations.

The stability of Strassen's algorithm and other fast-matrix-multiplication algorithms have been a topic of debate for decades \cite{M75,BLS91,H90}. Although it is established that such algorithms are stable \cite{DDHK07,DDH07}, similar to near-linear-time Laplacian solvers, such stability only holds in a norm-wise manner. In other words, fast-matrix-multiplication algorithms produce errors that can only be bounded norm-wise. This is observed by Higham \cite{H90}, and he provides an explicit example: 

\begin{align}
C = \begin{bmatrix}
1 & 0 \\ 0 & 1
\end{bmatrix}
\begin{bmatrix}
1 & \epsilon \\ \epsilon & \epsilon^2
\end{bmatrix}.
\end{align}

As Higham \cite{H90} shows, Strassen's algorithm fails to compute $C_{22}$ accurately and has less favorable stability properties than conventional matrix multiplication: it satisfies only a weaker, norm-wise error bound rather than a entrywise one. This norm-wise bound arises because Strassen’s method combines elements across the entire matrix $\AA$ (and similarly for $\BB$), rather than operating entrywise.

Higham further explains that this issue stems from the presence of genuine subtractions in Strassen’s algorithm, meaning that it necessarily involves adding positive and negative numbers together, assuming all additions involve nonzero terms. As noted in \cite{GL89}, this makes Strassen’s method unattractive in applications where all elements of $\AA$ and $\BB$ are nonnegative, such as in Markov processes \cite{H87}. 

Therefore for any algorithm based on fast-matrix-multiplication, we need to take the error parameter exponentially small in $n$. This gives $n^{\omega+1}$ bit operations which is worse than $n^3$.

\paragraph{Shifted and $p$-adic numbers.} 
Another approach for solving linear systems is to use $p$-adic numbers \cite{S05,D82}.
The advantage of this approach, which relies on Cramer's rule, is that the solution is computed exactly in rational number representation.
In addition, the running time of these algorithms does not depend on error parameters or the condition number of the matrix.
However, such algorithms work only with integer input matrices and vectors.
Therefore, they can be used to obtain running times of $O(n^{\omega} L)$, where $L$ is the bit complexity of the input numbers in fixed-point representation.
This running time holds even for general matrices (not just diagonally dominant ones) and produces \emph{exact} solutions in rational form.

Note that there can be exponentially small or large numbers that require $n$ bits for representation in fixed-point format but only $O(\log n)$ bits in floating-point format.
Such input would result in a running time of $O(n^{\omega + 1})$ for $p$-adic–based approaches, which is worse than $n^3$.
Adapting these algorithms to floating-point inputs, if possible, would require significant modifications.
Moreover, note that these approaches solve only a single linear system and do not compute the inverse; thus, to recover the inverse, they would require solving $n$ linear systems, resulting in a total running time of $O(n^{\omega + 1} L)$.

\paragraph{Gaussian elimination.}

There are empirical studies in the literature showing that Gaussian elimination is more stable than other approaches for computing the stationary distribution of Markov chains \cite{H87,GTH85,HP84}.
It has been observed and stated that this increased stability arises because the process avoids subtractions in this case.

Subsequent works have analyzed error bounds and the number of floating-point operations required for computing the stationary distribution \cite{O93} and for computing the LU factorization \cite{O96}.
Although the overarching idea in these works and in our \Cref{thm:main-recursion}—to perform operations only on nonnegative numbers—is similar, our approach is more parallelizable and suitable for large $n$ due to its recursive nature, whereas the prior works focus on Gaussian elimination, which requires $n$ sequential iterations.

\paragraph{Monte Carlo Based Methods.} Due to the close connection between the solutions of the linear systems we consider and the escape probabilities of random walks in the corresponding graphs, one might consider using Monte Carlo simulations to compute them.
However, for exponentially small escape probabilities (and entries of the solution), a naive implementation of this approach would require exponential time.
The repeated squaring approach discussed in \Cref{sec:repeated-square} can be seen as an efficient way of implementing Monte Carlo simulations.

\paragraph{Bit complexity and precision.}

There is a long history of developing more sophiscated
algorithmic primitives with better precision guarantees.
An example of more sophisticated algorithms to save precision is the Kahan summation algorithm \cite{K65,H93}.
There are many recent works that study the bit complexity of algorithms for linear algebraic primitives, such as diagonalization \cite{S23,BVKS23,BVS22a,BVS22b,DKRS23} and optimization \cite{GPV23,G23,ABGZ24,GLPSWY24}.
All of these works provide norm-wise error bounds.

\subsection{Technical Overview}
\label{sec:technical-overview}
Our results are divided into two settings: the entries of the input matrix are given in either fixed-point representation or floating-point representation. We first discuss our frameworks and techniques for the fixed-point setting, and then present our techniques for the floating-point setting.

\paragraph{\callalg{ThresholdDecay} framework.}
This framework begins by computing entrywise approximations for the large entries in the solution of the linear system. It then reduces the size of the system by substituting these approximations for the corresponding entries in the solution. The process is repeated: the framework computes approximations for the large entries in the solution of the smaller system and substitutes them to create an even smaller system. This iterative process continues until the linear system is fully solved.

The reason this framework works is based on the following three facts, which we prove in \Cref{sec:threshold-decay}:  
(1) the solution to a linear system contains at least one \emph{large} entry;  
(2) such a large entry can be approximately recovered from a norm-wise error-bounded approximate solution of the linear system;  
(3) removing the row and column corresponding to this entry by substituting its approximate value into the system does not significantly distort the other entries of the solution.  
Therefore, we can approximately recover the entries of the solution one by one, going from the largest to the smallest.

The \callalg{ThresholdDecay} framework formalizes this idea. In the \callalg{ThresholdDecay} framework, $\theta_t$ serves as a threshold in iteration $t$: all entries of the solution larger than or equal to $\theta_t$ are multiplicatively approximated. In addition, we show that $\theta_t$ decays exponentially, so only $O(n)$ iterations are required to recover all entries of the solution. In each iteration, the rows and columns corresponding to the recovered large entries are removed by substituting their approximate values into the system.

The result of \Cref{thm:entrywise-RDDL-solver} is obtained by using the almost-linear-time solver of \cite{CKPPRSV17} to compute the norm-wise error-bounded approximate solution in each iteration. For \Cref{thm:quadratic-sddm-inversion,thm:entrywise-SDDM-solver-fast}, we require an additional insight: if we set the values of the \emph{small} entries of the solution to zero, this does not significantly distort the \emph{large} entries of the solution. Therefore, if we knew which entries were small, we could exploit this to solve a smaller linear system and still recover the large entries accurately.

For \Cref{thm:quadratic-sddm-inversion}, this knowledge is obtained by observing that for SDDM matrices, if $\LL_{ij} \neq 0$, then $\LL^{-1} \ee^{(i)}$ is within a $\poly(nU)$ factor of $\LL^{-1} \ee^{(j)}$. Thus, given access to $\LL^{-1} \ee^{(i)}$, we can construct \emph{prediction intervals} for the entries of $\LL^{-1} \ee^{(j)}$. We use these prediction intervals in an iterative algorithm that computes $\LL^{-1} \ee^{(j)}$ by solving a series of smaller systems $\LL_{S_k,S_k} \xx = \ee^{(j)}_{S_k}$, where $S_k \subseteq [n]$ and each $\ell \in [n]$ appears in only a logarithmic number of $S_k$'s. Therefore, given $\LL^{-1} \ee^{(i)}$, an entrywise approximation to $\LL^{-1} \ee^{(j)}$ can be obtained with $\Otil(m \log^2 (U\epsilon^{-1} \delta^{-1}))$ bit operations. We can then use this to compute an approximate solution for another vector $\LL^{-1} \ee^{(p)}$ with $\LL_{jp} \neq 0$, and continue this process.

Finally, for \Cref{thm:entrywise-SDDM-solver-fast}, we consider a boundary of radius about $\sqrt{n}$ in the graph corresponding to $\LL$ for the entries that have already been computed. Entries corresponding to vertices outside this boundary are zeroed out, allowing us to solve smaller linear systems. We show that each vertex only remains in this boundary for at most $\sqrt{n}$ iterations until its approximate value is computed, and that the boundary is increased only $\sqrt{n}$ times—each time triggered by the presence of a large entry on the border of the current boundary. The details of this algorithm are presented in \Cref{sec:faster-sddm-solve}.

\paragraph{\callalg{ShortcutterFramework}.}

Let $G = ([n+1], E, w)$ be the graph associated with the $n \times n$ RDDL matrix $\LL$ (see \Cref{def:rddl-graph}). Define $\DD = \diag(\LL)$ and $\AA = \II - \DD^{-1} \LL$. Then, we have $\LL^{-1} = (\II - \AA)^{-1} \DD^{-1}$. Note that $(\II - \AA)^{-1} = \sum_{k=0}^\infty \AA^k$. The matrix $\AA$ represents the random walk matrix corresponding to $G$, and $(\II - \AA)^{-1}_{ij}$ equals the sum of probabilities of all random walks starting at vertex $i$, ending at vertex $j$, and avoiding vertex $n+1$, with the convention that a random walk of length zero has probability one. Thus, the entry $\LL^{-1}_{ij}$ can be computed by summing the probabilities of these random walks.

We leverage this perspective on $\LL^{-1}$ to explain the \callalg{ShortcutterFramework}. We denote by $W_{ij}^{(T)}$ the set of all random walks starting at vertex $i$, ending at vertex $j$, and avoiding a subset $T \subseteq [n+1]$. For brevity, we define $W_{ij}^{(n+1)} := W_{ij}^{(\{n+1\})}$.

The \callalg{ShortcutterFramework} selects a random subset $S$ of vertices, with each vertex included independently with probability $1/\sqrt{n}$, so the size of $S$ is approximately $\sqrt{n}$. Let $S = \{s_1, \ldots, s_k\}$. We partition the set of random walks $W_{ij}^{(n+1)}$ into $k+1$ parts: for each $\ell \in [k]$, $W_{ij}^{(n+1)}(s_\ell)$ denotes the walks that hit $s_\ell$ before hitting any other vertex in $S \setminus \{s_\ell\}$, and $W_{ij}^{(S \cup \{n+1\})}$ denotes the walks that avoid all vertices in $S$.

We show that the matrix corresponding to the probabilities of $W_{ij}^{(S \cup \{n+1\})}$ has condition number about $\exp(\sqrt{n})$ (see \Cref{thm:subcube-inverse-of-dense}). Moreover, each walk in $W_{ij}^{(n+1)}(s_\ell)$ can be decomposed into two parts: the segment up to the first visit to $s_\ell$ and the remainder of the walk. These insights imply that we can express the inverse of $\LL$ (up to multiplication by the diagonal matrix $\DD^{-1}$) as the sum of a rank-$|S|$ matrix (representing the probabilities of walks in $\bigcup_{\ell=1}^k W_{ij}^{(n+1)}(s_\ell)$) and a matrix with condition number about $\exp(\sqrt{n})$. Therefore, we only need to compute these two matrices to recover the inverse of $\LL$. Different strategies for computing them lead to different running times, as presented in \Cref{thm:subcube-inverse-of-dense,thm:subcube-inverse-of-sparse}.

Another insight used in analyzing our algorithms is the following: if a vertex $v$ is at distance much greater than $\sqrt{n}$ from vertex $u$ in the subgraph induced on $[n] \setminus S$, then any walk from $u$ to $v$ hits an element of $S$ with high probability (over the randomness of $S$). Using Markov's inequality, we show that such walks in $W_{uv}^{(S \cup \{n+1\})}$ can be ignored, incurring only a small multiplicative error in the computed matrix entry with high probability.

Moreover, if the distance from $u$ to $v$ is less than $\sqrt{n}$ in the induced subgraph, then the total probability of the walks in $W_{uv}^{(S \cup \{n+1\})}$ is at least $U^{-\sqrt{n}}$. Hence, we only need to recover the entries of the matrix corresponding to $W_{ij}^{(S \cup \{n+1\})}$ that are at least $U^{-\sqrt{n}}$. This allows us to work with an error parameter that is not excessively small for those computations.

\paragraph{Floating-point inputs.} 
The floating-point setting,
which we formalize at the end of Section~\ref{sec:prelim} is
essentially representing all edge weights in scientific notation:
a number of significant digits with an exponent.
Our main algorithmic result in this setting
is a cubic-time algorithm for matrix inversion (\Cref{thm:main-recursion}).
This running time is independent of the condition number / mixing time,
and depends only linearly on the number of significant bits plus the log of the desired accuracy.

Our algorithm is a recursive procedure that computes the inverse via a Schur complement approach (see \eqref{eq:sc-inversion}), halving the size of the matrix at each recursive step. Specifically, for block matrix
\[
\MM = \begin{bmatrix}
    \MM_{FF} & \MM_{FC} \\ \MM_{CF} & \MM_{CC}
\end{bmatrix},
\] 
let $\SS := \MM_{CC} - \MM_{CF} \MM_{FF}^{-1} \MM_{FC}$ denote the Schur complement of $\MM$ with respect to $C$, then the inverse is given by
\begin{align} \label{eq:sc}
\MM^{-1} = \begin{bmatrix}
    \MM_{FF}^{-1} + \MM_{FF}^{-1} \MM_{FC} \SS^{-1} \MM_{CF} \MM_{FF}^{-1} & -\MM_{FF}^{-1} \MM_{FC} \SS^{-1} \\
    - \SS^{-1} \MM_{CF} \MM_{FF}^{-1} & \SS^{-1}
\end{bmatrix}.
\end{align}
Our algorithm first recursively computes the inverse of $\MM_{FF}$ to obtain $\MM_{FF}^{-1}$, so $\SS$ can be computed; then we recursively invert $\SS$ to obtain $\SS^{-1}$, and use the above factorization to compute $\MM^{-1}$.

This blockwise approach is standard for solving linear systems; for example, it also appears in \cite{DDH07}. However, naively applying the approach does not work for entrywise guarantees, due to the following fact: for two invertible RDDL matrices $\MM$ and $\NN$, if $\MM \approxbar_{\eps} \NN$, we do not necessarily have $\MM^{-1} \approxbar_{\poly(n, \eps)} \NN^{-1}$. For example, consider matrices
\[
\MM=
\left[
\begin{matrix}
1 + \epsilon & -1\\
-1 & 1
\end{matrix}
\right],
\text{ and }
\NN=
\left[
\begin{matrix}
1 + \epsilon/M &  -1\\
-1 & 1
\end{matrix}
\right],
\]
for arbitrary $M \geq 1$.
Note that $\MM \approxbar_{\epsilon} \NN$, but 
\[
    \MM^{-1}=
    \left[
    \begin{matrix}
    1/\epsilon & 1/\epsilon\\
    1/\epsilon & 1 + 1/\epsilon
    \end{matrix}
    \right],
    \text{ and }
    \NN^{-1}=
    \left[
    \begin{matrix}
    M/\epsilon & M/\epsilon\\
    M/\epsilon & 1 + M/\epsilon
    \end{matrix}
    \right],
\]
Therefore, entries of $\MM^{-1}$ and $\NN^{-1}$ are bounded away from each other by a factor of $M$. In the standard approach, the error of $\MM_{FF}$ carries to $\SS$, so we only have an approximate $\widetilde{\SS} \approxbar \SS$. However, as the counterexample above shows, recursively inverting $\widetilde{\SS}$ does not give an entrywise approximation of $\SS^{-1}$.

Our contribution is the following lemma: for two invertible RDDL matrices $\MM$ and $\NN$, if $\MM \approxbar_{\epsilon} \NN$ and $\MM \boldsymbol{1} \approxbar_{\epsilon} \NN \boldsymbol{1}$, then $\MM^{-1} \approxbar_{2\epsilon n} \NN^{-1}$. We call the vector $\MM \boldsymbol{1}$ as the \emph{excess vector}, representing each row sum. We prove the lemma using the Kirchhoff's Matrix-Tree theorem, expressing each entry as a ratio of spanning tree counts (see \Cref{lemma:ApproxInvert}).

The main difference between our algorithm and a naive application of inversion using the Schur complement (Equation~\eqref{eq:sc}) is that we maintain the excess vector $\vv$ (defined as $-\MM \boldsymbol{1}$) throughout the recursion.
Interpreting the matrix $\LL$ as a graph, the absolute value of the vector $\vv$ can be seen as the weight of edges to a dummy vertex that absorbs the excess diagonal weight.

More specifically, under exact arithmetic, each entry $\vv_i$ can be computed as $\vv_i = -(\MM_{ii} + \sum_{j \in [n] \setminus \{i\}} \MM_{ij})$, where $\MM_{ii}$ is a positive number and $\MM_{ij}$'s are nonpositive. However, under floating-point arithmetic, this \emph{subtraction} can introduce significant error: the computed value of $\vv_i$ might not even be an entrywise approximation of the true difference. Therefore, we provide $\vv$ explicitly to the algorithm. We discuss the role of $\vv$ and its connection to our algorithm in more detail in \Cref{sec:float-inversion}.

\paragraph{Reduction from APSP problem.}
Our approach for proving \Cref{thm:RDDLInversionToAPSP} is as follows.
Given a graph $G = ([n], E, w)$ with edge weights in $[1, M]$, we encode the edge weights as the exponents of the entries of a diagonally dominant matrix:

\[
\MM_{ij} =
\begin{cases}
1 & \text{if } i = j, \\
0 & \text{if } (i, j) \notin E, \\
-2^{-n^2 M w_{ij}} & \text{if } (i, j) \in E.
\end{cases}
\]

Consider a random walk on the graph where the walk moves from vertex $i$ to $j$ with probability $-\MM_{ij}$. The excess probability in the diagonal entry corresponds to a transition to a dummy vertex $n+1$.
We observe that $\MM^{-1}_{ij}$ equals the total probability of all random walks from $i$ to $j$. Moreover, for a walk from $i$ to $j$ consisting of edges $(a_1, b_1), \ldots, (a_k, b_k)$, the probability of the walk is:

\[
2^{-n^2 M \sum_{\ell=1}^k w_{a_\ell b_\ell}}.
\]

We show that this value is only significant when the walk corresponds to the shortest path from $i$ to $j$. Therefore, the length of the shortest path can be recovered from the exponent of $\MM^{-1}_{ij}$. Consequently, the lengths of all-pairs shortest paths can be extracted from the exponents of the entries of $\MM^{-1}$.

\section{Preliminaries}
\label{sec:prelim}

\paragraph{Notations.}

We denote $\{1,\ldots,n\}$ by $[n]$. For a graph $G=(V,E,w)$ and $v\in V$, let $N^{-}(v), N^+(v) \subseteq V$ denote the set of incoming and outgoing neighbors of $v$, respectively. 
If the graph is undirected, we denote the neighbors by $N(v)$. For $u,v\in V$, $d^G(u,v)$ denotes the distance from $u$ to $v$ in terms of number edges in graph $G$. We denote the weighted distance with $d^G_{w}(u,v)$. When the graph is understood from the context, we omit $G$ in the notation and use $d(u,v)$ and $d_w(u,v)$. For a subset of vertices $S\subseteq V$, we denote the induced subgraph on $S$ with $G(S)$. 

We use bold lowercase letters and bold uppercase letters to denote vectors and matrices, respectively. 
The vector of all ones is denoted by $\boldsymbol{1}$, and the $i$-th standard basis vector is denoted by $\ee^{(i)}$. 
We do not explicitly specify the sizes of these vectors, as they will be clear from the context throughout the paper. 
For a matrix $\MM$, we denote its $(i,j)$-th entry by $\MM_{ij}$. 
Let $S$ and $T$ be subsets of the row and column indices of $\MM$, respectively. 
We denote the submatrix obtained by taking the entries $\{(i,j) : i\in S, j \in T\}$ by $\MM_{ST}$. 
If $S$ contains all rows, we denote this by $\MM_{:T}$; if $S=\{i\}$, we use $\MM_{iT}$. 
Similarly, we use the notation $\MM_{S:}$ and $\MM_{Sj}$ for columns. 

For $\vv \in \R^{n}$, let $\diag(\vv)\in\R^{n\times n}$ be the diagonal matrix with the entries of $\vv$ on its diagonal. 
For $\MM\in\R^{n\times n}$, $\diag(\MM)\in\R^{n\times n}$ denotes the diagonal matrix obtained by taking the diagonal of $\MM$. 
The Moore-Penrose pseudo-inverse of $\MM$ is denoted by $\MM^{\dagger}$. 

\paragraph{RDDL and SDDM matrices.}

In this paper, we study the following structured matrices.

\begin{definition}[RDDL and SDDM Matrices]
\label{def:rddl}
    A matrix $\MM\in \R^{n\times n}$ is an $L$-matrix if for all $i\neq j$, $\MM_{ij} \leq 0$, and for all $i\in [n]$, $\MM_{ii} > 0$. 
    $\MM$ is row diagonally dominant (RDD) if for all $i\in[n]$, $|\MM_{ii}| \geq \sum_{j\in [n] \setminus\{i\}} |\MM_{ij}|$. 
    $\MM$ is RDDL if it is both an $L$-matrix and RDD. 
    An SDDM matrix is an invertible \emph{symmetric} RDDL matrix.
\end{definition}

We can associate a graph to an RDDL matrix by adding a dummy vertex as follows. In general, the associated graph is directed, but if the matrix is SDDM, the graph is undirected.

\begin{definition}[Associated Graph of an RDDL Matrix]
\label{def:rddl-graph}
    For an RDDL matrix $\MM \in \R^{n \times n}$, we define the associated weighted graph $G = ([n+1], E, w)$ with an additional dummy vertex $n+1$. 
    For any $i,j \in [n]$ with $i\neq j$, we set the edge weight $w(i,j)$ in the graph to $-\MM_{ij}$. 
    Moreover, we set the edge weight $w(i, n+1)$ to $\sum_{j \in [n]} \MM_{ij}$. 
    Finally, for all $i\in[n]$,
    if the matrix is RDDL, we set $w(n+1,i)=0$, and if the matrix is SDDM,
    we set $w(n+1,i)=w(i,n+1)$.
    
    If for all vertices $i\in[n]$, there is a path in $G$ from $i$ to $n+1$ with all positive edge weights, then we say $G$ (the graph associated to the RDDL matrix) is \emph{connected to the dummy vertex}.
\end{definition}

RDDL matrices are not necessarily invertible since the Laplacian matrix of a graph is RDDL and not invertible. The following lemma characterizes invertible RDDL matrices (proved in \Cref{sec:probabilities}).

\begin{restatable}{lemma}{invertibleRDDL}
\label{lemma:rddl-invertible}
An RDDL matrix $\MM \in \R^{n \times n}$ is invertible if and only if its associated graph is connected to the dummy vertex.
\end{restatable}

The following lemma bounds the smallest eigenvalue (and hence the condition number) of an SDDM matrix.

\begin{lemma}[\cite{ST04:journal}]
\label{lemma:smallest-eig-sddm}
    Let $G$ be an undirected connected weighted graph and let $\LL$ be either the Laplacian matrix of $G$ or a principal square submatrix of the Laplacian. Then the smallest nonzero eigenvalue of $\LL$ is at least $\min(8w/n^2, w/n)$, where $w$ is the least weight of an edge of $G$ and $n$ is the dimension of $\LL$.
\end{lemma}

This immediately gives the following bound for the SDDM matrices we are interested in by observing that any SDDM matrix is the principal square submatrix of a Laplacian.

\begin{corollary} \label{cor:sddm-inverse-bound}
    Let $\LL$ be an $n \times n$ invertible SDDM matrix with integer entries in $[-U, U]$, then
    \begin{align*}
        \| \LL^{-1} \|_2 = \frac{1}{\lambda_{\min}(\LL)} \le n^2.
    \end{align*}
\end{corollary}

In \Cref{lemma:RDDL-escape}, we establish connections between the entries of the inverse of invertible RDDL matrices and quantities corresponding to the random walks on their associated graphs.

\paragraph{Random walks.} 

Given a directed weighted graph $G = ([n], E, w)$ with $w \in \R_{\geq 0}^{E}$, a random walk in the graph picks the next step independently of all previous steps. If the random walk is at vertex $i$, in the next step it moves to $j \in N^+(i)$ with probability $\AA_{ij} := \frac{w(i,j)}{\sum_{k \in N^+(i)} w(i,k)}$. 
In other words, we consider the Markov chain associated with the graph. 
Note that for the special case of unweighted graphs, the probability for each neighbor is $1/|N^+(i)|$. 
The matrix $\AA$ is the transition probability matrix of the random walk on $G$.

Let $G$ be the graph associated with the RDDL matrix $\MM \in \R^{n \times n}$. Moreover, let $\widetilde{\AA} \in \R^{(n+1) \times (n+1)}$ be the transition probability matrix associated with the random walk on $G$. Let $\widetilde{\MM} \in \R^{(n+1) \times (n+1)}$ be the Laplacian matrix of $G$. 

If $\MM$ is an SDDM matrix, then $G$ is undirected, and $\widetilde{\MM}$ is an undirected Laplacian. Otherwise, $G$ is directed, and $\widetilde{\MM}$ is a directed Laplacian with rows summing to zero. In the latter case, row $n+1$ of $\widetilde{\MM}$ is zero.

We have $\MM = \widetilde{\MM}_{[n][n]}$ and 
\[
\widetilde{\AA} = \diag(\widetilde{\MM})^{\dagger} (\diag(\widetilde{\MM}) - \widetilde{\MM}).
\] 
We use the pseudo-inverse for $\diag(\widetilde{\MM})$ because, in the non-symmetric case, its last row is zero. Let $\AA = \widetilde{\AA}_{[n][n]}$. Then 
\[
\II - \AA = \diag(\MM)^{-1} \MM.
\] 
Therefore, if $\MM$ is invertible, $\II - \AA$ is also invertible. In this case, the spectral radius of $\AA$ is less than one, and the inverse is given by
\[
(\II - \AA)^{-1} = \II + \AA + \AA^2 + \cdots.
\]
Since $\MM^{-1} = (\II - \AA)^{-1} \diag(\MM)^{-1}$, we can compute the inverse of $\MM$ efficiently by finding the inverse of $\II - \AA$. In general, $\II - \AA$ is an RDDL matrix and is not symmetric, even if $\MM$ is symmetric.

The inverse of $\II - \AA$ is closely related to the probabilities of random walks in the associated graph of $\MM$. Therefore, in the analysis of our algorithms, we rely heavily on these probabilities, which provide more intuitive proofs. 

The $(i,j)$ entry of $\AA^k$ is the sum of the probabilities of all random walks of length $k$ that start at vertex $i$, end at $j$, and avoid $n+1$ in the associated graph. Consequently, $(\II - \AA)^{-1}_{ij}$ is the sum of the probabilities of all random walks starting at $i$, ending at $j$, and avoiding $n+1$. Note that $(\II - \AA)^{-1}_{ij}$ is not necessarily less than one because it represents the \emph{sum of probabilities of random walks}, not the probability of a single random walk. 

Escape probability, on the other hand, is associated with an event and is therefore at most one.

\begin{definition}[Escape Probability]
    The escape probability $\P(s,t,p)$ denotes the probability of hitting vertex $t$ before vertex $p$ in a random walk starting at $s$.
\end{definition}

The relationship between the inverse of an RDDL matrix and the escape probabilities in its associated graph is stated as follows. The lemma is proved in \Cref{sec:probabilities}.

\begin{restatable}{lemma}{inverseRDDLentries}
\label{lemma:RDDL-escape}
For an invertible RDDL matrix $\MM \in \mathbb{R}^{n \times n}$, let $G = ([n+1], E, w)$ be the associated graph and $\AA = \diag(\MM)^{-1} (\diag(\MM) - \MM)$. 
Then, for $s, t \in [n]$, 
\[
\mathbb{P}(s, t, n+1) = \frac{(\MM^{-1})_{st}}{\MM_{tt}(\MM^{-1})_{tt}} = \frac{(\II-\AA)^{-1}_{st}}{(\II-\AA)^{-1}_{tt}}.
\]
Furthermore, we have $\MM^{-1}_{tt} > 0$ for every $t \in [n]$.
\end{restatable}

Note that there is symmetry associated with escape probability: $\P(s,t,p) = 1 - \P(s,p,t)$. However, due to floating-point issues (discussed later in this section), it is not advisable to compute $\P(s,t,p)$ from $\P(s,p,t)$ in this way. This involves subtracting a positive number from one, which can introduce uncontrollable multiplicative errors.

Absorption probabilities extend the concept of escape probability to subsets of more than two vertices. This concept is closely related to ``shortcutter vertices'', which we use in our algorithms and analysis for subcubic time inversion of RDDL matrices in \Cref{sec:subcube-dense}.

\begin{definition}[Absorption Probabilities] \label{def:absorption}
Let $T$ be a subset of vertices and $s$ be a vertex. The absorption probabilities $\P_{\textrm{abs}}(s,T) \in [0,1]^{|T|}$ form a vector where the $t$-th entry denotes the probability of hitting $t \in T$ before hitting any vertex in $T \setminus \{t\}$ in all possible random walks starting at $s$. For a subset of vertices $S$, we define $\P_{\textrm{abs}}(S,T) \in [0,1]^{|S| \times |T|}$ as a matrix with $(\P_{\textrm{abs}}(S,T))_{i_s:} = (\P_{\textrm{abs}}(s,T))^\top$ for all $s \in S$, where $i_s$ is the row index associated with $s$.
\end{definition}

It is easy to see that the absorption probabilities $\P_{\textrm{abs}}(s,\{t,p\}) \in [0,1]^2$ yield the escape probabilities $\P(s,t,p)$ and $\P(s,p,t)$.

\paragraph{Norms.}
The $\ell_1$-norm and $\ell_{\infty}$-norm of a vector $\vv \in \R^{n}$ are defined as $\norm{\vv}_1 := \sum_{i \in [n]} |\vv_i|$ and $\norm{\vv}_{\infty} := \max_{i \in [n]} |\vv_i|$, respectively. The induced $\ell_1$ and $\ell_{\infty}$ norms for a matrix $\MM \in \R^{n \times n}$ are defined as:
\[
\norm{\MM}_1 := \max_{\vv \in \R^n: \norm{\vv}_1 = 1} \norm{\MM \vv}_1, \quad \text{and} \quad \norm{\MM}_{\infty} := \max_{\vv \in \R^n: \norm{\vv}_{\infty} = 1} \norm{\MM \vv}_{\infty}.
\]
Both vector norms and induced norms satisfy the triangle inequality and consistency, e.g., $\norm{\AA \BB}_{\infty} \leq \norm{\AA}_{\infty} \cdot \norm{\BB}_{\infty}$. The spectral radius of a matrix $\MM$ is defined as $\rho(\MM) := \max\{|\lambda_1|, \ldots, |\lambda_n|\}$, where $\lambda_i$ are the eigenvalues of $\MM$. The condition number of $\MM$ is $\kappa(\MM) := \norm{\MM}_2 \cdot \norm{\MM^\dagger}_2$. When the context is clear, we denote the spectral radius and condition number by $\rho$ and $\kappa$, respectively.

\paragraph{Schur Complement.}
Given a block matrix
\[
\MM = \begin{bmatrix}
    \AA & \BB \\ \CC & \DD
\end{bmatrix},
\]
let $F$ denote the leading indices and $C$ the remaining indices. Then $\MM_{FF} = \AA$ and $\MM_{CC} = \DD$. The Schur complement of $\MM$ with respect to $C$ is $\sc(\MM,C) = \DD - \CC \AA^{-1} \BB$. If $\AA$ and $\SS := \sc(\MM,C)$ are invertible, then $\MM$ is invertible, and its inverse is given by:
\begin{align}
\label{eq:sc-inversion}
\MM^{-1} = \begin{bmatrix}
    \AA^{-1} + \AA^{-1} \BB \SS^{-1} \CC \AA^{-1} & -\AA^{-1} \BB \SS^{-1} \\
    - \SS^{-1} \CC \AA^{-1} & \SS^{-1}
\end{bmatrix}.
\end{align}

\paragraph{Asymptotic and Approximate Notation.}
We use $\Otil$ notation to suppress polylogarithmic factors in $n$ and $L$ and polylog-logarithmic factors in $U\kappa \epsilon^{-1} \delta^{-1}$. Formally, $\Otil(f) = O(f \cdot \log(nL \cdot \log(U\kappa \epsilon^{-1} \delta^{-1})))$.  

We also rely on the following notion of approximation for scalars and matrices.

\begin{definition}[Entrywise Approximation]
For two nonnegative scalars $a$ and $b$, and $\epsilon \geq 0$, we write $a \approx_{\epsilon} b$ if they are $\exp(\epsilon)$-approximation of each other,
\[
e^{-\epsilon} \cdot a \leq b \leq e^{\epsilon} \cdot a.
\]
We also write $a \approx_{O(\epsilon)} b$ if there exists a constant $c > 0$ such that $a \approx_{c\epsilon} b$. For two matrices $\AA$ and $\BB$ of the same size, we write
\[
\AA \approxbar_{\epsilon} \BB
\]
if $\AA_{ij} \approx_{\epsilon} \BB_{ij}$ for all $1 \leq i, j \leq n$. Considering a vector as a matrix with one column, the definition of the entry-wise approximation extends naturally to vectors $\uu$ and $\vv$ of the same size. 
\end{definition}

\begin{fact} \label{fact:approx}
    For nonnegative scalars $a, b, c$, and $d$, the following holds.
    \begin{itemize}
        \item If $a \approx_{\epsilon_1} b$ and $b \approx_{\epsilon_2} c$, then $a \approx_{\epsilon_1 + \epsilon_2} c$.
        \item If $a \approx_{\epsilon} c$ and $b \approx_{\epsilon} d$, then $a + b \approx_{\epsilon} c + d$.
        \item For $0 < \epsilon < 1/2$, if $(1-\epsilon) a \le b \le (1+\epsilon) a$, then $a \approx_{1.5\epsilon} b$.
    \end{itemize}
    The above properties generalize to vectors and matrices for entrywise approximations.
\end{fact}

In this paper, we focus on entrywise approximation for the inversion of RDDL matrices. Specifically, for an RDDL matrix $\MM$, we compute a matrix $\ZZ$ such that $\ZZ \approxbar_{\epsilon} \MM^{-1}$.

\paragraph{Floating-point numbers.} A floating point number in base $t$ is stored using two integer scalars $a$ and $b$ as $a \cdot t^{b}$. 
Scalar $a$ is the significand and scalar $b$ is the exponent. 
The most usual choice of the base on real computers is $t=2$. 
The number of bits of $a$ determines the bits of precision and the number of bits of $b$ determines how large or small our numbers can be. 
Let $L$ be the number of bits that we allow for $a$ and $b$, then for any number $c \in [t^{- 2^{L}+1},(2^{L} - 1) \cdot t^{2^{L} - 1}]$, there is a number $d$ in floating-point numbers with $L$ bits such that $d \leq c \leq (1 + t^{- 2^{L-1}}) d$. 
Therefore if our numbers are not too large or too small, we can approximate them using $O(\log(1/\epsilon))$ bits to $\exp(\epsilon)$ multiplicative error. 
Also, note that we can use floating-point numbers to show \emph{zero} exactly. 
Using $\Otil(\log(1/\epsilon))$ bit operations, by fast Fourier transform (FFT), we can perform the multiplication of two numbers to $e^{\epsilon}$ accuracy. Moreover with $O(\log(1/\epsilon))$ bit operations, we can perform an addition of two nonnegative numbers (or two nonpositive numbers) with a multiplicative error of $e^{\epsilon}$. 
However, if we add a positive number to a negative number, the error will be additive and depends on $\epsilon$---see Kahan's summation algorithm \cite{K65}. 
That is why we do not perform the addition of positive and negative numbers in our algorithms in this paper.

\section{Quadratic and Subquadratic Time Algorithms for Sparse SDDMs}
\label{sec:quadratic-sparse}

In this section, we mainly focus on developing algorithms for sparse SDDM matrices, which correspond to undirected graphs. We present a framework called \callalg{ThresholdDecay} for iteratively solving linear systems with diagonally dominant $L$-matrices. The \callalg{ThresholdDecay} framework begins by computing entrywise approximations for the large entries in the solution of the linear system. It then reduces the size of the system by substituting these approximations for the corresponding entries in the solution. The process is repeated: the framework computes approximations for the large entries in the solution of the smaller system and substitutes them to create an even smaller system. This iterative process continues until the linear system is fully solved.

We use this framework to develop three algorithms in this section by implementing steps of the framework with different approaches. More specifically given $\epsilon,\delta \in (0,1)$ for $n$-by-$n$ matrices with $m$ nonzero integer entries in $[-U,U]$ and condition number $\kappa$, we compute, with probability at least $1-\delta$, entrywise $\exp(\epsilon)$-approximations using the following number of iterations.
\begin{itemize}
    \item Solving a linear system with an RDDL matrix in $\Otil((m n^{1+o(1)}) \log^{O(1)}(U\kappa\epsilon^{-1} \delta^{-1}))$ bit operations (\Cref{thm:entrywise-RDDL-solver}).
    \item Inverting an SDDM matrix in $\Otil( mn \log^2 (U\eps^{-1} \delta^{-1}))$ bit operations (\Cref{thm:quadratic-sddm-inversion}).
    \item Solving a linear system $\LL \xx = \ee^{(i)}$ with an SDDM matrix in $\Otil( m \sqrt{n} \log^2 (U \eps^{-1} \delta^{-1}))$ bit operations (\Cref{thm:entrywise-SDDM-solver-fast}).
\end{itemize}

We start by stating running times of algorithms for solving SDDM and RDDL matrices from prior work that we use as subprocedures in our algorithms.

\begin{lemma}[Thereom 5.5 of \cite{ST04:journal}]
\label{lemma:near-linear-solver}
Let $\LL\in \R^{n\times n}$ be an SDDM matrix with $m$ nonzero entries
and $\bb \in \R^{n}$, all with integer entries in $[-U, U]$. Let $\epsilon > 0$ and $\delta\in(0,1)$, then there is an algorithm that with probability $1-\delta$ and $\Otil(m \log(U \epsilon^{-1} \delta^{-1}) \log(\epsilon^{-1}))$ bit operations computes vector $\widetilde{\xx}$ such that
\[
\norm{\widetilde{\xx} - \LL^{\dagger} \bb}_{\LL}
\leq
\epsilon \cdot \norm{\LL^{\dagger}\bb}_{\LL}.
\]
Moreover entries of vector $\widetilde{\xx}$ are represented with $O(n\log(U \epsilon^{-1}))$ bits in fixed point arithmetic.
\end{lemma}

The following immediately follows from \Cref{lemma:near-linear-solver}

\begin{corollary} \label{cor:near-linear-solver-SDDM}
    Let $\LL\in \R^{n\times n}$ be an invertible SDDM matrix with $m$ nonzero entries
    and $\bb \in \R^{n}$, all with integer entries in $[-U, U]$. Let $\epsilon > 0$ and $\delta \in (0,1)$, Then there is an algorithm that with probability $1-\delta$ and $\Otil(m \log^2(U \epsilon^{-1} \delta^{-1}))$ bit operations computes vector $\widetilde{\xx}$ such that
    \[
    \norm{\widetilde{\xx} - \LL^{\dagger} \bb}_{2}
    \leq
    \epsilon \cdot \norm{\bb}_{2}.
    \]
    Moreover entries of vector $\widetilde{\xx}$ are represented with $O(\log(n U \epsilon^{-1}))$ bits in fixed point arithmetic.
\end{corollary}
\begin{proof}
We apply \Cref{lemma:near-linear-solver} with $\widetilde{\epsilon} = \frac{\epsilon}{n^{1.5} U^{0.5}}$ to obtain $\widetilde{\xx}$ with $\Otil(m \log(U \widetilde{\epsilon}^{-1} \delta^{-1})\log( \widetilde{\epsilon}^{-1}))$ bit operations such that entries of $\widetilde{\xx}$ are represented with $O(\log(nU\epsilon^{-1}))$ bits and
\[
\norm{\widetilde{\xx} - \LL^{\dagger} \bb}_{\LL}
\leq
\frac{\epsilon}{n^{1.5} U^{0.5}} \cdot \norm{\LL^{\dagger}\bb}_{\LL}.
\]
Since $\LL$ is an SDDM matrix, it is a principal square submatrix of a Laplacian with integer entries in $[-U, U]$. Therefore, since $\LL$ is invertible, by \Cref{lemma:smallest-eig-sddm}, the smallest eigenvalue of $\LL$ is at least $1/n^2$ and
\[
\norm{\widetilde{\xx} - \LL^{\dagger} \bb}_{2} \leq n \cdot \norm{\widetilde{\xx} - \LL^{\dagger} \bb}_{\LL}
\leq
\frac{\epsilon}{n^{0.5} U^{0.5}} \cdot \norm{\LL^{\dagger}\bb}_{\LL} \leq \epsilon \cdot \norm{\bb}_{2}.
\]
\end{proof}

The following lemma provides an almost-linear-time algorithm for solving linear systems with RDDL matrices.

\begin{lemma}[Theorem 4.1 of \cite{CKPPRSV17}] 
\label{lemma:near-linear-solver-RDDL}
There exists an algorithm that for $\epsilon,\delta \in (0,1)$ and any invertible RDDL matrix $\MM$ with integer entries in $[-U,U]$ and vector $\bb \in [0,U]^n$ computes $\widehat{\xx} \in \R^{n}$, with probability at least $1-\delta$, such that
\[
\norm{\widehat{\xx} - \LL^{-1} \bb}_2 \leq \epsilon \cdot \norm{\bb}_2.
\]
with $\Otil((m + n^{1+o(1)}) \log^{O(1)}(U\kappa\epsilon^{-1} \delta^{-1}))$ bit operations. Moreover entries of $\widehat{\xx}$ are represented with $O(U\kappa\epsilon^{-1})$ bits.
\end{lemma}

\subsection{The Threshold Decay Framework}
\label{sec:threshold-decay}

In this section, we formalize the idea of extracting large entries repeatedly in our \callalg{ThresholdDecay} framework, which serves as a foundation for our results in this section for solving linear systems with SDDM and RDDL matrices and inverting SDDM matrices. In the \callalg{ThresholdDecay} framework, we have a threshold that decays exponentially over the iterations. We show that entries of the solution larger than the threshold are multiplicatively approximated in each iteration. This allows us to substitute our approximate values for the entries and tend to the smaller linear system in the next iteration.

To show that the solution to the smaller system will generate valid results, we need to argue that substituting the approximate values for a subset of entries does not distort the value of the solution in the smaller system significantly. This is proven in the following lemma.

\begin{lemma}
    \label{lemma:remove_large}
    For an invertible RDDL matrix $\MM \in \R^{n\times n}$ and a nonnegative vector $\bb$
    with exact solution $\xxbar := \MM^{-1} \bb$, if we are given a subset $F \subset [n]$
    and a vector $\xxtil$ defined on $F$ with
    $\xxtil \approxbar_{\theta} \xxbar_{F}$,
    then the solution $\xxhat$ defined on $S:=[n] \setminus F$ to the following system 
    \begin{align}
    \label{eq:partial-solve}
    \MM_{S, S} \xxhat
    =
    \bb_{S} - \MM_{S, F} \xxtil
    \end{align}
    satisfies $\xxhat \approxbar_{\theta} \xxbar_{S}$,
    or equivalently $[\xxtil; \xxhat] \approxbar_{\theta} \xxbar$.
\end{lemma}
\begin{proof}
Note that $\bb_S$ and $\xxtil$ are nonnegative vectors. Since $\MM_{S,F}$ is an off-diagonal block of $\MM$, all of its entries are nonpositive. Therefore
\[
\bb_{S} - \MM_{S, F} \xxtil \approxbar_{\theta} \bb_{S} - \MM_{S, F} \xxbar_F.
\]
By \Cref{lemma:entries-of-inverse}, $(\MM_{S,S})^{-1} \geq 0$. Therefore
\[
\widehat{\xx} = (\MM_{S,S})^{-1} (\bb_{S} - \MM_{S, F} \xxtil) \approxbar_{\theta} (\MM_{S,S})^{-1} (\bb_{S} - \MM_{S, F} \xxbar_F) = (\MM_{S,S})^{-1} \MM_{S,S} \xxbar_S = \xxbar_S,
\]
where the second equality follows from $\bb_S = \MM_{S,S} \xxbar_S + \MM_{S,F} \xxbar_F$.
\end{proof}

The following lemma states that for a sufficiently small error parameter, a norm-wise error-bounded approximate solution to the linear system provides an entrywise approximation for the large entries of the solution.

\begin{lemma}
\label{lemma:large_entries_guarantee}
For an invertible RDDL matrix $\MM$ and a nonnegative vector $\bb$ with exact solution $\xxbar := \MM^{-1} \bb$, for $0<\epsilon < \beta/2$, let $\xxtil$ be an approximate  solution with $\| \xxtil - \xxbar \|_2 \le \epsilon  \|\bb\|_2$.
Let $F := \{ i \in [n] : \xxbar_i \ge \beta \norm{\bb}_1 \}$ denote the set of large entries of the exact solution. Then for all $i \in F$, $\xxbar_i$ is approximated by $\xxtil_i$ as the following
\[
\frac{\beta-\epsilon}{\beta} \xxbar_i \leq \xxtil_i \leq \frac{\beta+\epsilon}{\beta} \xxbar_i.
\]
Moreover, $\xxtil_F \approxbar_{3\epsilon / (2\beta)} \xxbar_F$.
\end{lemma}

\begin{proof}
We have
\[
\abs{(\xxtil - \xxbar)_i}
\leq 
\norm{\xxtil - \xxbar}_2
\leq
\epsilon \cdot \norm{\bb}_2 \leq \epsilon \cdot \norm{\bb}_1
\]
Since $\xxbar_i \geq \beta \norm{\bb}_1$, we have $-\epsilon \norm{\bb}_1 \geq -\frac{\epsilon}{\beta} \xxbar_i$, and thus by triangle inequality
\[
\xxtil_i \geq \xxbar_i - \abs{(\xxbar - \xxtil)_i} \geq \xxbar_i - \epsilon \cdot \norm{\bb}_1 \geq \frac{\beta - \epsilon}{\beta} \xxbar_i.
\]
Moreover, by triangle inequality
\[
\xxtil_i \leq \xxbar_i + \abs{(\xxtil - \xxbar)_i} \leq  \xxbar_i + \epsilon \cdot \norm{\bb}_1 \leq \frac{\beta + \epsilon}{\beta} \xxbar_i.
\]
Finally $\xxtil_F \approxbar_{3\epsilon / (2\beta)} \xxbar_F$ follows by \Cref{fact:approx} since $\epsilon / \beta < 1/2$ .
\end{proof}

The following lemma states that the solution to the linear system has at least one large entry. This fact can be used in combination with \Cref{lemma:large_entries_guarantee} to show that recovering at least one entry of the solution requires only a small number of bits of precision—that is, it does not require a very small norm-wise error parameter.

\begin{lemma}
\label{lemma:at-least-on-large-entry}
Let $\MM \in \R^{n\times n}$ be an invertible RDDL matrix with integer entries in $[-U,U]$ and $\bb$ be a nonnegative vector.
There exists $i\in[n]$ with $(\MM^{-1} \bb)_i \geq \frac{1}{n U}\norm{\bb}_1$.
\end{lemma}
\begin{proof}
Let $\AA=\II-\diag(\MM)^{-1}\MM$. Then we have
\[
\MM^{-1} = \left( \sum_{k=0}^{\infty} \AA^{k} \right) \diag(\MM)^{-1}.
\]
Therefore since $\AA^{0}=\II$, $\AA^{k}\geq 0$ for all $k$, and 
the diagonal entries of $\MM$ are at most $U$, for all $j\in[n]$, we have $(\MM^{-1})_{jj} \geq \frac{1}{U}$ 
and $\MM^{-1} \geq 0$. 
Therefore
\[
\sum_{j=1}^n (\MM^{-1} \bb)_j =
\sum_{j=1}^{n} \sum_{k=1}^n (\MM^{-1})_{jk} \bb_k 
\geq \sum_{k=1}^n \bb_k (\MM^{-1})_{kk} \geq \frac{1}{ U} \norm{\bb}_1.
\]
The result follows by taking $i$ with the largest $(\MM^{-1} \bb)_i$ value.
\end{proof}

\Cref{lemma:remove_large,lemma:large_entries_guarantee,lemma:at-least-on-large-entry} naturally suggest an algorithmic approach as follows. By \Cref{lemma:at-least-on-large-entry}, the solution to the linear system contains at least one large entry. Then, by \Cref{lemma:large_entries_guarantee}, we can recover an entrywise approximation to such a large entry using a norm-wise error-bounded approximate solution. Finally, \Cref{lemma:remove_large} allows us to remove this entry and its corresponding row and column without significantly distorting the other entries of the solution. We repeat this process until all entries are recovered.

The \callalg{ThresholdDecay} framework, which we present next, formalizes this idea. In the \callalg{ThresholdDecay} framework, $\theta_t$ serves as a threshold: all entries of the solution larger than or equal to $\theta_t$ are multiplicatively approximated. We will show that $\theta_t$ decays exponentially, so only $O(n)$ iterations are required to recover all entries of the solution.

In each iteration, a norm-wise bounded error approximate solution $\widehat{\xx}^{(t)}$ is used to obtain multiplicative approximations for the large entries $F^{(t)}$ among the remaining ones. We then update the remaining index set via $S^{(t+1)} = S^{(t)} \setminus F^{(t)}$ and solve the following reduced system in the next iteration:
\begin{align*}
    \LL_{S^{(t+1)}, S^{(t+1)}} \xx = \widehat{\bb}^{(t+1)},
\end{align*}
where $\widehat{\bb}^{(t+1)}$ is computed similarly to the right-hand side of Equation~\eqref{eq:partial-solve}.

Remarkably, the computation steps in \callalg{ThresholdDecay} can be expressed purely in terms of \emph{approximation assumptions}. In particular, Steps~\ref{ThresDec:L-solver} and~\ref{ThresDec:updateB} produce vectors that approximate specific arithmetic expressions, capturing the errors introduced by Laplacian solvers and finite-precision arithmetic.

Within the framework, we do not require any specific method for generating these vectors—we only prove correctness guarantees under the approximation assumptions. Later, we will plug in efficient implementations for Steps~\ref{ThresDec:L-solver} and~\ref{ThresDec:updateB} tailored to different problem settings and analyze their bit complexity.

Note that Step~\ref{ThresDec:extractF} is not part of the approximation assumptions, since it only involves comparing entries of a vector to a power of two—a task that can be performed exactly in any binary-based numerical system.

The following lemma provides guarantees on the vectors produced by the \callalg{ThresholdDecay} framework.

\begin{figure}[t]
    \begin{algbox}
        \textbf{\uline{\textsc{ThresholdDecay}}}: The algorithm framework for our solvers. Implementations of Step \ref{ThresDec:L-solver} and \ref{ThresDec:updateB} are not specified and will be plugged in later. Only correctness is proved for the framework. 
        \\
        \uline{Input}: $\LL$: an $n \times n$ invertible RDDL matrix with integer entries in $[-U, U]$, \\
            $\bb$: a nonnegative vector with integer entries in $[0,U]$,
            \\
            $\epsilon$: the target accuracy, \\
            $\kappa$: an upper bound on the condition number of $\LL$, \\
            $T$: number of iterations. \\
            Assume that $U \ge 2, T \ge 10, \eps \in (0, 1)$ \\
        \uline{Output:} vector $\xxtil$ supported on some set $A \subseteq [n]$ such that $\xxtil_A \approxbar_{\epsilon} (\LL^{-1} \bb)_A$.
        
        \begin{enumerate}
            \item Initialize $S^{(0)} = [n]$,
                $\bbhat^{(0)} = \bb$, and $\widetilde{\xx}^{(0)}$ to a zero dimensional vector. Let $\eps_L := \frac{\epsilon}{64T(nU)^2}$.
            \item For $t = 0 \ldots T$:
            \begin{enumerate}
                \item \label{ThresDec:L-solver}
                    Compute 
                    $\xxhat^{(t)}$ defined on $S^{(t)}$ such that
                    $ \norm{\widehat{\xx}^{(t)} - \LL_{S^{(t)}, S^{(t)}}^{-1} \widehat{\bb}^{(t)}}_2 \leq
                    \eps_L
                    \norm{\widehat{\bb}^{(t)}}_2$.
                \item \label{ThresDec:set-threshold}
                    Let $\theta_{t} $ be the smallest power of two larger than $  \frac{1}{4(nU)^2} \norm{\widehat{\bb}^{(t)}}_1$.
                \item \label{ThresDec:extractF}
                    Let $F^{(t)} \subseteq S^{(t)}$ be the indices of all the entries in $\xxhat^{(t)}$ larger than or equal to $\theta_{t}$, and let $S^{(t + 1)} \leftarrow S^{(t)} \setminus F^{(t)}$.
                \item \label{ThresDec:updateX}
                    Let $\xxtil^{(t+1)}$ be defined on $[n]\setminus S^{(t+1)}$ with $[\xxtil^{(t+1)}_{[n] \setminus S^{(t)}};\xxtil^{(t+1)}_{F^{(t)}}] \leftarrow [\xxtil^{(t)}; \xxhat^{(t)}_{F^{(t)}}]$.
                \item \label{ThresDec:updateB}
                    Create vector $\bbhat^{\left(t + 1\right)}$ for the next iteration such that
                    \begin{align} \label{eq:iteration-b}
                        \bbhat^{\left(t + 1\right)}
                        \approxbar_{\epsilon/(8T)}
                        \bb_{S^{\left( t +1 \right)}}
                        -
                        \LL_{S^{\left( t +1 \right)}, [n] \setminus S^{(t+1)}}
                        \xxtil^{\left( t+1 \right)}
                    \end{align}
            \end{enumerate} 
            
            \item Output set $A:=[n] \setminus S^{(T)}$ and $\xxtil := [\xxtil_A; \xxtil_{S^{(T)}}] = [\xxtil^{(T)}; \mathbf{0}]$
        \end{enumerate}
    \end{algbox}
    \caption{The Threshold Decay framework}
    \labelalg{ThresholdDecay}
\end{figure}

\begin{lemma} [Guarantees of the Threshold Decay framework] \label{lemma:ThresDec-guarantees}
    Given $\epsilon \in (0,1)$, integer $T\geq 10$, $U \ge 2$, an $n \times n$ invertible RDDL matrix $\LL$ with integer entries in $[-U, U]$, and a nonnegative vector $\bb$ with integer weights in $[0,U]$, let $\xxbar := \LL^{-1} \bb$ denote the exact solution vector, \callalg{ThresholdDecay} has the following guarantees for all $t \in [0, T]$: 
    \begin{enumerate}
        \item  \label{induction:bound-b}
            $
            \norm{\widehat{\bb}^{(t)}}_1 \leq \frac{1}{(nU)^{t}} \norm{\bb}_1
            $.
            Moreover, if $t > 0$, 
            $\norm{\bbhat^{(t)}}_1 \leq \frac{1}{(nU)} \norm{\bbhat^{(t-1)}}_1
            $.
        \item  \label{induction:xtil-approx}
            $\widetilde{\xx}^{(t)} \approxbar_{\epsilon t / (4T)} \xxbar_{[n] \setminus S^{(t)}}$. Note that we define $\widetilde{\xx}^{(0)}$ to be a zero dimensional vector.
        \item \label{induction:bound-xbar}
            For all $i \in S^{(t+1)}$, $\xxbar_i < \frac{1}{(nU)^{2}} \| \bbhat^{(t)} \|_1 \le \frac{1}{(nU)^{t+2}} \norm{\bb}_1$.
    \end{enumerate}
\end{lemma}

Before proving \Cref{lemma:ThresDec-guarantees}, we present the following theorem, which states the implication of this lemma for solving an RDDL linear system.

\begin{theorem} [Correctness of the Threshold Decay framework]\label{thm:ThresDec-correctness}
Given $\epsilon \in (0,1)$, integer $T\geq 10$, $U \ge 2$, an $n \times n$ invertible RDDL matrix $\LL$ with integer entries in $[-U, U]$, 
and a nonnegative vector $\bb$ with integer weights in $[0,U]$, \callalg{ThresholdDecay} produces a vector $\xxtil$ and a set $A \subseteq [n]$ such that $\xxtil$ has support $A$, with the following properties hold:  
    \begin{itemize}
        \item For any $i \in [n]$ with $(\LL^{-1} \bb)_i \geq (nU)^{-(T+1)} \norm{\bb}_1$, we have $i \in A$. 
        \item Our output is an entrywise $\exp(\eps)$-approximation on $A$, i.e.,  $\xxtil_A \approxbar_{\epsilon} (\LL^{-1} \bb)_A$.
        \item Specifically, when $T = n$, we have $\xxtil \approxbar_{\epsilon} \LL^{-1} \bb$.
    \end{itemize}
\end{theorem}

\begin{proof}
    Note that the output set $A = [n] / S^{(T)}$. By \Cref{induction:bound-xbar} in \Cref{lemma:ThresDec-guarantees}, we have 
    \begin{align*}
        \| \xxbar_{S^{(T)}} \|_\infty < \frac{1}{(nU)^{T+1}} \| \bb \|_1,
    \end{align*}
     so the first statement holds.
    By \Cref{induction:xtil-approx}, we have
    \begin{align*}
        \xxtil^{(T)} \approxbar_{\eps/4} \xxbar_{[n]\setminus S^{(T)}},
    \end{align*}
    which implies the second statement. 

    For the third statement, we want to show $(\LL^{-1} \bb)_i = 0$ for all $i \in [n] \setminus A$. 
    By Cramer's rule, for any $i \in [n]$, the entry $(\LL^{-1} \bb)_i$ is either zero or at least $1/\det(\AA^{(i)}_{\bb})$, where $\AA^{(i)}_{\bb}$ is the matrix obtained from replacing the $i$-th column of $\AA$ by $\bb$. By Hadamard's inequality, we have
    \begin{align*}
        \det(\AA^{(i)}_{\bb}) \le U^n n^{n/2} \le (nU)^T.
    \end{align*}
    Therefore, if $(\LL^{-1} \bb)_i$ is not zero, we have 
    \begin{align*}
        \frac{1}{(nU)^{T+1}} \| \bb \|_1 \le \frac{1}{(nU)^{T+1}} \cdot (nU) \le 1/ \det(\AA^{(i)}_{\bb})
    \end{align*}
    and thus $i \in A$, concluding the third statement.
\end{proof}

\begin{proof}[Proof of \Cref{lemma:ThresDec-guarantees}.]
We prove the first two items by induction on $t \ge 0$. The third item follows as a corollary. 

For the base case, \Cref{induction:bound-b} holds since $\bbhat^{(0)} = \bb$. \Cref{induction:xtil-approx} holds since $S^{(0)} = [n]$ and $\widetilde{\xx}^{(0)}$ is a zero dimensional vector.

Now suppose that the induction hypothesis holds for $t$ and we prove it for $t+1$. 
By construction in Step \ref{ThresDec:L-solver}, we have
\[
\norm{\widehat{\xx}^{(t)} - \LL_{S^{(t)}, S^{(t)}}^{-1} \widehat{\bb}^{(t)}}_2 \leq \frac{\epsilon}{64 T (nU)^2} \norm{\widehat{\bb}^{(t)}}_2.
\]
Therefore
\begin{align} \label{eq:xxhat-bound}
    \norm{\widehat{\xx}^{(t)} - \LL_{S^{(t)}, S^{(t)}}^{-1} \widehat{\bb}^{(t)}}_{\infty} \leq \frac{\epsilon}{64 T (nU)^2} \norm{\widehat{\bb}^{(t)}}_1.
\end{align}
Thus for all $i \in F^{(t)}$, by triangle inequality, we have 
\[
(\LL_{S^{(t)}, S^{(t)}}^{-1}  \widehat{\bb}^{(t)})_i \geq  \widehat{\xx}^{(t)}_i - \frac{\epsilon}{64 T (nU)^2} \norm{\widehat{\bb}^{(t)}}_1 \geq \theta_t - \frac{\epsilon}{64 T (nU)^2} \norm{\widehat{\bb}^{(t)}}_1 \geq \frac{15}{64(nU)^2} \norm{\widehat{\bb}^{(t)}}_1.
\]
Therefore by \Cref{lemma:large_entries_guarantee}, for all $i \in F^{(t)}$,
\begin{align}
\label{eq:ft-approx}
\widehat{\xx}^{(t)}_i \approx_{\epsilon/(8T)} (\LL_{S^{(t)}, S^{(t)}}^{-1} \widehat{\bb}^{(t)})_i\,.
\end{align}
Moreover for $i \in S^{(t+1)}$, we have 
\begin{align}
\label{eq:s-t-plus-1-xhat-bound}
\widehat{\xx}^{(t)}_i < \theta_t \leq  \frac{1}{2(nU)^2} \norm{\widehat{\bb}^{(t)}}_1.
\end{align}
Next we show for $i \in S^{(t+1)}$,
\begin{align}
\label{eq:approx-sol-entry-bound}
(\LL_{S^{(t)}, S^{(t)}}^{-1} \widehat{\bb}^{(t)})_i < \frac{4}{7(nU)^2} \norm{\widehat{\bb}^{(t)}}_1.
\end{align}
This is trivial if $(\LL_{S^{(t)}, S^{(t)}}^{-1} \widehat{\bb}^{(t)})_i < \frac{1}{2(nU)^2} \norm{\widehat{\bb}^{(t)}}_1$.
If $(\LL_{S^{(t)}, S^{(t)}}^{-1} \widehat{\bb}^{(t)})_i \geq \frac{1}{2(nU)^2} \norm{\widehat{\bb}^{(t)}}_1$, then by \Cref{lemma:large_entries_guarantee},
\begin{align*}
    (\LL_{S^{(t)}, S^{(t)}}^{-1} \widehat{\bb}^{(t)})_i \leq \frac{\frac{1}{2(nU)^2}}{\frac{1}{2(nU)^2} - \frac{\epsilon}{64 T(nU)^2}} \cdot \widehat{\xx}^{(t)}_i \leq \frac{8}{7} \widehat{\xx}^{(t)}_i \le
    \frac{4}{7(nU)^2} \norm{\widehat{\bb}^{(t)}}_1
\end{align*}
where the last inequality follows from \eqref{eq:s-t-plus-1-xhat-bound}.
Now note that for any $i \in S^{(t+1)}$,
\begin{align} 
    \nonumber
    \LL_{i,i}\overline{\xx}_i 
    & = \bb_{i}
    -
    \sum_{j\neq i }\LL_{i,j}
    \overline{\xx}_j 
    \\ & \geq \nonumber 
    \bb_{i} - \sum_{j \in [n] \setminus S^{(t+1)}} \LL_{i,j} \overline{\xx}_j 
    \\ & \geq  \nonumber
    \bb_{i}
    -
    \exp \left( -  \frac{1\epsilon}{8T} \right)
    \sum_{j \in F^{(t)}}\LL_{i,j} \xxhat^{\left( t \right)}_j -
    \exp \left( - \frac{\epsilon t}{4T} \right)
    \sum_{j \in [n] \setminus S^{(t)}}\LL_{i,j} \xxtil^{\left( t \right)}_j 
    \\ & \geq \label{eq:xbar-bound}
    \exp \left( - \frac{\epsilon(t+0.5)}{4T}  \right) \left(\bb_{i}
    -
    \sum_{j \in [n] \setminus S^{(t+1)}}\LL_{i,j} \xxtil^{\left( t+1 \right)}_j \right),
\end{align}
where the second-to-last inequality is due to \eqref{eq:ft-approx} (applied to the case $j \in F^{(t)}$) and the induction hypothesis \Cref{induction:xtil-approx} (applied to the case $j \in [n] \setminus S^{(t)}$), and the last inequality follows from $\xxtil^{(t+1)} = [\xxhat^{(t)}; \xxtil^{(t)}]$.
Therefore, by Step \ref{ThresDec:updateB},
\[
\bbhat^{\left(t + 1\right)}
\approx_{\epsilon/(8T)}
\bb_{S^{\left( t +1 \right)}}
-
\LL_{S^{\left( t +1 \right)}, [n] \setminus S^{(t+1)}}
\xxtil^{\left( t+1 \right)}\,,
\]
we have
\begin{align}
\label{eq:bound-bhat-with-xbar}
    \bbhat^{\left(t + 1\right)}_i
    \leq
    \exp(\eps/8T)
        \left(\bb_{i} - \sum_{j \in [n] \setminus S^{(t+1)}}\LL_{i,j} \xxtil^{\left( t+1 \right)}_j \right)
    \le
    U \exp \left( \frac{\eps (t+1)}{4T}  \right) \overline{\xx}_i
\end{align}
where the last inequality follows from $\LL_{i,i} \leq U$ and \eqref{eq:xbar-bound}.
Note that by 
the induction hypothesis \Cref{induction:xtil-approx}, $\xxtil^{(t)} \approxbar_{\eps t/(4T)} \xxbar_{[n]\setminus S^{(t)}}$, it follows from \Cref{lemma:remove_large} that
\begin{align*}
    \xxbar_{S^{(t)}}
    \approxbar_{\eps t/(4T)}
    \LL_{S^{(t)},S^{(t)}}^{-1}
    (\bb_{S^{\left( t \right)}}
    -
    \LL_{S^{\left( t\right)}, [n] \setminus S^{(t)}}
    \xxtil^{\left( t \right)}).
\end{align*}
By Step \ref{ThresDec:updateB} in the previous iteration, we have
\begin{align*}
    \bb_{S^{\left( t \right)}}
    -
    \LL_{S^{\left( t\right)}, [n] \setminus S^{(t)}}
    \xxtil^{\left( t \right)}
    \approxbar_{\eps/(8T)}
    \bbhat^{(t)}
\end{align*}
which also holds for $t=0$. Combining the two above gives
\begin{align} \label{eq:xxbar-approx-LL-inverse-bbhat}
    \xxbar_{S^{(t)}}
    \approxbar_{\eps (t+0.5)/(4T)}
    \LL_{S^{(t)},S^{(t)}}^{-1}
    \bbhat^{(t)}.
\end{align}
Thus, for any $i \in S^{(t)}$, we have
\begin{align}
    \overline{\xx}_i 
    & \leq \label{eq:bound-xbar-with-bhat}
    \exp \left(\frac{\epsilon (t+0.5)}{4T} \right)
    (\LL_{S^{(t)},S^{(t)}}^{-1} \widehat{\bb}^{(t)})_i
    \,.
\end{align}
Therefore, by \eqref{eq:bound-bhat-with-xbar}, \eqref{eq:bound-xbar-with-bhat}, and \eqref{eq:approx-sol-entry-bound}, 
we have for any $i \in S^{(t+1)}$, 
\begin{align} \label{eq:bhat-entry-bound}
\bbhat^{\left(t + 1\right)}_i \leq U \exp\left(\frac{\epsilon(2t+2)}{4T} \right)
(\LL_{S^{(t)}, S^{(t)}}^{-1} \widehat{\bb}^{(t)})_i <
U \cdot \frac{7}{4} \cdot \frac{4}{7(nU)^2} \norm{\widehat{\bb}^{(t)}}_1
=
\frac{1}{n^2 U} \norm{\widehat{\bb}^{(t)}}_1\,,
\end{align}
where the last inequality follows from  $\exp(\frac{\epsilon (2t+2)}{4T}) < 7/4$ when $T\geq 10$, $t \in [0,T]$ and $\epsilon\in(0,1)$.
Thus, summing over the entries of $\bbhat^{(t+1)}$ gives 
\begin{align} \label{eq:ell1-norm-bhat-is-decreasing}
    \norm{\widehat{\bb}^{(t+1)}}_1
    <
    \frac{1}{nU} \norm{\widehat{\bb}^{(t)}}_1,
\end{align}
and then by the induction hypothesis \Cref{induction:bound-b},  we obtain 
\begin{align*}
    \norm{\widehat{\bb}^{(t+1)}}_1
    \leq
    \frac{1}{(nU)^{t+1}} \norm{\bb}_1\,
\end{align*}
which proves \Cref{induction:bound-b} for $t+1$.

We need to prove \Cref{induction:xtil-approx}, $\widetilde{\xx}^{(t+1)} \approxbar_{\epsilon (t+1) / (4T)} (\LL^{-1} \bb)_{[n] \setminus S^{(t+1)}}$. This trivially holds for entries $i\in [n] \setminus S^{(t)}$ by construction of $\widetilde{\xx}^{(t)}$ and the induction hypothesis \Cref{induction:xtil-approx} for iteration $t$. Therefore, we only need to prove this for $i \in F^{(t)}$. By \eqref{eq:ft-approx}, for $i\in F^{(t)}$, we have
\begin{align}
\label{eq:ft-approx-2}
\widehat{\xx}^{(t)}_i \approx_{\epsilon/(8T)} (\LL_{S^{(t)}, S^{(t)}}^{-1} \widehat{\bb}^{(t)})_i\,,
\end{align}
and then by \eqref{eq:xxbar-approx-LL-inverse-bbhat},
\[
\widehat{\xx}^{(t+1)}_i \approxbar_{\epsilon (t+1)/(4T)} \overline{\xx}_i\,.
\]
This concludes \Cref{induction:xtil-approx} and completes the induction proof.

Finally, to prove \Cref{induction:bound-xbar}, note that by \eqref{eq:bound-xbar-with-bhat}, \eqref{eq:approx-sol-entry-bound}, and \Cref{induction:bound-b}, for $i\in S^{(t+1)}$,
\[
\overline{\xx}_i \le
\exp \left(\frac{\epsilon (t+0.5)}{4T} \right)
    (\LL_{S^{(t)},S^{(t)}}^{-1} \widehat{\bb}^{(t)})_i
    <
    \frac{7}{4} \cdot \frac{4}{7(nU)^2} \norm{\widehat{\bb}^{(t)}}_1
    \leq
    \frac{1}{(nU)^{t+2}} \norm{\bb}_1\,.
\]
\end{proof}

\subsection{Entrywise Approximate Linear System Solving for RDDL Matrices}

In this section, we present our linear system solver for RDDL matrices as a direct application of the \callalg{ThresholdDecay} framework, equipped with the almost-linear-time solver of \cite{CKPPRSV17}. We prove that an entrywise $\exp(\epsilon)$-approximate solution to any RDDL linear system $\LL \xx = \bb$ can be computed with probability $1-\delta$ and $\Otil((m n^{1+o(1)}) \log^{O(1)}(U\kappa\epsilon^{-1} \delta^{-1}))$ bit operations (\Cref{thm:entrywise-RDDL-solver}).

To plug in a Laplacian solver for Step~\ref{ThresDec:L-solver}, we observe that in iteration $t$, the entries of the vector $\widehat{\bb}^{(t)}$ are roughly on the scale of the threshold $\theta_t$, i.e., within the range $(nU)^{-t \pm O(1)}$. This allows us to run the Laplacian solver under fixed-point arithmetic after appropriate normalization.

For the remainder of this section, we use the notation
\[
\textsc{LaplacianSolver}(\LL, \bb, \epsilon, \delta)
\]
to denote a function that, with probability $1 - \delta$, computes a vector $\widehat{\xx}$ satisfying
\[
\norm{\widehat{\xx} - \LL^{-1} \bb}_2 \leq \epsilon \cdot \norm{\bb}_2.
\]
By \Cref{lemma:near-linear-solver-RDDL}, \textsc{LaplacianSolver} can be implemented with $\widetilde{\mathcal{O}}\big((m + n^{1+o(1)}) \log^{O(1)}(U\kappa\epsilon^{-1} \delta^{-1})\big)$ bit operations.

We begin with the following lemma, which shows that the nonzero entries of $\widehat{\bb}^{(t)}$ are within a factor of $32U$ from each other. 

\begin{lemma} \label{lemma:b-bound}
    In \callalg{ThresholdDecay}, for any integer $t \in [0, T-1]$, let 
    \[
    \bbbar^{(t+1)} := 
        \bb_{S^{\left( t +1 \right)}}
        -
        \LL_{S^{\left( t +1 \right)}, [n] \setminus S^{(t+1)}}
        \xxtil^{\left( t+1 \right)}
    \]
        denote the RHS of Step \ref{ThresDec:updateB}, the entries are either zero or bounded by 
    \begin{align*}
        \frac{1}{4(nU)^2} \norm{\widehat{\bb}^{(t)}}_1
        \le 
        \bbbar^{(t+1)}_i
        \le
        \frac{2}{n^2 U} \norm{\widehat{\bb}^{(t)}}_1
    \end{align*}
    for all $i \in S^{(t+1)}$. Furthermore, since $\bbhat^{(t+1)} \approx_{(\eps/8T)} \bbbar^{(t+1)}$, we have for all $i \in S^{(t+1)}$, 
    \begin{align*}
        \frac{1}{8(nU)^2} \norm{\widehat{\bb}^{(t)}}_1
        \le 
        \bbhat^{(t+1)}_i
        \le
        \frac{4}{n^2 U} \norm{\widehat{\bb}^{(t)}}_1.
    \end{align*}
\end{lemma}
\begin{proof}
    By construction, the entries of $\xxhat^{(t)}$ are at least $\theta_t$, and $\xxtil^{(t+1)}$ is the concatenation of $\xxhat^{(0)}, \xxhat^{(1)}, \dots, \xxhat^{(t)}$, 
    so we have
    \[
    \widetilde{\xx}^{(t+1)}_i
    \geq
    \min_{ \ell \in \{0,\ldots,t\}} \theta_\ell
    = \frac{1}{4(nU)^2} \norm{\widehat{\bb}^{(t)}}_1,
    \]
    where the equality follows from \eqref{eq:ell1-norm-bhat-is-decreasing}, stating that $\|\bbhat^{(t)}\|$ is decreasing. 
    Therefore, for $i \in S^{(t+1)}$, the entry
    \[
    (\bbbar^{(t+1)})_i = 
    \bb_{i}
    -
    (\LL_{S^{\left( t +1 \right)}, [n] \setminus S^{(t+1)}}
    \xxtil^{\left( t+1 \right)})_i
    =
    \bb_{i}
    -
    \sum_{j \in [n] \setminus S^{(t+1)}} \LL_{i,j}
    (\xxtil^{\left( t+1 \right)})_j
    ,
    \]
    is either zero or at least $\frac{1}{4(nU)^2} \norm{\widehat{\bb}^{(t)}}_1$ because $\bb_{i}$ and $(-\LL_{i,j})$ are all non-negative integers.
    Moreover, note that $\bbbar^{(t+1)} \approxbar_{\eps/(8T)} \bbhat^{(t+1)}$, so we have 
    for all $i\in S^{(t+1)}$,
    \[
    \bbbar^{(t+1)}_i \le 2 \bbhat^{(t+1)}_i
    \le \frac{2}{n^2 U} \norm{\widehat{\bb}^{(t)}}_1\,
    \]
    where the last inequality is due to \eqref{eq:bhat-entry-bound}.
\end{proof}

Our entrywise approximate linear system solver for RDDL matrices is presented in \callalg{RDDLSolver}. Next we analyze its bit complexity. 

\begin{figure}[t] 
    \begin{algbox}
        \textbf{\uline{\textsc{RDDLSolver}}}: The entrywise approximate linear system solver for RDDL matrices based on \callalg{ThresholdDecay}. Implementations of Step \ref{ThresDec:L-solver} and \ref{ThresDec:updateB} are plugged in.  \\
        \uline{Input}: 
            $\LL, \bb, T, \eps, \kappa$ of \callalg{ThresholdDecay} and $\delta \in(0,1)$. \\    
        \uline{Output:} vector $\xxtil$ such that $\xxtil \approxbar_{\epsilon} (\LL^{-1} \bb)$. \\
        Run \callalg{ThresholdDecay} for $T$ iterations, with vectors $\bbhat^{(t)}$ represented by $O(\log(nU/\eps))$-bit floating-point numbers, and vectors $\xxhat^{(t)}$ represented by $O(\log(nU\kappa/\eps))$-bit floating-point numbers. We implement Step \ref{ThresDec:L-solver} and Step \ref{ThresDec:updateB} as follows. \\
        \textbf{Step \ref{ThresDec:L-solver}: }
        \begin{enumerate}
            \item 
                    We call the Laplacian solver in \Cref{lemma:near-linear-solver-RDDL} under fixed-point arithmetic
                    \begin{align*}
                        \xxhat^{(t)} \gets \| \bbhat^{(t)}\|_1 \cdot \textsc{LaplacianSolver}(\LL_{S^{(t)}, S^{(t)}}, \bbhat^{(t)} / \| \bbhat^{(t)}\|_1,  \eps_L/2, \delta/n).
                    \end{align*}
        \end{enumerate}
        \textbf{Step \ref{ThresDec:updateB}: }
        \begin{enumerate}
                \item
                    Compute with floating-point arithmetic
                    \begin{align} \label{eq:compute-bhat}
                        \bbhat^{\left(t + 1\right)}
                        \gets 
                        \bb_{S^{\left( t +1 \right)}}
                        -
                        \LL_{S^{\left( t +1 \right)}, [n] \setminus S^{(t+1)}}
                        \xxtil^{\left( t+1 \right)}.
                    \end{align}
        \end{enumerate}
    \end{algbox}
    \caption{The entrywise approximate linear system solver for RDDL matrices}
    \labelalg{RDDLSolver}
\end{figure}

\rddlQuadraticSolve*
\begin{proof}
    We implement Step \ref{ThresDec:L-solver} using the almost-linear-time directed Laplacian solver in \Cref{lemma:near-linear-solver-RDDL}. Note that $\bbhat^{(t)}$ is exponentially small, but by \Cref{lemma:b-bound}, every nonzero entry of $\bbhat^{(t)}$ is bounded by $32U$ factor of each other. If we divide by $\|\bbhat^{(t)}\|_1$, the entries will be either zero or in the range $[1/(8nU), 1]$. So the inputs to the Laplacian solver can be converted to fixed-point numbers with $O(\log(nU))$ bits, and we run the solver normally on the fixed-point representations. We generate $\xxhat^{(t)}$ by
    \begin{align} \label{eq:generate-xL}
         \xx^L := \| \bbhat^{(t)}\|_1 \cdot \textsc{LaplacianSolver}(\LL_{S^{(t)}, S^{(t)}}, \bbhat^{(t)} / \| \bbhat^{(t)}\|_1,  \eps_L/2, \delta/n)
    \end{align}
    so we have 
    \begin{align*}
        \left\| \xx^L / \| \bbhat^{(t)}\|_1  - \LL_{S^{(t)}, S^{(t)}}^{-1} \bbhat^{(t)} / \| \bbhat^{(t)}\|_1  \right\|_2 \le \frac{\eps_L}{2\| \bbhat^{(t)}\|_1} \| \bbhat^{(t)}\|_2
    \end{align*}
    and therefore
    \begin{align*}
        \left\| \xx^L  - \LL_{S^{(t)}, S^{(t)}}^{-1} \bbhat^{(t)} \right\|_2 \le \frac{\eps_L}{2} \| \bbhat^{(t)}\|_2.
    \end{align*}
    We can compute $\xxhat^{(t)}$ by standard floating-point arithmetic with $O(\log(nU\kappa))$ bits such that
    \begin{align*}
        \| \xxhat^{(t)} - \xx^L \|_2
        & \le
            \frac{\eps_L }{2 \kappa} \|\xx^L\|_2
        \\ & \le
            \frac{1}{\kappa} \left( \|\LL_{S^{(t)}, S^{(t)}}^{-1} \bbhat^{(t)}\|_2 + \frac{\eps_L}{2} \| \bbhat^{(t)}\|_2 \right)
        \\ & \le 
            \frac{\eps_L}{2} \| \bbhat^{(t)}\|_2
    \end{align*}
    and the assumption of Step \ref{ThresDec:L-solver}, $\| \xxhat^{(t)} - \LL_{S^{(t)}, S^{(t)}}^{-1} \bbhat^{(t)} \|_2 \le \eps_L \| \bbhat^{(t)}\|_2 $, follows by triangle inequality.
    
    By \Cref{lemma:near-linear-solver-RDDL}, the vector produced by the Laplacian solver is in $O(\log(nU\kappa / \eps))$-bit fixed-point representation. So we convert it into floating-point and compute $\xxhat^{(t)}$, then store he result with $O(\log(nU\kappa /\eps))$-bit floating-point numbers. Therefore, Step \ref{ThresDec:L-solver} can be done in  
    \begin{align*}
        \Otil((m+n^{1+o(1)}) \log^{O(1)}(U\kappa\epsilon^{-1} \delta^{-1}))
    \end{align*}
    bit operations. 
    
    For Step \ref{ThresDec:updateB}, we use floating-point arithmetic operations to compute the RHS, $\bbbar^{(t+1)} := 
        \bb_{S^{\left( t +1 \right)}}
        -
        \LL_{S^{\left( t +1 \right)}, [n] \setminus S^{(t+1)}}
        \xxtil^{\left( t+1 \right)}$.
    Since $\bb_{S^{\left( t +1 \right)}}$, $-
        \LL_{S^{\left( t +1 \right)}, [n] \setminus S^{(t+1)}}$, and $\xxtil^{\left( t+1 \right)}$ are all nonnegative, standard arithmetic operations for floating-point numbers with $O(\log(nUT/\eps)) = O(\log(nU/\eps))$ bits will achieve the desired accuracy $\eps/(8T)$. 
        This step can be done in $\Otil(m)$ bit operations, since $\LL$ only has $m$ nonzero entries.

        Since the approximation guarantees are satisfied for Step \ref{ThresDec:L-solver} and \ref{ThresDec:updateB}, the correctness of this algorithm follows from \Cref{thm:ThresDec-correctness}. For the running time, it is easy to verify the other steps of the algorithm run in $\Otil(n)$ bit operations per iteration. In total, there are $T$ iterations, so the bit complexity is $\Otil((m+n^{1+o(1)}) \cdot T \log^{O(1)}(U\kappa\epsilon^{-1} \delta^{-1}))$. 
\end{proof}

If the input matrix is an SDDM matrix, we can substitute the undirected Laplacian solver from \Cref{cor:near-linear-solver-SDDM} into Step~\ref{ThresDec:L-solver} to obtain \callalg{SDDMSolver}. Then, a proof similar to that of \Cref{thm:entrywise-RDDL-solver} yields the following result.

\begin{corollary} \label{thm:entrywise-SDDM-solver}
    There exists a randomized algorithm \labelalg{SDDMSolver} \callalg{SDDMSolver}, given the inputs $\LL, \bb, \eps$ under the assumption of \Cref{thm:ThresDec-correctness}, if $\LL$ is an SDDM matrix with $m$ nonzero entries, then it runs in $\Otil(m n \log^2(U\eps^{-1} \delta^{-1}))$ bit operations and with probability $1-\delta$ outputs $\xxtil$ such that $\xxtil \approxbar_{\eps} \LL^{-1} \bb$.
\end{corollary}

\subsection{Entrywise Approximate Matrix Inversion for SDDM Matrices}
\label{sec:SDDM-matrix-invert}

In the previous subsection, we showed that a linear system with an RDDL matrix can be entrywise approximately solved using $\Otil((m n^{1+o(1)}) \log^{O(1)}(U\kappa\epsilon^{-1} \delta^{-1}))$ bit operations. Moreover, for SDDM matrices, by \Cref{thm:entrywise-SDDM-solver}, solving one linear system can be performed using $\Otil(m n \log^2(U\eps^{-1} \delta^{-1}))$ bit operations.
Here, we show that the inverse of an SDDM matrix can be entrywise approximately computed using an asymptotically equal number of bit operations.

Let $\LL$ be an invertibe SDDM matrix with integer weights in $[-U,U]$ and $G=([n+1],E,w)$ be the graph associated with $\LL$.
We use the following two insights to develop our algorithm for inverting SDDM matrices with $\Otil(m n \log^2(U\eps^{-1} \delta^{-1}))$ bit operations.

\begin{itemize}
    \item If vertices $i,j \in [n]$ are connected, then $\LL^{-1} \ee^{(i)}$ and $\LL^{-1} \ee^{(j)}$ within a $\poly(nU)$ factor from each other (\Cref{lemma:solution-of-neighbors}).

    \item For a linear system $\LL \xx = \bb$, if $C$ denotes the set of \emph{small} entries of $\LL^{-1} \bb$, then we can solve the reduced system $\LL_{F,F} \yy_F = \bb_F$, where $F = [n] \setminus C$, and the solution to this smaller system closely approximates the original solution on the entries in $F$ (\Cref{lemma:zero_out_error}).
\end{itemize}

These two insights naturally suggest the following algorithmic approach. Suppose we have an entrywise approximation to $\LL^{-1} \ee^{(i)}$ and that vertices $i$ and $j$ are connected. Then, we can construct prediction intervals $[p_k, q_k]$ such that $\frac{q_k}{p_k} = \poly(nU)$ and $(\LL^{-1} \ee^{(j)})_k \in [p_k, q_k]$. 

The second insight allows us to zero out the entries that are much smaller than the largest entry of $\LL^{-1} \ee^{(j)}$, which is at least $\max_k p_k$, and to retain only those entries within a $\poly(nU)$ factor of the largest entry (based on our prediction intervals). We then solve the linear system restricted to this reduced set to recover the large entries of the solution. After that, we substitute our approximation to the large entries to form a smaller linear system, and decrease the threshold used for including entries in the next iteration.

We formalize these ideas in \Cref{lemma:entrywise-SDDM-solver-with-prediction} and show how they can be implemented within the \callalg{ThresholdDecay} framework. This allows us to compute an entrywise approximate solution to the system $\LL \xx = \ee^{(j)}$ using $\Otil\big(m \log^2(U \epsilon^{-1} \delta^{-1})\big)$ bit operations, given an entrywise approximation to $\LL^{-1} \ee^{(i)}$. By repeating this process to compute the solution for all systems $\LL \xx = \ee^{(j)}$ for $j \in [n] \setminus \{i\}$—moving from one vertex to a neighboring one—we can compute the entire inverse using $\Otil\big(m n \log^2(U \epsilon^{-1} \delta^{-1})\big)$ bit operations (\Cref{thm:quadratic-sddm-inversion}).

We start by showing that the solution for neighboring vertices are multiplicatively close to each other.

\begin{lemma}
    \label{lemma:solution-of-neighbors}
    Let $\LL$ be an invertible SDDM matrix with integer weights in $[-U,U]$.
    Let $i,j\in[n]$ such that $\LL_{ij} \neq 0$. Then $\LL^{-1} \ee^{(i)} \leq nU \cdot \LL^{-1} \ee^{(j)}$.    
\end{lemma}
\begin{proof}
    $(\LL^{-1}\ee^{(i)})_{k}$ is the sum of probabilities of all random walks that start at $k$ and end at $i$. Note that for any random walk $w$ that starts at $k$ and ends at $i$, there is a random walk $w' = w \cup \{(i,j)\}$ that starts at $k$ and ends at $j$. The probability of picking the edge $(i,j)$ in the last step is at least $\frac{1}{nU}$. Moreover, there are random walks starting at $k$ and ending at $j$ where the penultimate vertex is not $i$. Therefore 
    \[
    \frac{1}{n U} \cdot (\LL^{-1} \ee^{(i)})_k \leq (\LL^{-1} \ee^{(j)})_k.
    \]
    Rearranging the above equation implies the result.
\end{proof}

Note that for SDDM matrices the corresponding graph is undirected.
Therefore, the above lemma implies mutual bounds that $\LL^{-1} \ee^{(i)} \leq nU \LL^{-1} \ee^{(j)} \leq (nU)^2 \LL^{-1} \ee^{(i)}$. The next lemma shows that if we zero out the small entries in the solution, it does not change the norm-wise error of the solution significantly.

\begin{lemma}
    \label{lemma:zero_out_error}
    Let $\LL \in \R^{n\times n}$ be an invertible SDDM matrix with weights in $[-U, U]$,
    $\bb \in \R^{n}$, and $\xx = \LL^{-1} \bb$. Let $\theta >0$ and $C \subseteq [n]$ such that $\xx_{C} \leq \theta$. Moreover 
    let $\yy$ be the solution obtained by discarding entries in $C$ and
    solving on the remaining entries, $F:=[n] \setminus C$, i.e.  
    \begin{align*}
        \yy := \Mbegin
            \yy_{F} \\
            \yy_{C}
        \Mend 
        =
        \Mbegin
            \LL_{F, F}^{-1} \bb_{F} \\
            \mathbf{0}
        \Mend.
    \end{align*}
    Then
    \[
    \norm{
        \yy
        -
        \xx
    }_{2}
    \leq
    n^3 U \theta.
    \]
\end{lemma}

\begin{proof}
    Note that 
    \begin{align*}
        \Mbegin
            \LL_{F, F} & \LL_{F, C} \\
            \LL_{C, F} & \LL_{C, C}
        \Mend
        \cdot
        \Mbegin
            \xx_{F} \\
            \xx_{C}
        \Mend
        =
        \Mbegin
            \bb_{F} \\
            \bb_{C}
        \Mend.
    \end{align*}
    The original $\xx$ satisfies
    \[
    \LL_{F, F} \xx_{F}
    +
    \LL_{F, C} \xx_{C}
    = 
    \bb_{F},
    \]
    or upon left-multiplying both sides by
    $\LL_{F, F}^{-1}$:
    \[
    \xx_{F}
    +
    \LL_{F, F}^{-1}
    \LL_{F, C} \xx_{C}
    =
    \LL_{F, F}^{-1} \bb_{F}.
    \]
    So the error of what we computed is just
    $\LL_{F, F}^{-1}
    \LL_{F, C} \xx_{C}$.
    Moreover
    \[
    \norm{\xx_{F} - 
    \LL_{F, F}^{-1} \bb_{F}}_2
    =
    \norm{\LL_{F, F}^{-1}
    \LL_{F, C} \xx_{C}}_2
    \leq
    \| \LL_{F, F}^{-1} \|_{2}
    \norm{\LL_{F, C}}_{2}
    \norm{\xx_{C}}_2.
    \]
    We incorporate bounds on each of these terms:
    \begin{itemize}
        \item By \Cref{cor:sddm-inverse-bound}, the first term is bounded by $|F|^2$.
        \item The second is at most $\sqrt{|F|} U$ because edge weights of $\LL$
        are at most $U$.
        \item The third is at most $\sqrt{n-|F|} \theta$ due to the given condition of
        $\xx_{C} \leq \theta$ entry-wise.
    \end{itemize}
    Multiplying them gives the $|F|^{2.5} \sqrt{n - |F|} U \theta$ bound. If $C = \emptyset$, then $\yy=\xx$ and the result trivially holds. Otherwise, $|F| \leq n-1$. Then
    \begin{align*}
        \| \yy - \xx \|_2 \le \| \xx_F - \LL_{F, F}^{-1} \bb_F \|_2 + \| \xx_C\|_2 \le |F|^{2.5} \sqrt{n - |F|} U \theta + \sqrt{n - |F|} \theta \le n^3 U \theta.
    \end{align*}
\end{proof}

\begin{definition}
    For $k \ge 1$, a vector $\pp \in \R^n$ is called a $k$-prediction for vector $\vv \in \R^n$ if for all $i \in [n]$, 
    \begin{align*}
        \pp_i / (nU)^k \le \vv_i \le \pp_i.    
    \end{align*} 
\end{definition}

Now we can implement Step \ref{ThresDec:L-solver} efficiently by using a good prediction vector. We present this algorithm in \callalg{PredictiveSDDMSolver}. 

\begin{figure}[ht] 
    \begin{algbox}
        \textbf{\uline{\textsc{PredictiveSDDMSolver}}}: The entrywise approximate linear system solver for SDDM matrices based on \callalg{ThresholdDecay} and a prediction vector. \\
        \uline{Input}: 
            $\LL, \bb, \eps, \kappa$ of \callalg{ThresholdDecay}, where $\LL$ is an invertible SDDM matrix, \\
            $\pp$: a $k$-prediction for the solution $\LL^{-1} \bb$ such that $\pp / (nU)^k \le \LL^{-1} \bb \le \pp$. \\
        \uline{Output:} vector $\xxtil$ such that $\xxtil \approxbar_{\epsilon} \LL^{-1} \bb$. \\
        Run \callalg{ThresholdDecay}, with vectors $\bbhat^{(t)}, \xxhat^{(t)}$ represented by $O(\log(nU/\eps))$-bit floating-point numbers. We implement Step \ref{ThresDec:L-solver} and Step \ref{ThresDec:updateB} as follows. \\
        \textbf{Step \ref{ThresDec:L-solver}: }
        \begin{enumerate}
            \item Let $\zeta^{(t)}$ be the smallest power of two that is at least $\eps_L \| \bbhat^{(t)} \|_1 /  (8 n^3 U)$.
            \item Compute $V^{(t)} := \left\{ i \in S^{(t)}: \pp_i \ge \zeta^{(t)} \right\}$.
            \item 
                    We call the Laplacian solver in \Cref{cor:near-linear-solver-SDDM} under fixed-point arithmetic,
                    \begin{align*}
                        \xxhat^{(t)}_{V^{(t)}} \gets \| \bbhat^{(t)}_{V^{(t)}}\|_1 \cdot (\textsc{LaplacianSolver}(\LL_{V^{(t)}, V^{(t)}}, \bbhat^{(t)}_{V^{(t)}} / \| \bbhat^{(t)}_{V^{(t)}}\|_1,  \eps_L/4, \delta/n)),
                    \end{align*}
                    such that $\xxhat^{(t)}$ is a vector with all nonzero entries only in $V^{(t)}$.
        \end{enumerate}
        \textbf{Step \ref{ThresDec:updateB}: } Initialize $n$-dimensional vector $\vvhat^{(0)} = 0$. In iteration $t$:
        \begin{enumerate}
                \item
                    Compute $\vvhat^{(t+1)}$ defined on $S^{(t+1)}$, using in-place updates to $\vvhat^{(t)}$, 
                    \begin{align} \label{eq:compute-v} 
                        \vvhat^{(t+1)} \gets \vvhat^{(t)}_{S^{(t+1)}} 
                        -
                        \LL_{S^{\left( t +1 \right)}, F^{(t)}}
                        \xxhat^{\left( t \right)}_{F^{(t)}}.
                    \end{align}
                \item 
                    Then, we have an implicit representation for $\bbhat^{(t+1)}$, that is, 
                    \begin{align} \label{eq:bhat-rep} 
                        \bbhat^{\left(t + 1\right)}
                        :=
                        \bb_{S^{\left( t +1 \right)}}
                        +
                        \vvhat^{(t+1)}.
                    \end{align}
        \end{enumerate}
    \end{algbox}
    \caption{The entrywise approximate linear system solver for SDDM matrices given a prediciton}
    \labelalg{PredictiveSDDMSolver}
\end{figure}

\begin{lemma} \label{lemma:entrywise-SDDM-solver-with-prediction}
     Let $\LL, \bb, \eps$ be the inputs that satisfy the assumptions of \Cref{thm:ThresDec-correctness} and $\LL$ has $m$ nonzero entries. Given a vector $\pp$ that is a $k$-prediction vector for $\LL^{-1} \bb$, \callalg{PredictiveSDDMSolver} runs with
     $\Otil( k(m+n) \log^2 (U \eps^{-1} \delta^{-1}))$ bit operations and with probability $1-\delta$ outputs vector $\xxtil$ such that $\xxtil \approxbar_{\epsilon} \LL^{-1} \bb$.
\end{lemma}

\begin{proof}
    We prove the correctness and then analyze the bit complexity.
    
    \textbf{Correctness: Step~\ref{ThresDec:L-solver}. } 
    We call the near-linear time Laplacian solver in \Cref{lemma:near-linear-solver}, but on a smaller set of vertices. For iteration $t$,
    we solve on the set $V^{(t)}$
    with error $\eps_L/4$, the same arguments as \eqref{eq:generate-xL}
    give the guarantee 
    \begin{align*}
        \| \xxhat^{(t)}_{V^{(t)}} - \LL_{V^{(t)}, V^{(t)}}^{-1} \bbhat^{(t)}_{V^{(t)}} \|_2 
        \le
        \frac{\eps_L}{2}
        \| \bbhat^{(t)}_{V^{(t)}} \|_2.
    \end{align*}
    By \Cref{lemma:remove_large}, the system reduced to $S^{(t)}$ satisfies
    \begin{align*}
        (\LL^{-1}_{S^{(t)}, S^{(t)}} \bbhat^{(t)})_i \approxbar_{\eps} (\LL^{-1} \bb)_i
    \end{align*}
    for all $i \in S^{(t)}$. Thus, by the property of $V^{(t)}$, we have
    \begin{align} \label{eq:V-property-pSDDM}
        (\LL^{-1}_{S^{(t)}, S^{(t)}} \bbhat^{(t)})_i \le 2 (\LL^{-1} \bb)_i \le 2 \pp_i \le 2 \zeta^{(t)}    
    \end{align}
    for all $i \in S^{(t)} \setminus V^{(t)}$.
    Then, we can apply \Cref{lemma:zero_out_error} to the submatrix $\LL_{S^{(t)}, S^{(t)}}$. Let $\yy := [\yy_{V^{(t)}} ; \yy_{S^{(t)} \setminus V^{(t)}}] := [\LL_{V^{(t)}, V^{(t)}}^{-1} \bbhat^{(t)}_{V^{(t)}}  ; \mathbf{0}]$, we have
    \begin{align*}
        \| \yy - \LL_{S^{(t)}, S^{(t)}}^{-1} \bbhat^{(t)} \|_2
        \le
        n^3 U \cdot 
        2\zeta^{(t)}
        \le
        \frac{\eps_L}{2} \| \bbhat^{(t)} \|_2
    \end{align*}
    where  $\zeta^{(t)} \le  \eps_L \| \bbhat^{(t)} \|_1 /  (4 n^3 U)$.
    By triangle inequality, since $\xxhat^{(t)}$ has support on $V^{(t)}$,
    \begin{align*}
        \| \xxhat^{(t)}  - \LL_{S^{(t)}, S^{(t)}}^{-1} \bbhat^{(t)} \|_2
        & \le
        \| \xxhat^{(t)}_{V^{(t)}} - \LL_{V^{(t)}, V^{(t)}}^{-1} \bbhat^{(t)}_{V^{(t)}} \|_2 
        +
        \| \yy - \LL_{S^{(t)}, S^{(t)}}^{-1} \bbhat^{(t)} \|_2
        \\ & \le
        \frac{\eps_L}{2}
        \| \bbhat^{(t)}_{V^{(t)}} \|_2
        +
        \frac{\eps_L}{2} \| \bbhat^{(t)} \|_2
        \le
        \eps_L \| \bbhat^{(t)} \|_2,
    \end{align*}
    so the approximation guarantee of Step \ref{ThresDec:L-solver} is satisfied.

    \textbf{Correctness: Step~\ref{ThresDec:updateB}.} We show that we are equivalently computing \eqref{eq:compute-bhat}. Define
    \begin{align*}
         \vv^{(t)} := 
         -
        \LL_{S^{\left( t \right)}, [n] \setminus S^{(t)}}
        \xxtil^{\left( t \right)}.
    \end{align*}
    Let $\bbbar^{(t+1)} := 
        \bb_{S^{\left( t +1 \right)}}
        -
        \LL_{S^{\left( t +1 \right)}, [n] \setminus S^{(t+1)}}
        \xxtil^{\left( t+1 \right)}$
    denote the vector we want to compute, we have 
    \begin{align*}
    \vv^{(t+1)} & =
    \bbbar^{(t+1)} - \bb_{S^{\left( t +1 \right)}}
    \\ & =
        -
        \LL_{S^{\left( t +1 \right)}, [n] \setminus S^{(t+1)}}
        \xxtil^{\left( t+1 \right)}
    \\ & =
        - 
         \LL_{S^{\left( t +1 \right)}, [n] \setminus S^{(t)}}
          \xxtil^{(t)}
        -
         \LL_{S^{\left( t +1 \right)}, F^{(t)}}
          \xxhat^{(t)}
    \\ & = 
         (-\LL_{S^{\left( t \right)}, [n] \setminus S^{(t)}}
          \xxtil^{(t)})_{S^{(t+1)}}
        -
         \LL_{S^{\left( t +1 \right)}, F^{(t)}}
          \xxhat^{(t)}
        =
         (\vv^{(t)})_{S^{(t+1)}}
        -
         \LL_{S^{\left( t +1 \right)}, F^{(t)}}
          \xxhat^{(t)}
    \end{align*}
    where the third equality follows from $[\xxtil^{(t+1)}_{[n] \setminus S^{(t)}};\xxtil^{(t+1)}_{F^{(t)}}] = [\xxtil^{(t)}; \xxhat^{(t)}_{F^{(t)}}]$ in Step \ref{ThresDec:updateX}. By standard arithmetic operations, using $O(\log(nU \epsilon^{-1}))$ bits floating-point representations, we can compute for all $t$,
    $
        \vvhat^{(t+1)} \approxbar_{\eps/(8n)} \vv^{(t+1)}
    $
    and thus
    $
        \bbhat^{(t+1)} \approxbar_{\eps/(8n)} \bb_{S^{(t+1)}} + \vv^{(t+1)},
    $
    where the error does not blow up because we only compute the sum of nonnegative numbers.
    Therefore, the approximation guarantee of Step \ref{ThresDec:updateB} is satisfied.
    
    Since the approximation guarantees of Step \ref{ThresDec:L-solver} and \ref{ThresDec:updateB} are satisfied, the correctness of the algorithm then follows from the correctness of the Threshold Decay framework by \Cref{thm:ThresDec-correctness}.

    \textbf{Bit complexity. } For the running time, note that we cannot afford $O(n)$ operations per iteration, so we cannot compute a single whole vector of dimension $n$. However, observe that $\xxhat^{(t)}$ only has nonzero entries in $V^{(t)}$, so all the steps \ref{ThresDec:set-threshold}, \ref{ThresDec:extractF}, \ref{ThresDec:updateX} can be done in $\Otil( | V^{(t)}|)$ arithmetic operations. We aim to bound
    \begin{align} \label{eq:sum-size-V}
        \sum_{t=0}^{T} | V^{(t)} |
        =
        \sum_{i=1}^{n} \sum_{t=0}^{T} \mathbbm{1}[i \in V^{(t)}]
        =
        \sum_{i=1}^{n} \sum_{t=0}^{T} \mathbbm{1}[\pp_i \ge \zeta^{(t)} \land i \in S^{(t)}].
    \end{align}
    Fix an $i \in [n]$. By the third guarantee of \ref{lemma:ThresDec-guarantees}, if $(\LL^{-1} \bb)_i \ge \| \bbhat^{(t)} \|_1 / (nU)^2$, then $i \notin S^{(t+1)}$. Thus, the condition $i \in S^{(t)}$ implies
    \begin{align} \label{eq:cond1}
        (\LL^{-1} \bb)_i < \| \bbhat^{(t-1)} \|_1 / (nU)^2.
    \end{align}
    Moreover, since $\pp$ is a $k$-prediction, we have $\pp_i / (nU)^k \le (\LL^{-1} \bb)_i \le \pp_i$. Thus, $\pp_i \ge \zeta^{(t)}$ implies
    \begin{align*}
        (\LL^{-1} \bb)_i (nU)^k \ge \zeta^{(t)} \ge \frac{\eps_L}{4n^3 U} \| \bbhat^{(t)} \|_1.
    \end{align*}
    Lastly, by the first guarantee of \Cref{lemma:ThresDec-guarantees}, we have $\| \bbhat^{(t+1)} \|_1 \le \| \bbhat^{(t)} \|_1 / (nU)$. Therefore, if $\pp_i \ge \zeta^{(t)}$ holds, we have for all $\Delta \ge k + O(1)$, 
    \begin{align*}
        \| \bbhat^{(t+\Delta-1)} \|_1 / (nU)^2
        \le
        \| \bbhat^{(t)} \|_1 / (nU)^{\Delta+1}
        \le
         (\LL^{-1} \bb)_i (nU)^{k-\Delta-1}\cdot \frac{4n^4U}{\eps_L} < (\LL^{-1} \bb)_i,
    \end{align*}
    which contradicts \eqref{eq:cond1} for $t + \Delta$, so $i \notin S^{(t+\Delta)}$. Therefore,
    \begin{align*}
        \sum_{t=0}^{T-1} \mathbbm{1}[\pp_i \ge \zeta^{(t)} \land i \in S^{(t)}] \le k + O(1)
    \end{align*}
    and we plug in \eqref{eq:sum-size-V} and obtain $\sum_{t=0}^{T} | V^{(t)} | = O(kn)$. To bound the bit complexity of Laplacian solver in Step \ref{ThresDec:L-solver}, It suffices to bound the sum of the number of nonzero entries of $\LL_{V^{(t)}, V^{(t)}}$, which corresponds to the number of edges in the induced graph of $V^{(t)}$. Since each vertex only appears in $k+O(1)$ many sets of $V^{(t)}$, the total number of edges is $O(km)$. By \Cref{cor:near-linear-solver-SDDM}, it takes
    \begin{align*}
        \Otil(km \log^2(U\eps^{-1} \delta^{-1}))
    \end{align*}
    bit operations.
    Therefore, the total bit complexity of Steps \ref{ThresDec:L-solver}, \ref{ThresDec:set-threshold}, \ref{ThresDec:extractF}, \ref{ThresDec:updateX} is $\Otil(km  \log^2(U\eps^{-1} \delta^{-1})$, since the other steps are dominated by Step \ref{ThresDec:L-solver}.

    For Step \ref{ThresDec:updateB}, note that 
    $\vvhat^{(t+1)}$ can be obtained from $\vvhat^{(t)}$ by first zeroing out the entries of $\vvhat^{(t)}$ in $F^{(t)}$, and then compute $\LL_{S^{(t+1)}, F^{(t)}} \xxhat^{(t)}_{F^{(t)}}$. This 
    can be computed in $O(|F^{(t)}| + \textrm{nnz}_{\LL}(S^{\left( t +1 \right)}, F^{(t)}))$ arithmetic operations if we only compute the nonzero entries, where we denote for matrix $\MM$ of size $n \times n$ and sets $S, T \subseteq [n]$, 
    \begin{align*}
        \textrm{nnz}_{\MM}(S, T) := \sum_{i \in S, j \in T} \mathbbm{1}[\MM_{i,j} \neq 0]
    \end{align*}
    as the number of nonzero entries of the submatrix $S \times T$ of $\MM$. Note that we implicitly store the vectors by the positions and the values of the nonzero entries, so using standard data structures like dictionaries will give $O(\log n)$ time for a single query on an entry of $\bbhat^{(t)}$. Therefore, summing over $t$, the total number of arithmetic operations for this step is
    \begin{align} \label{eq:update-b-time}
        \sum_{t=0}^{T} O(|F^{(t)}| + \textrm{nnz}_{\LL}(S^{\left( t +1 \right)}, F^{(t)}))
        \le
        \sum_{t=0}^{T} O(|F^{(t)}| + \textrm{nnz}_{\LL}([n], F^{(t)}))
        \le
        O(n + m).
    \end{align}
    Together with the previous paragraph, we conclude the bit complexity of this algorithm.     
\end{proof}

We present our algorithm for inverting an SDDM matrix in 
\callalg{InvertSDDM} by first solving a single column of the inverse matrix, and propagating to its neighbors using the column itself as a prediction.

\begin{figure}[ht] 

    \begin{algbox}    
       \textbf{\uline{\textsc{InvertSDDM}}}: The entrywise approximate matrix inversion for SDDM matrices based on \callalg{ThresholdDecay} and \callalg{PredictiveSDDMSolver}. \\
        \uline{Input}: 
            $\LL, \eps$ of \callalg{ThresholdDecay}, where $\LL$ is an invertible SDDM matrix.\\
             WLOG, assume that $\LL$ is irreducible. \\
        \uline{Output:} matrix $\ZZ$ such that $\ZZ \approxbar_{\epsilon} \LL^{-1}$. \\
        
        \begin{enumerate}
            \item Let $G = ([n], E)$ be the underlying graph of $\LL$, such that $(i,j) \in E \Leftrightarrow \LL_{i,j} \neq 0$.
            \item \label{SDDMInvert:compute-1}
                Compute $\pp^{(1)} \gets $\callalg{SDDMSolver}$(\LL, \ee^{(1)}, \eps, \delta)$. 
            \item \label{SDDMInvert:BFS}
                 Explore the graph in BFS order, $v_1=1, v_2, v_3, \dots, v_n$. For $i = 2 \dots n$:
                 \begin{enumerate}
                    \item Let any of $v_i$'s explored neighbor be $v_j$ with $j < i$.
                     \item Compute $\pp^{(v_i)} \gets$  \callalg{PredictiveSDDMSolver}$(\LL, \ee^{(v_i)}, \eps, \pp^{(v_j)} \cdot (2nU), \delta/n)$.
                 \end{enumerate}
            \item Return matrix $\ZZ = [\pp^{(1)}\ \pp^{(2)}\ \dots \ \pp^{(n)}]$. 
        \end{enumerate}
       
    \end{algbox}
    \caption{The entrywise approximate matrix inversion for SDDM matrices}
     \labelalg{InvertSDDM}
\end{figure}

\quadSDDMInv*

\begin{proof}
    We first note that if $\LL$ is not irreducible, we can invert the corresponding submatrix of each connected component $C_1, C_2, \dots, C_q$, and the inverse is the matrix $\ZZ$ with $\ZZ_{C_i, C_i} = (\LL_{C_i, C_i})^{-1}$ and zeros in the rest entries. Therefore, WLOG we assume that $\LL$ is irreducible. 
    
    For the correctness, note that $\ZZ \approxbar_{\epsilon} \LL^{-1}$ if and only if $\pp^{(i)} \approxbar_{\epsilon} \LL^{-1} \ee^{(i)}$ for all $i \in [n]$. Step \ref{SDDMInvert:compute-1} has the guarantee $\pp^{(1)} \approxbar_\eps \LL^{-1} \ee^{(1)}$ and runs in $\Otil(mn \log^2 (U\eps^{-1} \delta^{-1}))$ from \Cref{thm:entrywise-SDDM-solver}. In Step \ref{SDDMInvert:BFS}, all the vertices are explored in the BFS order $v_1, v_2, \dots, v_n$. We prove by induction that $\pp^{(v_i)} \approxbar_\eps \LL^{-1} \ee^{(v_i)}$. The base case holds since $v_1 = 1$. For $i>1$, since $v_j$ is a neighbor of $v_i$ with $j<i$, we have $\pp^{(v_j)} \approxbar_\eps \LL^{-1} \ee^{(v_j)}$. By \Cref{lemma:solution-of-neighbors}, since $v_j$ and $v_i$ are mutual neighbors, we have
    \begin{align*}
        \frac{1}{(nU)^2} \LL^{-1} \ee^{(v_j)}
        \le
        \frac{1}{nU} \LL^{-1} \ee^{(v_i)}
        \le
        \LL^{-1} \ee^{(v_j)}.
    \end{align*}
    Therefore, $2nU \pp^{(v_j)}$ is a $3$-prediction of $\LL^{-1} \ee^{(v_i)}$. By \Cref{lemma:entrywise-SDDM-solver-with-prediction}, $\pp^{(v_i)} \approxbar_\eps \LL^{-1} \ee^{(v_i)}$, and the correctness follows.

    For the bit complexity, we bound the total number of bit operations of one call of \callalg{RDDLSolver} and $n-1$ calls of \callalg{PredictiveSDDMSolver} with $3$-prediction, so the algorithm runs in 
    \begin{align*}
        \Otil( mn \log^2 (U\eps^{-1} \delta^{-1})) + (n-1) \cdot \Otil( 3m\log^2 (U\eps^{-1} \delta^{-1})) = \Otil(mn \log^2 (U\eps^{-1} \delta^{-1}))
    \end{align*}
    bit operations.
\end{proof}

\subsection{Faster Entrywise Approximate Linear System Solving for SDDM Matrices}
\label{sec:faster-sddm-solve}

Recall that in \Cref{thm:entrywise-SDDM-solver}, as a direct consequence of \callalg{RDDLSolver}, we obtained an algorithm for solving SDDM linear systems that runs in $\Otil(mn \log^2 (U\eps^{-1} \delta^{-1}))$ bit operations. In this section, we improve the running time of the SDDM solver to $\Otil(m \sqrt{n} \log^2 (U\eps^{-1} \delta^{-1}))$ using the ideas for SDDM matrix inversions in \Cref{sec:SDDM-matrix-invert}.

Note that the \callalg{PredictiveSDDMSolver} relied on upper bounds on the values of small entries of the solution to make them zero. Here since we are solving a linear system from scratch, we do not have access to such upper bounds and prediction intervals.

First, we prove the following lemma analogous to \Cref{lemma:solution-of-neighbors}, showing that if $i, j$ are neighboring vertices, the two coordinates of the solution $\LL^{-1} \bb$ will be ``close'' to each other, bounded by factor $nU$.

\begin{lemma}
    \label{lemma:solution-of-neighbors-single-system}
    Let $\LL$ be an $n\times n$ invertible SDDM matrix with integer weights in $[-U,U]$, and $\bb$ be a nonnegative vector of dimension $n$.
    For any $i,j\in[n]$ with $\LL_{ij} \neq 0$, we have
    \begin{align*}
        (nU)^{-1} \cdot (\LL^{-1} \bb)_j \le (\LL^{-1} \bb)_i \leq nU \cdot (\LL^{-1} \bb)_j.
    \end{align*}
\end{lemma}
\begin{proof}
    By \Cref{lemma:solution-of-neighbors} and the fact that $\LL$ is symmetric, for any $k \in [n]$, 
    \begin{align*}
        (nU)^{-1} \cdot (\LL^{-1})_{jk}  \le (\LL^{-1})_{ik}  \leq nU \cdot (\LL^{-1})_{jk}.
    \end{align*}
    Multiplying by $\bb_k$ and summing over $k$ concludes the lemma.
\end{proof}
However, when solving one linear system, the smaller entries of $(\LL^{-1} \bb)$ can be extracted only when the larger entries are approximated accurately. The lemma applied to the larger entries only gives \emph{lower bounds} on the smaller entries. This does not help us reduce the number of vertices involved in the Laplacian solver if we apply \Cref{lemma:zero_out_error}, since the lemma requires \emph{upper bounds} on the smaller entries so that we remove them without affecting the larger entries in the solution. 

Recall that in Step~\ref{ThresDec:L-solver}, we solve a reduced Laplacian system by a standard solver which gives norm-wise guarantees.
Every vertex is involved in the solver for $O(n)$ iterations, leading to a running time of $\Otil(mn)$. Our new idea is to gradually add vertices into the Laplacian solver, so that the extremely small entries, e.g., whose values smaller than $(nU)^{-\sqrt{n}}$ times the threshold, are not involved in the Laplacian solver. We introduce the boundary set and prove the following upper bound for our purpose.

\begin{definition}[(Inner) boundary of a set]
    For an undirected graph $G = (V, E)$, let $F \subseteq S \subseteq V$, the (inner) boundary of $F$ with respect to the subgraph $S$ is
    \begin{align*}
        \partial_S F := \{ v \in F : \exists w \in (S \setminus F), s.t., (v,w) \in E\}.
    \end{align*}
    For an $n \times n$ undirected matrix $\LL$, we define boundary with respect to the underlying graph $G = ([n], E)$ where $E = \{ (i,j) : \LL_{ij} \neq 0\}$. That is, for $F \subseteq S \subseteq [n]$,
    \begin{align*}
        \partial_S F := \{ i \in F :  \exists j \in (S \setminus F), s.t., \LL_{ij} \neq 0\}.
    \end{align*}
    We write $\partial F := \partial_V F$ (or $\partial_{[n]} F$) when $S = V$ (or $S = [n]$).
\end{definition}

\begin{lemma} \label{lemma:boundary-upper-bound}
    Let $\LL$ be an $n\times n$ invertible SDDM matrix with integer weights in $[-U,U]$, and $\bb$ be a nonnegative vector of dimension $n$, we denote $\xx := \LL^{-1} \bb$. Let $F$ be a set with $F \supseteq \{i \in [n] : \bb_i > 0 \}$, and let $C := [n] \setminus F$. Define the maximum entry of the reduced system in $F$ on $\partial F$, 
    \begin{align*}
        \Delta := \max_{u \in \partial F} (\LL_{F, F}^{-1} \bb_F)_u
    \end{align*}
    where we define the maximum to be zero if $\partial F = \emptyset$, then for any $i \in C$, we have $\xx_i \le n^3U \Delta$.
\end{lemma}
\begin{proof}
    Note that 
    \begin{align*}
        \Mbegin
            \LL_{F, F} & \LL_{F, C} \\
            \LL_{C, F} & \LL_{C, C}
        \Mend
        \cdot
        \Mbegin
            \xx_{F} \\
            \xx_{C}
        \Mend
        =
        \Mbegin
            \bb_{F} \\
            \bb_{C}
        \Mend.
    \end{align*}
    The original $\xx$ satisfies
    \[
    \LL_{C, F} \xx_{F}
    +
    \LL_{C, C} \xx_{C}
    = 
    \bb_{C},
    \]
    By multiplying $\LL_{C, C}^{-1}$ to both sides, 
    \begin{align*}
        \xx_C = \LL_{C, C}^{-1} \bb_C - \LL_{C, C}^{-1} \LL_{C, F} \xx_{F} = - \LL_{C, C}^{-1} \LL_{C, F} \xx_{F}.
    \end{align*}
    where the last equality is due to $\bb_C = \mathbf{0}$ by the property of $F$.
    Note that $\LL_{C, F} \xx_{F} = \LL_{C, \partial F} \xx_{\partial F}$, so  
    \begin{align*}
        \| \xx_C \|_2 \le \| \LL_{C, C}^{-1}\|_2 \cdot \|\LL_{C, \partial F}\|_2  \cdot \|\xx_{\partial F}\|_2 \le n^2 \cdot \sqrt{n}U \cdot \sqrt{n} \Delta = n^3 U \Delta.
    \end{align*}
    This implies that $\xx_i \le n^3 U \Delta$ for all $i \in C$.
\end{proof}
We present our algorithm in \callalg{SDDMSolverFast}, optimizing the Step~\ref{ThresDec:L-solver} of the framework \callalg{ThresholdDecay}, by gradually expanding the vertex set $F$ until $\Delta$, the maximum entry on the boundary $\partial F$, is small enough for us to apply \Cref{lemma:zero_out_error}. Once the upper bound is obtained for the small entries outside $F$, the algorithm is similar to \callalg{PredictiveSDDMSolver} -- It runs the Laplacian solver only in $F$.

We introduce \textbf{the expansion step}. Recall that in each iteration, we want to solve for the largest entries in the remaining system on $S^{(t)}$. For this, we use expansion steps in an inner loop. If the maximum entry on the boundary $\partial F$ is small, we are done because there are no large entries outside $F$. Otherwise, there exists a large entry on the boundary, we expand $F$ by including all vertices with distance $\le \sqrt{n}$. The new vertices might be smaller than the threshold, but are still large that they will be removed after $O(\sqrt{n})$ iterations. We will show that there are at most $O(\sqrt{n})$ expansion steps globally.

For $i, j \in [n]$, we denote $\dis(i,j)$ as the minimal number of edges that one needs to take to go from $i$ to $j$ in the underlying graph.

\begin{figure}[t]
    \begin{algbox}
        \textbf{\uline{\textsc{SDDMSolverFast}}}: The entrywise approximate linear system solver for SDDM matrices based on \callalg{ThresholdDecay} and a gradually expanding strategy. \\
        \uline{Input}: 
            $\LL, \bb, \eps, \kappa$ of \callalg{ThresholdDecay}, where $\LL$ is an invertible SDDM matrix. $\delta \in (0,1)$.\\
            WLOG, assume $\LL$ is irreducible. \\
        \uline{Output:} vector $\xxtil$ such that $\xxtil \approxbar_{\epsilon} \LL^{-1} \bb$. \\
        Run \callalg{ThresholdDecay}, with vectors $\bbhat^{(t)}, \xxhat^{(t)}$ represented by $O(\log(nU/\eps))$-bit floating-point numbers. We implement Step \ref{ThresDec:L-solver} and Step \ref{ThresDec:updateB} as follows. \\
        \textbf{Step \ref{ThresDec:L-solver}: } Initially, let $F^{(0)}$ be the set $\{i : \bb_i > 0\}$. 
        In iteration $t$:
        \begin{enumerate}
            \item Let $\eta$ be the smallest power of two that is $\ge \eps_L \| \bbhat^{(t)} \|_1 / (16n^{6}U^2)$. Set $\delta_L \gets \eps_L /(16 n^6 U^2)$.
            \item Repeat until break in Step (b)
            \begin{enumerate}
                \item We call the Laplacian solver in \Cref{cor:near-linear-solver-SDDM} under fixed-point arithmetic,
                    \begin{align*}
                        \xxhat^{(t)}_{V} \gets \| \bbhat^{(t)}_{V}\|_1 \cdot (\textsc{LaplacianSolver}(\LL_{V, V}, \bbhat^{(t)}_{V} / \| \bbhat^{(t)}_{V}\|_1,  \delta_L/4,\delta/2n)).
                    \end{align*}
                where $V := F^{(t)} \cap S^{(t)}$.
                \item \label{SDDMSolverFast:expand} 
                If either $\max_{i \in V \cap \partial F^{(t)}} \xxhat^{(t)}_i  \le \eta$ or $F^{(t)} = [n]$, break the loop and set $F^{(t+1)} \gets F^{(t)}$. \\ 
                Otherwise, perform \textbf{the expansion step} -- Update $F^{(t)}$ such that
                \begin{align*}
                    F^{(t)} \gets \{ i \in [n] :  \exists j \in [n], s.t., \mathrm{dis}(i,j) \le \sqrt{n},\  \xxhat^{(t)}_j > \eta\}
                \end{align*}
                by running a BFS on the underlying graph from the vertex set $\{ j \in [n] : \xxhat^{(t)}_j > \eta\}$.
            \end{enumerate}
        \end{enumerate}
        \textbf{Step \ref{ThresDec:updateB}: } The same implementation as in \callalg{PredictiveSDDMSolver}. 
    \end{algbox}
    \caption{The entrywise approximate linear system solver for SDDM matrices in $\Otil(m \sqrt{n})$ time}
    \labelalg{SDDMSolverFast}
\end{figure}

\mrootnthm*
\begin{proof}
    It suffices to analyze the implementation for Step~\ref{ThresDec:L-solver}, since the implementation for Step~\ref{ThresDec:updateB} is identical to that in \callalg{PredictiveSDDMSolver}, and the correctness and the bit complexity are established in the proof of \Cref{lemma:entrywise-SDDM-solver-with-prediction}. By \eqref{eq:update-b-time}, Step~\ref{ThresDec:updateB} runs in $O(n+m)$ arithmetic operations, which is insignificant compared to other steps.

    Note that we also assume that $\LL$ is irreducible without loss of generality, because the system is independent for each connected component.  
    
    \paragraph{Correctness.} In iteration $t$, when the algorithm returns, it breaks the loop, so either $F^{(t)} = [n]$ or
    \begin{align*}
        \max_{i \in V \cap \partial F} \xxhat^{(t)}_i \le \eta.
    \end{align*}
    We want to apply \Cref{lemma:boundary-upper-bound} to the matrix $\LL_{S^{(t)}, S^{(t)}}$, the vector $\bbhat^{(t)}$ and the set $V$. Let $B := \partial_{S^{(t)}} V = V \cap \partial F^{(t)}$, so we want to show $\Delta := \max_{u \in B}  \LL_{V, V}^{-1} \bbhat^{(t)}_V$ is small. 
    
    Since $\xxhat^{(t)}_V$ is the approximated solution to the system $\LL_{V, V} \xx = \bb_V$, by the same arguments as \eqref{eq:generate-xL}, 
    \begin{align} \label{eq:xhat-approx}
        \| \xxhat^{(t)}_{V} - \LL_{V, V}^{-1} \bbhat^{(t)}_{V} \|_2 
        \le
        \frac{\delta_L}{2}
        \| \bbhat^{(t)}_{V} \|_2.
    \end{align}
    For the former case of the breaking loop condition, we have $B = \emptyset$ and $\Delta = 0$. For the latter case, for all $i \in B$, it follows from \eqref{eq:xhat-approx} that $(\LL_{V, V}^{-1} \bbhat^{(t)}_{V})_i \le 2\xxhat^{(t)} \le 2 \eta$. Therefore, in both cases we have $\Delta \le 2 \eta$ so we apply \Cref{lemma:boundary-upper-bound}: For all $i \in C := S^{(t)} \setminus V$, 
    \begin{align} \label{eq:V-property}
        (\LL_{S^{(t)}, S^{(t)}}^{-1} \bbhat^{(t)})_i \le
        n^3 U \cdot 2\eta
        \le
            \eps_L \| \bbhat^{(t)} \|_1 / (4n^{3} U) = : \zeta^{(t)}.
    \end{align}
    where we used the fact $\eta \le 2 \eps_L \| \bbhat^{(t)} \|_1 / (16n^6 U^2)$.

    Observe that it becomes almost identical to the case in \callalg{PredictiveSDDMSolver}, where $V$ corresponds to $V^{(t)}$ with property in \eqref{eq:V-property-pSDDM}, and $\xxhat^{(t)}$ is produced by the same way.
    The remaining correctness proof is adapted from \callalg{PredictiveSDDMSolver}.
    
    By \Cref{lemma:remove_large}, the system reduced to $S^{(t)}$ satisfies
    \begin{align} \label{eq:reduce-system-approx}
        (\LL^{-1}_{S^{(t)}, S^{(t)}} \bbhat^{(t)})_i \approxbar_{\eps} (\LL^{-1} \bb)_i
    \end{align}
    for all $i \in S^{(t)}$. Thus, by the property of $V$ in \eqref{eq:V-property}, we have
    $(\LL^{-1}_{S^{(t)}, S^{(t)}} \bbhat^{(t)})_i \le \zeta^{(t)}$ for all $i \in S^{(t)} \setminus V$.
    Then, we can apply \Cref{lemma:zero_out_error} to the submatrix $\LL_{S^{(t)}, S^{(t)}}$. Let $\yy := [\yy_{V} ; \yy_{S^{(t)} \setminus V}] = [\LL_{V, V}^{-1} \bbhat^{(t)}_{V}  ; \mathbf{0}]$, we have
    \begin{align*} 
        \| \yy - \LL_{S^{(t)}, S^{(t)}}^{-1} \bbhat^{(t)} \|_2
        \le
        n^3 U \cdot 
        \zeta^{(t)}
        \le
        \frac{\eps_L}{4} \| \bbhat^{(t)} \|_2.
    \end{align*}
    By triangle inequality, since $\xxhat^{(t)}$ has support on $V$, 
    \begin{align*}
        \| \xxhat^{(t)}  - \LL_{S^{(t)}, S^{(t)}}^{-1} \bbhat^{(t)} \|_2
        & \le 
        \| \xxhat^{(t)}_{V} - \LL_{V, V}^{-1} \bbhat^{(t)}_{V} \|_2 
        +
        \| \yy - \LL_{S^{(t)}, S^{(t)}}^{-1} \bbhat^{(t)} \|_2
        \\ & \le
        \frac{\eps_L}{2}
        \| \bbhat^{(t)}_{V} \|_2
        +
        \frac{\eps_L}{4} \| \bbhat^{(t)} \|_2
        \le
        \eps_L \| \bbhat^{(t)} \|_2,
    \end{align*}
    so the approximation guarantee of Step \ref{ThresDec:L-solver} is satisfied.

    \paragraph{Bit complexity.} Let $\xxbar:= \LL^{-1} \bb$. To bound the running time, we need to show the following two properties, which are also the reason for choosing the parameter $\sqrt{n}$.
    \begin{itemize}
        \item Each vertex $i \in [n]$ is involved in the set $V^{(t)}$ for $O(\sqrt{n})$ iterations in $t$, where $V^{(t)}$ denotes the set $F^{(t)} \cap S^{(t)}$ at the end of iteration $t$.
        \item The expansion step in Step~\ref{SDDMSolverFast:expand} is performed at most $O(\sqrt{n})$ times over all iterations.
    \end{itemize}
    For the first property,
    we fix an $i \in F^{(t)}$.
    We want to show the $i$-th entry is large so that it is removed after $O(\sqrt{n})$ iterations, i.e., $(\LL^{-1} \bb)_i \ge \| \bbhat^{(t)} \|_1 \cdot  (nU)^{-O(\sqrt{n})}$.
    For the case $i \in F^{(0)}$, we trivially have $(\LL^{-1} \bb)_i \ge \bb_i \ge 1$.
    Otherwise, there exists $j \in [n]$ such that $\mathrm{dis}(i,j) \le \sqrt{n}$, and $\xxhat^{(t)}_j > \eta$.
    By \eqref{eq:xhat-approx}, we have
    \begin{align*}
        (\LL_{V, V}^{-1} \bbhat^{(t)}_{V})_j \ge \eta - \frac{\delta_L}{2} \| \bbhat^{(t)}_{V} \|_2
        \ge
        \eta / 2.
    \end{align*}
    because $\delta_L \le \eta / \| \bbhat^{(t)} \|_1$. By \eqref{eq:reduce-system-approx}, we have
    \begin{align*}
        (\LL^{-1} \bb)_j \ge (\LL_{S^{(t)}, S^{(t)}}^{-1} \bbhat^{(t)})_j / 2 \ge  (\LL_{V, V}^{-1} \bbhat^{(t)}_V)_j / 2
    \end{align*}
    where the last inequality follows from the combinatorial interpretation that the sum of the random walks becomes smaller if restricted to a submatrix. Therefore, $(\LL^{-1} \bb)_j \ge \eta/4$. By \eqref{lemma:solution-of-neighbors-single-system} applied to vertices within distance $\sqrt{n}$,
    \begin{align*}
        (\LL^{-1} \bb)_i \ge \frac{\eta}{4} \cdot (nU)^{-\sqrt{n}} \ge \| \bbhat^{(t)} \|_1 \cdot (nU)^{-O(\sqrt{n})}.
    \end{align*}
    Note that by \Cref{lemma:ThresDec-guarantees}, we have $\| \bbhat^{(t+1)} \|_1 \le \| \bbhat^{(t)} \|_1 / (nU)$ and $(\LL^{-1} \bb)_i \le \| \bbhat^{(t-1)} \|_1$, so $i$ is included in $V^{(t)}$ for $O(\sqrt{n})$ iterations of $t$, and we conclude the first property.

    For the second property, for the $k$-th expansion step, let the corresponding $t$ be denoted as $t_k$, we define $i \in S^{(t_k)} \cap \partial F$ to be a \textit{trigger} if $\xxhat^{(t_k)}_i > \eta$. Every expansion step is triggered by at least one vertex, except for the last one where $F = [n]$. We fix a sequence of triggers, $v_1, \dots, v_K$ where $v_i$ is a trigger for the $k$-th expansion step for $k \in [1, K]$. For any $1 \le i < j \le K$, we claim that $dis(v_i,v_j) \ge \sqrt{n}$. Otherwise, when we update $F$ in the $i$-th expansion step, we will have
    \begin{align*}
        \dis(v_i, v_j) < \sqrt{n},\quad  \xxhat^{(t_i)}_j > \eta,
    \end{align*}
    so $v_j$ and all its neighbors have distance at most $\sqrt{n}$, therefore are all added to $F$. Since the set $F$ is expanding and never shrinking, $v_j$ is not a trigger since it is not in $\partial F$ in the $j$-th expansion, a contradiction.

    Now, since $v_1, \dots, v_K$ has the property that $\dis(v_i, v_j) \ge \sqrt{n}$ for all $1 \le i < j \le K$, it holds that the sets $N(v_i, \sqrt{n}/2-1)$ are non-intersecting, where 
    \begin{align*}
        N(v, d) := \{ u \in [n] : \dis(u, v) \le k\}.
    \end{align*}
    Therefore, $|N(v_i, \sqrt{n}/2 -1) \ge \sqrt{n}/2$ for all $i \in [K]$, so there are at most a total of $K+1 \le O(\sqrt{n})$ expansion steps, concluding the second property.
    
    We can now bound the total number of vertices involved in the fixed-point Laplacian solver step. Let $V^{(t, q)}$ denotes the set $V = F \cap S^{(t)}$ at the $q$-th call of \textsc{LaplacianSolver}. For any $i$, we have
    \begin{align*}
        \sum_{t=1}^{T} \sum_{q} \mathbbm{1} [ i \in  V^{(t, q)} ] 
        & \le
        \sum_{t=1}^{T}  \mathbbm{1} [ i \in  V^{(t, q)} ] +  \sum_{t=1}^{T} \sum_{q \ge 2}  \mathbbm{1} [ i \in  V^{(t, q)} ]
        \\ & \le 
        \sum_{t=1}^{T} \mathbbm{1}[ i \in V^{(t, 1)}] +  \sum_{t=1}^{T} \sum_{q \ge 2} 1
        \\ & \le 
        O(\sqrt{n}) + O(\sqrt{n})  = O(\sqrt{n}). 
    \end{align*}
    where the last inequality is due to the two properties we proved above.
    Therefore, by \Cref{cor:near-linear-solver-SDDM}, the bit complexity is 
    \begin{align*}
        & \sum_{t=1}^{T} \sum_{q}  \Otil \left(\log^2(U \delta_L^{-1} \delta^{-1}) \sum_{i=1}^{n} \deg (i) \mathbbm{1} [ i \in  V^{(t, q)} ] \right)
        \\ \le & 
         \Otil \left(\log^2 (U\eps^{-1} \delta^{-1})\sum_{i=1}^{n} \deg (i) \sqrt{n} \right)
         =
         \Otil( (m+n) \sqrt{n} \log^2 (U\eps^{-1} \delta^{-1})).
    \end{align*}
    It is easy to verify that other steps are dominated by this step, since they can be done linearly in
    \begin{align*}
        \sum_{ i \in  V^{(t, q)} } \deg(i). 
    \end{align*}
    For the expansion step, we can afford $\Otil(m+n)$ since it is performed for at most $O(\sqrt{n})$ times.
\end{proof}

\section{Subcubic Time Algorithms for RDDL Inversion}
\label{sec:subcube-dense}

In this section, we present subcubic-time randomized algorithms for entrywise inversion of RDDL matrices associated with dense and sparse graphs.

Our algorithm for dense graphs requires $\Otil(n^{\omega + 0.5} \log(U) \log(\delta^{-1}\epsilon^{-1}))$ bit operations to compute, with probability at least $1-\delta$, an entrywise $\exp(\epsilon)$-approximate inverse of an RDDL matrix with integer weights in $[-U,U]$. 

For sparse graphs, our algorithm requires $\Otil\left((m \sqrt{n} + n^{1.5 + o(1)}) \log^{O(1)}(\kappa \epsilon^{-1} \delta^{-1})\right)$ bit operations to compute such an approximate inverse, where $\kappa$ is the condition number of the matrix. The latter is particularly suitable for matrices with polynomial (in $n$) condition number due to its logarithmic dependence on $\kappa$.

Both algorithms build upon our \callalg{ShortcutterFramework}, introduced in \Cref{sec:shortcutter}. This framework requires certain approximate solutions for matrix inversion and computation of specific probabilities. These computations can be carried out using different approaches, and our two distinct algorithms arise precisely from these choices. Specifically, for dense graphs, we utilize fast matrix multiplication techniques to compute the approximate inverse and required probabilities. For sparse graphs, we employ our algorithms from \Cref{sec:quadratic-sparse} along with the almost-linear-time algorithm of \cite{CKPPRSV17} to solve linear systems involving RDDL matrices.

The remainder of this section is organized as follows. In \Cref{sec:shortcutter}, we introduce the \callalg{ShortcutterFramework} and prove its correctness. In \Cref{sec:subcubic-dense-graph} and \Cref{sec:subcubic-sparse-graphs}, we present our algorithms for dense and sparse graphs, respectively, along with proofs of the corresponding bounds on the required number of bit operations.

\subsection{Shortcutters Framework}
\label{sec:shortcutter}

Let $G = ([n+1], E, w)$ be the graph associated with the $n$-by-$n$ RDDL matrix $\LL$. Define diagonal matrix $\DD = \diag(\LL)$ and matrix $\AA = \II - \DD^{-1}\LL$. Then, we have $\LL^{-1} = (\II - \AA)^{-1}\DD^{-1}$. Notice that $(\II - \AA)^{-1} = \sum_{k=0}^\infty \AA^{k}$. The matrix $\AA$ represents the random walk matrix corresponding to $G$, and $(\II - \AA)^{-1}_{ij}$ equals the sum of probabilities of all random walks starting at vertex $i$, ending at vertex $j$, and avoiding vertex $n+1$, with the convention that a random walk of length zero has probability one. Thus, the entry $\LL^{-1}_{ij}$ 
can be readily computed by summing the probabilities of these random walks. We leverage this perspective of $\LL^{-1}$ to explain the \callalg{ShortcutterFramework}.

For the remainder of this section, we denote by $W_{ij}^{(T)}$ the set of all random walks starting at vertex $i$, ending at vertex $j$, and avoiding a subset $T \subseteq [n+1]$. For brevity, we set $W_{ij}^{(n+1)} := W_{ij}^{(\{n+1\})}$.

The \callalg{ShortcutterFramework} selects a random subset $S$ of vertices, each vertex included independently with probability $1/\sqrt{n}$, so the size of $S$ is approximately $\sqrt{n}$. For each pair of vertices $(i, j) \in [n] \times [n]$, we classify random walks from $i$ to $j$ into \emph{long walks} and \emph{short walks}. A \emph{long walk} visits at least $2^k$ unique vertices distinct from $i$ and $j$, where $k$ is the smallest power of two greater than $\sqrt{n} \log(8 n^2 \epsilon^{-1} \delta^{-1})$. Conversely, a \emph{short walk} visits fewer than $2^k$ unique vertices.

\begin{figure}[t]
\begin{algbox}
\textbf{\uline{\textsc{ShortcutterFramework}}}

\uline{Input}: $n$-by-$n$ RDDL matrix $\LL$ with integer edge weights in $[-U, U]$.\\
$\epsilon$: error parameter, $\delta$: probability parameter\\
\uline{Output:} matrix $\ZZ$ such that $\ZZ \approx_{\epsilon} \LL^{-1}$.
\begin{enumerate}
\item Select a random set $S$ in which each $i \in [n]$ is independently included with probability $1/\sqrt{n}$. 
\item Set $\overline{S} = [n] \setminus S$. Let $i \rightarrow g_i$ be the one-to-one monotone mapping from $\overline{S}$ to $[|\overline{S}|]$ and $i\rightarrow h_i$ be the one-to-one monotone mapping from $[|\overline{S}|]$ to $\overline{S}$.
\item 
Let $2^k$ be the smallest power of two larger than $\sqrt{n}\log(8 n^2\epsilon^{-1}\delta^{-1})$.
\item \label{alg-item:subcube-algo-set-N}
Compute the set $N=\{(i,j): i,j \in \overline{S}, d^{G(\overline{S})}(i,j) \leq 2^k\}$.
\item \label{alg-item:subcube-algo-near-pairs} 
Compute matrix $\YY \in \R^{|\overline{S}|\times |\overline{S}|}$ such that for all $(i,j) \in N$, $\YY_{g_i g_j} \approx_{\epsilon/4} (\LL_{\overline{S},\overline{S}})^{-1}_{g_i g_j}$. 

\item 
Set $\XX \in \R^{n\times n}$ to a matrix where entries are zero except for entries $(i,j) \in N$ for which $\XX_{ij} = \YY_{g_i g_j}$.

\item \label{alg-item:subcube-algo-absorption}
Compute the absorption probabilities (see \Cref{def:absorption})
from all vertices onto the set $\widetilde{S}:=S \cup \{n+1\}$
to $\exp(\epsilon/4)$ approximation. Gather the probabilities for $S$ in matrix $\UU \in \R^{n\times |S|}$. 

\item \label{alg-item:subcube-escape-probs}
For all $i \in S$ and $j \in [n]$, compute the sum of probabilities of random walks in $W^{(n+1)}_{ij}$ to $\exp(\epsilon/4)$ approximation and gather the probabilities in matrix $\VV  \in \R^{n\times |S|}$.

\item \label{alg-item:subcube-final-output}
Return $\ZZ$ such that $\ZZ \approxbar_{\epsilon/4} \XX + \UU \VV^\top \diag(\LL)^{-1}$.

\end{enumerate}
\end{algbox}
\caption{Subcubic Time Algorithm for Computing the Entrywise Approximate Inverse}
\labelalg{ShortcutterFramework}
\end{figure}

One can easily see that any long walk hits at least one vertex in $S$ with probability very close to one. Thus, using Markov's inequality, we argue that with high probability, the sum of probabilities of long random walks in $W^{(n+1)}_{ij}$ is within a factor of $\exp(\epsilon/4)$ of the sum of probabilities of those walks that hit at least one vertex in $S$. Hence, computing the latter suffices as an approximation of the former. We then utilize vertices in $S$ as \emph{shortcutters} to compute these probabilities efficiently by computing the absorption probabilities from vertex $i$ into $S \cup \{n+1\}$ (Step \ref{alg-item:subcube-algo-absorption}) and, for each vertex $s \in S$, computing the sum of probabilities of walks from $s$ to $j$ (Step \ref{alg-item:subcube-escape-probs}). The desired sum is then obtained as the dot product of the two resulting $|S|$-dimensional vectors.

If the distance between vertices $i$ and $j$ in $G$ is greater than $2^k$, there are no short walks from $i$ to $j$, and the computation above is sufficient. Similarly, if set $S$ intersects all short walks from $i$ to $j$, i.e., the distance from $i$ to $j$ in $G([n]\setminus S)$ exceeds $2^k$, the previous computation is also sufficient, as all short walks pass through $S$.

Consider now a vertex pair $(i,j)$ with distance at most $2^k$ in the subgraph $G(\overline{S})$, where $\overline{S}=[n]\setminus S$. In this scenario, we compute the contribution of random walks that hit $S$ as described above. We then compute separately the contribution of random walks in $W^{(n+1)}_{ij}$ that do not hit $S$. It can be verified that the sum of probabilities of these walks is captured by the corresponding entry of $(\LL_{\overline{S},\overline{S}})^{-1}$. Our framework relies on the following properties for efficient computation of the sum of probabilities of these walks for \emph{near} vertex pairs:

\begin{itemize}
    \item With high probability, the condition number of $\LL_{\overline{S},\overline{S}}$ is at most $\exp(6\sqrt{n} \log(U) \log(n \delta^{-1}))$. See \Cref{lemma:cond-number-of-submatrix}.

    \item If the distance from $i$ to $j$ is at most $2^k$ in $G(\overline{S})$, there is at least one short walk from $i$ to $j$. Moreover, the probability of a short walk (and thus the sum of probabilities of random walks in $W^{(n+1)}_{ij}$ that do not hit $S$) is at least $U^{-2^k} \geq U^{-2\sqrt{n} \log(8 n^{2} \epsilon^{-1} \delta^{-1})}$.
\end{itemize}

Consequently, for inverting $\LL_{\overline{S},\overline{S}}$, an error parameter
\[
\widehat{\epsilon} = \exp(-20\sqrt{n} \log(U) \log(8 n^2 \delta^{-1}\epsilon^{-1}))
\]
is sufficient to approximate the sum of probabilities of walks that do not hit $S$ within a factor of $\exp(\epsilon/10)$ (see Step \ref{alg-item:subcube-algo-near-pairs}).

We first prove the condition number of $\LL_{\overline{S},\overline{S}}$ is much smaller than the worst-case $\exp(O(n))$ with high probability.

\begin{lemma}
\label{lemma:cond-number-of-submatrix}
Let $\delta \in (0,1)$, $n,U > 3$, and $\MM \in \R^{n\times n}$ be an invertible RDDL matrix with integer entries in $[-U,U]$. Let $S$ be a random subset of $[n]$ where each $i\in [n]$ is independently included in $S$ with probability $1/\sqrt{n}$. Let $\overline{S}=[n]\setminus S$. Then with probability at least $1-\delta$, the condition number of $\MM_{\overline{S},\overline{S}}$ is at most
\[
\exp(6\sqrt{n} \log(U) \log(n \delta^{-1})).
\]
\end{lemma}
\begin{proof}
Let
\[
\DD = \diag(\MM_{\overline{S},\overline{S}}), ~~~ \text{and} ~~~ \AA = \DD^{-1} (\DD - \MM_{\overline{S},\overline{S}}).
\]
be the random walk probability matrix of the induced subgraph $\overline{S}$.
Let $G=([n+1],E,w)$ be the associated graph to $\MM$. Then for any $i\in \overline{S}$ there is a path of length less than or equal to $\lceil\sqrt{n}\log(n\delta^{-1})\rceil$ to $S\cup\{n+1\}$ with probability at least 
\[
1-\left(1-\frac{1}{\sqrt{n}}\right)^{\lceil\sqrt{n}\log(n\delta^{-1}) \rceil}\geq 1-\frac{\delta}{n},
\]
since all vertices are connected to the dummy vertex $n+1$ by \Cref{lemma:rddl-invertible}.
Therefore, by union bound, with probability at least $1-\delta$ for all $i\in \overline{S}$, there is a path of length 
at most $\lceil\sqrt{n}\log(n\delta^{-1})\rceil$ to the set $S\cup\{n+1\}$. We call this event $\mathcal{E}$.

By setting $\alpha=\lceil \sqrt{n}\log(n\delta^{-1}) \rceil$, under event $\mathcal{E}$, a random walk of length $\alpha$ ends up at $\overline{S}\cup \{n+1\}$ with probability at least $U^{-\alpha}$. Therefore for all $i\in [n-|S|]$,
\[
\sum_{j \in [n-|S|]} (\AA^{\alpha})_{ij} \leq 1- U^{-\alpha}.
\]
Therefore 
\[
\norm{\AA^{\alpha}}_{\infty} \leq 1- U^{-\alpha}.
\]
Since
\[
(\II-\AA)^{-1} = \sum_{k=0}^{\infty} \AA^{k} = (\II + \AA + \AA^2 + \cdots + \AA^{\alpha}) \sum_{k=0}^{\infty} \AA^{k (\alpha+1)},
\]
we have
\begin{align*}
\norm{(\II-\AA)^{-1}}_{\infty}
& \leq 
(\alpha+1) \sum_{k=0}^{\infty} \left( 1- U^{-\alpha} \right)^k
\leq (\alpha+1) U^{\alpha}.
\end{align*}
Now note that 
$
\MM_{\overline{S},\overline{S}} = \DD(\II-\AA)
$.
Therefore since for any $n$-by-$n$ matrix $\NN$, $\norm{\NN}_{2} \leq \sqrt{n} \norm{\NN}_{\infty}$,
\begin{align*}
\kappa(\MM_{\overline{S},\overline{S}}) 
& = \norm{\MM_{\overline{S},\overline{S}}}_2 \cdot \norm{\MM_{\overline{S},\overline{S}}^{-1}}_2 
\leq 
\norm{\DD}_2 \norm{\II-\AA}_2 \norm{\DD^{-1}}_2 \norm{(\II-\AA)^{-1}}_2 
\\ & \leq 
n U \cdot \norm{\II-\AA}_{\infty} \norm{(\II-\AA)^{-1}}_{\infty}
\\ & \leq
2 nU \cdot (\alpha+1) U^{\alpha}.
\end{align*}
The result immediately follows from above by the assumption $n>3$.
\end{proof}

We are now ready to prove that the \callalg{ShortcutterFramework} computes a valid entrywise approximation of the inverse of an RDDL matrix with high probability.

\begin{theorem} [Correctness of the shortcutter framework]\label{thm:shortcutter-correctness}
Given $\epsilon \in (0,1)$, integer $T\geq 10$, $U \ge 2$, an $n \times n$ invertible RDDL matrix $\LL$ with integer entries in $[-U, U]$, and a nonnegative vector $\bb$ with integer weights in $[0,U]$, \callalg{ShortcutterFramework} produces a matrix $\ZZ$ such that $\ZZ \approxbar_{\epsilon} \LL^{-1}$ with probability at least $1-2\delta$.
\end{theorem}
\begin{proof}
Let $G=([n+1],E)$ be the graph associated with $L$.
Recall that $2^k$ is the smallest power of two larger than $\sqrt{n}\log(8 n^2\epsilon^{-1}\delta^{-1})$. 
For $i,j \in [n]$, we call $(i,j)$ a near pair if $d^{G(\overline{S})}(i,j)\leq 2^k$. Otherwise, we call $(i,j)$ a far pair. 

Note that for $i,j \in \overline{S}$, if $d^{G}(i,j) > 2^k$, then $(i,j)$ is a far pair. However, if $d^{G}(i,j) \leq 2^k$, $(i,j)$ might be a far or near pair. Our proof shows that \callalg{ShortcutterFramework} outputs a valid approximation for entry $(i,j)$ for the following three cases separately.
\begin{enumerate}
    \item $i\in S$ or $j\in S$.
    \item $i,j \in \overline{S}$ and $d^G(i,j) > 2^k$.
    \item $i,j \in \overline{S}$ and $d^G(i,j) \leq 2^k$.
\end{enumerate}

Let $\DD=\diag(\LL)$ and $\AA = \DD^{-1}(\DD - \LL)$. Then $\LL^{-1} = (\II - \AA)^{-1} \DD^{-1}$. Recall that entry $(i,j)$ of $(\II-\AA)^{-1}$ is the sum of the probabilities of random walks in $W_{ij}^{(n+1)}$.

We start with the first case, i.e., $i\in S$ or $j\in S$. If $i\in S$, then the absorption probabilities of $i$ onto $\widetilde{S}$ is all zero except for $i$ itself for which the probability is equal to one, i.e., $\UU_{i:} = \ee^{(f_i)} \in \R^{|S|}$, where $i \rightarrow f_i$ is the one-to-one monotone mapping from $S$ to $[|S|]$. Therefore, $(\UU \VV^\top)_{ij} = \UU_{i:} (\VV_{j:})^\top = \VV_{j,f_i}$, i.e., $(\UU \VV^\top)_{ij}$ is the sum of the probabilities of random walks in $W^{(n+1)}_{ij}$ up to a multiplicative factor of $\exp(\epsilon/2)$. Therefore, since $\XX_{ij} = 0$, $\XX_{ij} + (\UU \VV^\top)_{ij} \DD^{-1}_{jj} \approx_{\epsilon/2} (\LL^{-1})_{ij}$. Thus, since $\ZZ_{ij} \approx_{\epsilon/4} \XX_{ij} + (\UU \VV^\top)_{ij} \DD^{-1}_{jj}$, the output for such entries is correct.

Now suppose $j \in S$ and $i \notin S$. 
We can partition the set of random walks $W^{(n+1)}_{ij}$ to $|S|$ sets $W^{(n+1)}_{ij}(k)$, $k\in S$, where $W^{(n+1)}_{ij}(k)$ is the set of random walks in $W^{(n+1)}_{ij}$ that hit vertex $k$ before any other vertex in $S$.
For any $k \in S$, the absorption probability of $i$ to $k$ with respect to $S \cup \{n+1\}$ is the sum of probabilities of all random walks in $W_{ik}^{(n+1)}$ that do not hit any vertex in $(S \cup \{n+1\}) \setminus \{j\}$---this can be seen by examining \eqref{eq:absorption-lin-system}. Therefore the sum of probabilities of random walks in $W^{(n+1)}_{ij}(k)$ is the absorption probability of $i$ to $k$ with respect to $S \cup \{n+1\}$ multiplied by the sum of the probabilities of the random walks in $W^{(n+1)}_{kj}$ because any random walk in $W_{ij}^{(n+1)}(k)$ can be split to two parts, where the first part of the random walk is until the first time it hits $k$, and the second part is the rest of the random walk. Thus
\[
(\II-\AA)^{-1}_{ij} \approxbar_{\epsilon/2} (\UU \VV^\top)_{ij} = \UU_{i:} (\VV_{j:})^\top,
\]
and since $\XX_{ij} = 0$, $\ZZ_{ij} \approx_{\epsilon/4} \XX_{ij} + (\UU \VV^\top)_{ij} \DD^{-1}_{jj}$. Therefore the output for such entries is also correct.

We now turn to the second case in which $i,j \in \overline{S}$ and $d^G(i,j) > 2^k$. Note that in this case, $d^{G(S)}(i,j) > 2^k$ and therefore $(i,j) \notin N$ and $\XX_{ij} = 0$. Consider a random walk $w \in W^{(n+1)}_{ij}$. Since $d^G(i,j) > 2^k$, $w$ uses at least $2^k$ unique vertices other than $i$ and $j$. Since each $k \in [n]$ is included in $S$ independently with probability $\frac{1}{\sqrt{n}}$, the probability that no vertex of $w$ is in $S$ is at most
\[
\left(1 - \frac{1}{\sqrt{n}} \right)^{2^k} < \left(1 - \frac{1}{\sqrt{n}} \right)^{8\sqrt{n} \log(n^2\epsilon^{-1} \delta^{-1})} \leq \exp(-\log(8 n^2\epsilon^{-1} \delta^{-1})) = \frac{\epsilon \delta}{8n^2}
\]
Therefore, denoting the sum of probabilities of random walks in $W^{(n+1)}_{ij}$ with $\alpha_{ij}$, and the sum of probabilties of random walks in $W^{(n+1)}_{ij}$ that hit at least one vertex in $S$ with $\beta_{ij}$, by linearity of expectation
\[
\E[\beta_{ij}] \geq (1- \frac{\epsilon\delta}{8n ^2}) \alpha_{ij},
\]
where the expectation is over the choice of set $S$.
Let $\theta_{ij} = \alpha_{ij} - \beta_{ij}$, which is a non-negative random variable. Therefore, by Markov's inequality,
\[
\P[\theta_{ij} \geq \frac{\epsilon}{8}\alpha_{ij}] \leq \frac{\E[\theta_{ij}]}{\epsilon \alpha_{ij}} \leq \frac{\frac{\epsilon \delta \alpha_{ij}}{8 n^2}}{\frac{\epsilon \alpha_{ij}}{8}} = \frac{ \delta}{n^2}.
\]
Because $\theta_{ij} = \alpha_{ij}-\beta_{ij}$, this implies that the probability of $\beta_{ij}\leq (1-\frac{\epsilon}{8}) \alpha_{ij}$ is at most $\frac{\delta}{n^2}$.
Since there are at most $n^2$ pairs of vertices, by the union bound, with probability at least $1-\delta$, $\beta_{ij} \approx_{\epsilon/4} \alpha_{ij}$ for all pairs $(i,j)$ with $d^{G}(i,j) > 2^k$.

For $k \in S$, let $W^{(n+1)}_{ij}(k)$ be the set of random walks in $W^{(n+1)}_{ij}$ that hit vertex $k$ before any other vertex in $S$. Note that since $i,j \in \overline{S}$, $W^{(n+1)}_{ij} \setminus \bigcup_{k \in S} W^{(n+1)}_{ij}(k)$ is not necessarily empty, but with high probability the sum of probabilities of random walks in $W^{(n+1)}_{ij}$ is within a factor of $\exp(\epsilon/4)$ of the sum of probabilities of random walks in $\bigcup_{k \in S} W^{(n+1)}_{ij}(k)$, which is within a factor of $\exp(\epsilon/2)$ of $(\UU \VV^\top)_{ij}$, where the latter follows from a similar argument to the previous case. Combining these and noting that $\XX_{ij} = 0$ and $\ZZ_{ij} \approx_{\epsilon/4} \XX_{ij} + (\UU \VV^\top)_{ij} \DD^{-1}_{jj}$, we conclude that the output for entries of the second case is also a valid approximation.

We now consider the third case for which $i,j \in \overline{S}$ and $d^G(i,j) \leq 2^k$. Let $\widetilde{S} = S \cup \{n+1\}$. Then the random walks in $W^{(n+1)}_{ij}$ are partitioned to $W^{(\widetilde{S})}_{ij}$, and $W^{(n+1)}_{ij}(k)$, $k \in S$. Let $\widetilde{W}^{(n+1)}_{ij}(k)$ be the set of random walks in $W^{(n+1)}_{ij}(k)$ with at least $2^{k}$ unique vertices other than $i$ and $j$. 
We have two cases, either $d^{G(\overline{S})}(i,j) > 2^k$ or $d^{G(\overline{S})}(i,j) \leq 2^k$.

If $d^{G(\overline{S})}(i,j) > 2^k$, then any random walk $w$ in $W^{(\widetilde{S})}_{ij}$ has at least $2^{k}$ unique vertices other than $i$ and $j$. Therefore, with a similar argument to the previous part, we can conclude that the sum of the probabilities of random walks in $\bigcup_{k \in S} \widetilde{W}^{(n+1)}_{ij}(k)$ is within a factor of $\exp(\epsilon/4)$ of the sum of the probabilities of random walks in $W^{(\widetilde{S})}_{ij}\cup \bigcup_{k \in S} \widetilde{W}^{(n+1)}_{ij}(k)$. Therefore $\alpha_{ij} \approxbar_{\epsilon/4} \beta_{ij}$ with probability at least $\frac{\delta}{n^2}$. Moreover $\beta_{ij} \approx_{\epsilon/2} (\UU \VV^\top)_{ij}$, and by combining these and noting that $\XX_{ij} = 0$ and $\ZZ_{ij} \approx_{\epsilon/4} \XX_{ij} + (\UU \VV^\top)_{ij} \DD^{-1}_{jj}$, it follows that the output for entries of the second case is also a valid approximation.

We finally consider the case with $d^{G(\overline{S})}(i,j) \leq 2^k$. In this case $\XX_{ij} \neq 0$. 
Step \ref{alg-item:subcube-algo-near-pairs} of \callalg{ShortcutterFramework} guarantees that for all pairs $(i,j) \in N$, $\YY_{g_i g_j} \approx_{\epsilon/4} (\LL_{\overline{S},\overline{S}})^{-1}_{g_i g_j}$.
Since $\LL_{\overline{S},\overline{S}} = \DD_{\overline{S},\overline{S}} (\II - \AA)_{\overline{S},\overline{S}}$, 
$(\LL_{\overline{S},\overline{S}})^{-1}_{g_i g_j}$ is
the sum of the probabilities of random walks in $W_{ij}^{(\widetilde{S})}$ multiplied by $\DD^{-1}_{jj}$, and $\YY_{g_i g_j}$ is an $\exp(\epsilon/4)$ approximation of this quantity.

Now note that with an argument similar to the previous cases, the sum of the probabilities of random walks in $\bigcup_{k \in S} W^{(n+1)}_{ij}(k)$ is within a factor of $\exp(\epsilon/2)$ of $(\UU \VV^\top)_{ij}$. Therefore since $\XX_{ij} = \YY_{g_i g_j}$, we have $\XX_{ij} + (\UU \VV^\top)_{ij} \DD^{-1}_{jj} \approx_{\epsilon/2} (\LL^{-1})_{ij}$. Thus, since $\ZZ_{ij} \approx_{\epsilon/4} \XX_{ij} + (\UU \VV^\top)_{ij} \DD^{-1}_{jj}$, the output for the entries of third case is also correct and this concludes the proof.
\end{proof}

\subsection{Algorithm for Dense Graphs}
\label{sec:subcubic-dense-graph}

In this section, we show that the \callalg{ShortcutterFramework} can be implemented to compute an entrywise $\exp(\epsilon)$-approximate inverse of an RDDL matrix with integer weights in $[-U,U]$, using $\Otil(n^{\omega + 0.5} \log(U) \log(\delta^{-1}\epsilon^{-1}))$ bit operations, with probability at least $1-\delta$. To achieve this, we need to show that Steps \ref{alg-item:subcube-algo-set-N}, \ref{alg-item:subcube-algo-near-pairs}, \ref{alg-item:subcube-algo-absorption}, \ref{alg-item:subcube-escape-probs}, and \ref{alg-item:subcube-final-output} of the \callalg{ShortcutterFramework} can be computed within the stated number of bit operations.

We start by stating a couple of lemmas from prior works for computing the inverse of a matrix and solving a linear system in matrix multiplication time. The following lemma will be used for computing an approximate inverse of $\LL_{\overline{S},\overline{S}}$, i.e., Step \ref{alg-item:subcube-algo-near-pairs}.

\begin{lemma}[\cite{DDH07,DDHK07}]
\label{lemma:fast-MM-inversion}
There exists an algorithm that for $\widehat{\epsilon}>0$ and any invertible matrix $\MM\in\Z^{n\times n}$ with condition number bounded by $\kappa$, and the bit complexity of entries bounded by $\log(\kappa)$, computes a matrix $\ZZ$ with $\Otil(n^{\omega} \cdot \log(\kappa/\widehat{\epsilon}))$ bit operations such that
\[
\norm{\MM^{-1} - \ZZ}_{F} \leq \widehat{\epsilon} \cdot \norm{\MM^{-1}}_{F}.
\]
\end{lemma}

The next lemma is used for computing the probabilities associated with set $S$, i.e., Steps \ref{alg-item:subcube-algo-absorption} and \ref{alg-item:subcube-escape-probs}. Note that the running time of the following lemma does not have a dependence on the condition number of the matrix.

\begin{lemma}[\cite{S05}]
\label{lemma:storjohann-lin-sys}
There exists a Las Vegas algorithm that for any invertible matrix $\MM\in\Z^{n\times n}$ and $\bb\in\Z^{n}$ returns $\MM^{-1}\bb\in\Q^{n}$ with $O(n^{\omega}\cdot (\log n)\cdot (\log \vertiii{\MM}_{\infty}+\frac{\log\norm{\bb}_{\infty}}{n}+\log n)\cdot C^2)$ expected number of bit operations, where $C = \log((\log \vertiii{\MM}_{\infty}+\frac{\log\norm{\bb}_{\infty}}{n}+\log n))$.
\end{lemma}

We use the next lemma to show that Step \ref{alg-item:subcube-algo-set-N} of the \callalg{ShortcutterFramework}, i.e., computing the set of pairs $(i,j)$ with $i,j\in\overline{S}$ and $d^{G(\overline{S})}(i,j)$, can be performed with $\Otil(n^{\omega})$ bit operations.

\begin{lemma}
\label{lemma:find-near-pairs-with-fmm}
There exists an algorithm that
given any directed graph $G=(V,E,w)$ with $|V|=n$ and nonnegative integer $k$ computes the set of pairs $(i,j)$ with $d^{G}(i,j) \leq 2^k$ with $\Otil(n^{\omega})$ bit operations.
\end{lemma}
\begin{proof}
Let $\AA$ be the unweighted adjacency matrix of $G$ and $\CC^{(0)}=\II + \AA$. For integer $\ell>0$, let
\begin{align*}
    \CC^{(\ell)} = \textsc{Extract01}(\CC^{(\ell-1)} \CC^{(\ell-1)}),
\end{align*}
where $\textsc{Extract01}(\BB)$ denotes the indicator matrix of matrix $\BB$, i.e., 
\begin{align*}
\textsc{Extract01}(\BB)_{ij} = 1 & \Leftrightarrow \BB_{ij} \neq 0,
\\
\textsc{Extract01}(\BB)_{ij} = 0 & \Leftrightarrow \BB_{ij} = 0.
\end{align*}
Then $\CC^{(\ell)}_{ij} = 1$ if and only if the distance from $i$ to $j$ in $G$ is at most $2^\ell$. Since all entries of $\CC^{(\ell-1)} \CC^{(\ell-1)}$ are in $[0,n]$, we can take a prime number $p>n$ and compute $\CC^{(\ell-1)} \CC^{(\ell-1)}$ over $\mathbb{Z}_p$, which can be performed with $\Otil(n^{\omega})$ bit operations. 

Then given $\CC^{(\ell-1)} \CC^{(\ell-1)}$, $\textsc{Extract01}(\CC^{(\ell-1)} \CC^{(\ell-1)})$ can be computed in $O(n^2)$ time.
Finally since either there is no path from $i$ to $j$ or $d^{G}(i,j)\leq n$, we can assume $k \leq \lceil \log_2 n \rceil$.
Therefore the set of all pairs $(i,j)$ with $d^{G}(i,j) \leq 2^k$ can be computed with $\Otil(n^{\omega})$ bit operations.
\end{proof}

The next lemma is used to show that Step \ref{alg-item:subcube-algo-near-pairs} of the \callalg{ShortcutterFramework} can be performed with the stated number of bit operations.

\begin{lemma}
\label{lemma:all-close-pairs}
There exists an algorithm such that for any $\epsilon, \delta \in (0,1)$ and $n, U > 3$, given any invertible RDDL matrix $\LL \in \mathbb{R}^{n \times n}$ with integer entries in $[-U, U]$, its associated graph $G = ([n+1], E, w)$, and a random subset $S \subseteq [n]$ where each $i \in [n]$ is independently included in $S$ with probability $1/\sqrt{n}$, the algorithm—with probability at least $1 - \delta$ and using $\Otil\big(n^{\omega + 0.5} \log(U) \log(\delta^{-1} \epsilon^{-1})\big)$ bit operations—computes $\exp(\epsilon)$-approximations of $(\LL_{\overline{S}, \overline{S}}^{-1})_{g_i g_j}$ for all pairs $i, j \in \overline{S}$ satisfying $d^{G(\overline{S})}(i, j) \leq 2^k$, where:
\begin{itemize}
\item $\overline{S} = [n] \setminus S$,
  \item $2^k$ is the smallest power of two greater than $\sqrt{n} \log(n^2 \epsilon^{-1} \delta^{-1})$, and
  \item $i \mapsto g_i$ is the unique monotone mapping from $\overline{S}$ to $[|\overline{S}|]$.
\end{itemize}
\end{lemma}

\begin{proof}
Let
\[
\DD = \diag(\LL), ~~ \text{and} ~~ \AA = \DD^{-1} (\DD - \LL).
\]
Note that $\LL_{\overline{S},\overline{S}} = \DD_{\overline{S},\overline{S}} (\II - \AA)_{\overline{S},\overline{S}}$.
By \Cref{lemma:cond-number-of-submatrix}, the condition number of $\LL_{\overline{S},\overline{S}}$ is at most $\exp(6\sqrt{n} \log(U) \log( n \delta^{-1}))$ with probability at least $1-\delta$. Therefore
by \Cref{lemma:fast-MM-inversion}, using $\Otil\big(n^{\omega + 0.5} \log(U) \log(\delta^{-1} \epsilon^{-1})\big)$ bit operations, we can compute matrix $\YY$ such that
\begin{align}
\norm{\YY - (\LL_{\overline{S},\overline{S}})^{-1}}_2 
& \leq \nonumber
\exp(-20\sqrt{n} \log(U) \log(8 n^2 \delta^{-1}\epsilon^{-1})) \norm{(\LL_{\overline{S},\overline{S}})^{-1}}_2
\\ & \leq \nonumber
\nonumber
\exp(-14(\sqrt{n} \log(U) \log(8 n^2 \delta^{-1} \epsilon^{-1})))
\\ & \leq \label{eq:close-pairs-run-omega-1}
\frac{\epsilon}{20} \cdot U^{-2\sqrt{n} \log(8 n^2 \epsilon^{-1} \delta^{-1})-1},
\end{align}
where the second inequality follows from \Cref{lemma:cond-number-of-submatrix} and $\norm{\LL_{\overline{S},\overline{S}}}_2 \geq 1$ which implies with probability $1-\delta$, we have
\[
\norm{(\LL_{\overline{S},\overline{S}})^{-1}}_2 \leq \kappa(\LL_{\overline{S},\overline{S}})\leq \exp(6\sqrt{n} \log(U) \log( n \delta^{-1})).
\]
Moreover $(\LL_{\overline{S},\overline{S}})^{-1} = ((\II - \AA)_{\overline{S},\overline{S}})^{-1} \DD_{\overline{S},\overline{S}}^{-1}$. Note that $((\II - \AA)_{\overline{S},\overline{S}})^{-1}_{g_i,g_j}$ is the sum of the probabilities of random walks in $W^{S \cup \{n+1\}}_{ij}$. 
Therefore for any $i,j\in \overline{S}$ with $d^{G(\overline{S})}(i,j) \leq 2^k \leq 2\sqrt{n}\log(8n^2\epsilon^{-1}\delta^{-1})$, $((\II - \AA)_{\overline{S},\overline{S}})^{-1}_{g_i,g_j}$ is at least $U^{-2\sqrt{n} \log(8 n^{2} \epsilon^{-1} \delta^{-1})}$. Thus
\begin{align}
    \label{eq:close-pairs-run-omega-2}
(\LL_{\overline{S},\overline{S}}^{-1})_{g_i g_j}=((\II -\AA)_{\overline{S},\overline{S}}^{-1})_{g_i g_j} \cdot \DD^{-1}_{j j} \geq U^{-2\sqrt{n} \log(8 n^{2} \epsilon^{-1} \delta^{-1}) -1}.
\end{align}
Therefore by \eqref{eq:close-pairs-run-omega-1} and \eqref{eq:close-pairs-run-omega-2},
\[
(1-\frac{\epsilon}{20})(\LL_{\overline{S},\overline{S}}^{-1})_{g_i g_j} \leq \YY_{g_i g_j} \leq (1+\frac{\epsilon}{20})  (\LL_{\overline{S},\overline{S}}^{-1})_{g_i g_j}.
\]
Thus by \Cref{fact:approx}, $\YY_{g_i g_j}$ is within a factor of $\exp(\epsilon/10)$ of $(\LL_{\overline{S},\overline{S}}^{-1})_{g_i g_j}$ and this concludes the proof.
\end{proof}

We bound the number of bit operations required for Step~\ref{alg-item:subcube-algo-absorption} of the \callalg{ShortcutterFramework}, i.e., computing the absorption probabilities onto the set $S \cup \{n+1\}$ for all vertices $i \in [n]$, using the following lemma. Note that, due to the use of \Cref{lemma:storjohann-lin-sys}, we compute these probabilities exactly in rational number representation, and then round them to floating-point numbers with an appropriate number of bits to satisfy the approximation guarantees stated in \callalg{ShortcutterFramework}.

\begin{lemma}
\label{lemma:from-all-to-shortcutters}
There exists an algorithm that for any invertible RDDL matrix $\MM\in \R^{n\times n}$ with integer entries in $[-U,U]$, its associated graph $G=([n+1],E,w)$, and $S\subseteq [n+1]$ with $n+1\in S$ computes the absorption probabilities $\P_{\textrm{abs}}([n],S)$ with $\Otil(n^{\omega} |S| \log U )$ bit operations in expectation. 
\end{lemma}
\begin{proof}
For $t \in S$, the absorption probabilities to $t$ is the solution to the following system.
\begin{align*} 
    \begin{cases} 
       \pp_i  =  \sum_{j\in[n+1], j \neq i} \frac{w(i,j)}{\MM_{ii}} \pp_j, & i \in [n]\setminus S \\
       \pp_t = 1 \\
       \pp_{i} = 0, & i \in S \setminus\{t\}
    \end{cases}
\end{align*}
Since $w(i,j)= -\MM_{ij}$ (for $i,j\in[n]$), for $t\in S \setminus\{n+1\}$, letting $F_t := [n] \setminus (S \setminus \{t\})$, we can write the above system as
\begin{align}
\label{eq:absorption-lin-system}
\MM_{F_t,F_t}^{(t)} \pp = \ee^{(t)} \MM_{tt},
\end{align}
where $\MM^{(t)}$
is the matrix obtained by zeroing out the $t$-th row of $\MM$ except for the diagonal entry $\MM_{tt}$.
For $t=n+1$, the corresponding linear system is
\begin{align}
\label{eq:absorption-lin-system2}
\MM_{\overline{S},\overline{S}} \pp = \MM_{\overline{S}:} \boldsymbol{1}
\end{align}
By \Cref{lemma:storjohann-lin-sys}, we can compute the exact rational solution of \eqref{eq:absorption-lin-system} and \eqref{eq:absorption-lin-system2} with $\Otil(n^{\omega} \log U)$ bit operations in expectation (because the algorithm of \Cref{lemma:storjohann-lin-sys} is a Las Vegas algorithm). Doing this for all elements of $t\in S$ gives a total complexity of $\Otil(n^{\omega} |S| \log U)$.  
\end{proof}

We bound the number of bit operations required for Step~\ref{alg-item:subcube-escape-probs} of the \callalg{ShortcutterFramework}, i.e., computing the sum of probabilities of random walks in $W_{st}^{(n+1)}$ for all $s\in S$ and $t\in[n]$, using the following lemma. Again, due to the use of \Cref{lemma:storjohann-lin-sys}, we compute these exactly in rational number representation, and then round them to floating-point numbers with an appropriate number of bits to satisfy the approximation guarantees stated in \callalg{ShortcutterFramework}.

\begin{lemma}
\label{lemma:from-shortcutters-to-all}
There exists an algorithm that for any invertible RDDL matrix $\MM\in \R^{n\times n}$ with integer entries in $[-U,U]$, its associated graph $G=([n+1],E,w)$, and $S\subseteq [n]$ computes the sum of probabilities of random walks in $W_{st}^{(n+1)}$ for all $s\in S, t \in [n]$ with $\Otil(n^{\omega} |S| \log U )$ bit operations in expectation. 
\end{lemma}

\begin{proof}
Let $\AA = \II - \diag(\MM)^{-1}\MM$. Then $(\II-\AA)^{-1}_{st}$ gives the sum of probabilities of random walks in $W_{st}^{(n+1)}$. Note that
\[
(\II-\AA)^{-1} = \MM^{-1} \diag(\MM).
\]
Therefore
\[
(\II-\AA)^{-1}_{st} = \MM^{-1}_{st} \MM_{tt}
\]
To compute $(\MM^{-1})_{s:}$, we can solve the following system
\[
\MM^\top \xx = \ee^{(s)}.
\]
Therefore, by \Cref{lemma:storjohann-lin-sys}, we can compute the exact rational representation of $(\MM^{-1})_{S:}$ with $\Otil(n^{\omega} |S| \log U)$ bit operations. Moreover the numerator and denominator of entries of $(\MM^{-1})_{S:}$ are given by $O(n\log U)$ bits. Therefore computing
\[
(\II-\AA)^{-1}_{S:} = (\MM^{-1})_{S:} \diag(\MM)
\]
can be done with $\Otil(n^2 |S| \log U)$ bit operations. Therefore the total number of bit operations is $\Otil(n^{\omega} |S| \log U)$ to compute the sum of probabilities of random walks in $W_{st}^{(n+1)}$ for all $s\in S, t \in [n]$.
\end{proof}

The next lemma shows that the product $\UU \VV^\top$ in Step~\ref{alg-item:subcube-final-output} of \callalg{ShortcutterFramework} can be computed using $\Otil(n^{\omega} |S| \log U \log \epsilon^{-1})$ bit operations. The approach is to first compute $\UU$ and $\VV$ using \Cref{lemma:from-all-to-shortcutters,lemma:from-shortcutters-to-all}, respectively. Then, we round the inputs to floating-point numbers with $O(\log(n\epsilon^{-1}) + \log(nU))$ bits and compute an entrywise approximation to $\UU \VV^\top$ using floating-point operations and standard matrix multiplication that takes $O(n^2 |S|)$ arithmetic operations.

\begin{lemma}
\label{lemma:compute-shortcutter-contribution}
There exists an algorithm that for any invertible RDDL matrix $\MM \in \R^{n\times n}$, its associated graph $G=([n+1],E,w)$, and $S \subseteq [n]$ computes an $\exp(\epsilon)$ approximation to the sum of probabilities of random walks in $W_{ij}^{(n+1)}$ that go through at least one vertex in $S$ for all $i,j \in [n]$, i.e., $\UU \VV^\top$ in \callalg{ShortcutterFramework}, with $\Otil(n^{\omega} |S| \log U \log \epsilon^{-1})$ bit operations in expectation. 
\end{lemma}

\begin{proof}
Let $\widetilde{W}_{ij}^{(n+1)}$ be the set of random walks in $W_{ij}^{(n+1)}$ that go through at least one vertex in $S$. Let $\widetilde{S} = S \cup \{n+1\}$. 

First note that for $s \in S$, the absorption probability of $i$ to $s$, i.e., $[\P_{\textrm{abs}}([n],\widetilde{S})]_{is}$ is the sum of probabilities of all random walks in $W_{is}^{(n+1)}$ that do not hit any vertex in $\widetilde{S} \setminus \{s\}$---this can be seen by examining \eqref{eq:absorption-lin-system}. Moreover, we can partition the random walks in $\widetilde{W}_{ij}^{(n+1)}$ to $S$ subsets where subset corresponding to $s\in S$ is the set of random walk in $\widetilde{W}_{ij}^{(n+1)}$ that hit $s$ first before hitting any other vertices in $S$. Now consider the set of random walks in $\widetilde{W}_{ij}^{(n+1)}$ that hit $s\in S$ first. Any random walk $w$ in this set can be partitioned to two parts: 1) the starting part of $w$ which hits $s$ for the first time in its last move and does not have any vertices in $S$ other than that; 2) the rest of the random walk. The probability of $w$ can be obtained by multiplying the probabilities of these two parts. Therefore the sum of probabilities of random walks in $\widetilde{W}_{ij}^{(n+1)}$ that hit $s$ first is equal to $[\P_{\textrm{abs}}([n],\widetilde{S})]_{is} \cdot \P(s,j,n+1)$.

Then the sum of probabilities of random walks in $\widetilde{W}_{ij}^{(n+1)}$ is given by
\[
\sum_{s \in S} [\P_{\textrm{abs}}([n],\widetilde{S})]_{is} \cdot \P(s,j,n+1).
\]
Gathering the matrices $\UU,\VV \in \R^{n\times |S|}$ with $\UU_{i s} = [\P_{\textrm{abs}}([n],\widetilde{S})]_{is}$ and $\VV_{i s} = \P(s,i,n+1)$ for $i \in [n]$, $s \in S$, $(\UU \VV^\top)_{ij}$ is equal to the sum of probabilities of random walks in $W_{ij}^{(n+1)}$ that go through at least one vertex in $S$ for all $i,j \in [n]$. Note that since matrices $\UU$ and $\VV$ are obtained by applying \Cref{lemma:storjohann-lin-sys}, the entries are given in rational number representation. Therefore we first obtain a floating point representation of $\UU$ and $\VV$ which we denote by $\widetilde{\UU}$ and $\widetilde{\VV}$ so that the entries are $\exp(\epsilon/n^2)$ approximations. This can be done in $\Otil(n^2 |S| \log U \log \epsilon^{-1})$ bit operations. Then we compute $\widetilde{\UU} \widetilde{\VV}^\top$ which can be performed also with $\Otil(n^2 |S| \log U \log \epsilon^{-1})$ bit operations.
\end{proof}

We are now equipped to prove \Cref{thm:subcube-inverse-of-dense}.
This theorem employs \callalg{ShortcutterFramework}, where Steps \ref{alg-item:subcube-algo-set-N}, \ref{alg-item:subcube-algo-near-pairs}, \ref{alg-item:subcube-algo-absorption}, and \ref{alg-item:subcube-escape-probs} are implemented using methods based on fast matrix multiplication. Since the correctness of \callalg{ShortcutterFramework} is established in \Cref{thm:shortcutter-correctness}, the majority of the proof of \Cref{thm:subcube-inverse-of-dense} focuses on bounding the number of bit operations (using the results above) and analyzing the probability that the algorithm produces a correct output.

\subcubeTheorem*

\begin{proof}
We employ the \callalg{ShortcutterFramework} with probability parameter $\widehat{\delta} = \delta / 4$. The correctness of the output is established in \Cref{thm:shortcutter-correctness} and occurs with probability at least $1 - \delta/2$.

To implement Step~\ref{alg-item:subcube-algo-set-N}, we use \Cref{lemma:find-near-pairs-with-fmm} on $G(\overline{S})$ to compute the set $N$. This requires $\Otil(n^{\omega})$ bit operations.

For Step~\ref{alg-item:subcube-algo-near-pairs}, we use fast matrix multiplication to compute the inverse. By \Cref{lemma:all-close-pairs}, this can be performed with $\Otil(n^{\omega + 0.5} \log(U) \log(\delta^{-1} \epsilon^{-1}))$ bit operations and provides the required approximation guarantee.

For Steps~\ref{alg-item:subcube-algo-absorption} and~\ref{alg-item:subcube-escape-probs}, we use \Cref{lemma:from-all-to-shortcutters} and \Cref{lemma:from-shortcutters-to-all}, respectively, to compute exact rational solutions for the matrices $\UU$ and $\VV$. We then apply \Cref{lemma:compute-shortcutter-contribution} to compute an entrywise $\exp(\epsilon/2)$ approximation of $\UU \VV^\top$. By \Cref{lemma:from-all-to-shortcutters,lemma:from-shortcutters-to-all,lemma:compute-shortcutter-contribution}, all of these steps can be completed using $\Otil(n^{\omega + 0.5} \log(U) \log(\delta^{-1} \epsilon^{-1}))$ bit operations in expectation.

Finally, given the above, Step~\ref{alg-item:subcube-final-output} can be implemented with $\Otil(n^{2.5} \log(U) \log(\delta^{-1} \epsilon^{-1}))$ bit operations. Therefore, the algorithm requires $\Otil(n^{\omega + 0.5} \log(U) \log(\delta^{-1} \epsilon^{-1}))$ bit operations in expectation and succeeds with probability at least $1 - \delta$.
\end{proof}

\subsection{Algorithm for Sparse Graphs}
\label{sec:subcubic-sparse-graphs}

In this section, we show that the \callalg{ShortcutterFramework} can be implemented to compute an entrywise $\exp(\epsilon)$-approximate inverse of an RDDL matrix with $m$ nonzero entries and integer weights in $[-U,U]$, using $\Otil(m n^{1.5+o(1)}\log^{O(1)}(U \kappa\delta^{-1}\epsilon^{-1}))$ bit operations, with probability at least $1-\delta$. Similar to the previous section, we need to show that Steps \ref{alg-item:subcube-algo-set-N}, \ref{alg-item:subcube-algo-near-pairs}, \ref{alg-item:subcube-algo-absorption}, \ref{alg-item:subcube-escape-probs}, and \ref{alg-item:subcube-final-output} of the \callalg{ShortcutterFramework} can be computed within the stated number of bit operations.

Our main tool in this section for solving linear systems and implementing the mentioned steps of \callalg{ShortcutterFramework} for sparse RDDL matrices is \Cref{thm:entrywise-RDDL-solver} which uses the almost-linear-time algorithms from \cite{CKPPRSV17} and our \callalg{ThresholdDecay} framework.

We start by showing that the matrix $\YY$ in Step \ref{alg-item:subcube-algo-near-pairs} of \callalg{ShortcutterFramework} can be computed using the stated number of bit operations.

\begin{lemma}
\label{lemma:sparse-near-pair-compute}
There exists an algorithm such that for any $\epsilon, \delta \in (0,1)$ and $n, U > 3$, given any invertible RDDL matrix $\LL \in \mathbb{R}^{n \times n}$ with $m$ nonzero integer entries in $[-U, U]$, its associated graph $G = ([n+1], E, w)$, and a random subset $S \subseteq [n]$ where each $i \in [n]$ is independently included in $S$ with probability $1/\sqrt{n}$, the algorithm, with probability at least $1 - \delta$ and using $\Otil(m n^{1.5+o(1)}\log^{O(1)}(U \kappa\delta^{-1}\epsilon^{-1}))$ bit operations, computes $\exp(\epsilon)$-approximations of $(\LL_{\overline{S}, \overline{S}}^{-1})_{g_i g_j}$ for all pairs $i, j \in \overline{S}$ satisfying $d^{G(\overline{S})}(i, j) \leq 2^k$, where:
\begin{itemize}
\item $\kappa$ is the condition number $\LL$,
\item $\overline{S} = [n] \setminus S$,
  \item $2^k$ is the smallest power of two greater than $\sqrt{n} \log(n^2 \epsilon^{-1} \delta^{-1})$, and
  \item $i \mapsto g_i$ is the unique monotone mapping from $\overline{S}$ to $[|\overline{S}|]$.
\end{itemize}
\end{lemma}
\begin{proof}
We compute a matrix $\YY$ such that $\YY_{g_i g_j} \approx_{\epsilon} (\LL_{\overline{S}, \overline{S}}^{-1})_{g_i g_j}$ for all pairs $i, j \in \overline{S}$ satisfying $d^{G(\overline{S})}(i, j) \leq 2^k$.
Let $\DD = \diag(\LL)$ and $\AA = \DD^{-1} (\DD - \LL)$. Therefore 
\[
(\LL_{\overline{S},\overline{S}}^{-1} \ee^{(g_j)})_{g_i} =
(\LL_{\overline{S},\overline{S}}^{-1})_{g_i g_j} =  ((\II - \AA)_{\overline{S},\overline{S}}^{-1})_{g_i g_j} \DD^{-1}_{jj}
\]
Note that $((\II - \AA)_{\overline{S},\overline{S}}^{-1})_{g_i g_j}$ is the sum of probabilities of random walks in $W^{S \cup \{n+1\}}_{ij}$. Therefore for $i,j \in \overline{S}$ with $d^{G(\overline{S})}(i, j) \leq 2^k \leq 2 \sqrt{n} \log(n^2 \epsilon^{-1} \delta^{-1})$,
\begin{align}
\label{eq:sparse-near-pairs-1}
(\LL_{\overline{S},\overline{S}}^{-1} \ee^{(g_j)})_{g_i} \geq -U^{2 \sqrt{n} \log(n^2 \epsilon^{-1} \delta^{-1}) - 1}.
\end{align}
For $j \in \overline{S}$, we obtain $\YY_{:g_j}$ by solving the linear system $\LL_{\overline{S},\overline{S}} \xx = \ee^{(g_j)}$. 
Since the condition number of $\LL_{\overline{S},\overline{S}}$ is bounded by $\kappa$, by 
\Cref{thm:entrywise-RDDL-solver}, we can compute set $A_{g_j}$ and $\widetilde{\xx}^{(g_j)}$ with $\Otil((m+n^{1+o(1)}) \cdot T \log^{O(1)}(U\kappa\epsilon^{-1} \delta^{-1}))$ bit operations and probability $1-\delta/n$ such that for all $i \in \overline{S}$ with $(\LL_{\overline{S},\overline{S}}^{-1} \ee^{(g_j)})_{g_i} \geq (nU)^{-(T+1)}$, we have $g_i \in A_{g_j}$ and $(\LL_{\overline{S},\overline{S}}^{-1} \ee^{(g_j)})_{g_i} \approx_{\epsilon} \widetilde{\xx}^{(g_j)}_{g_i}$. Taking
\[
T = \lceil 2 \sqrt{n} \log(n^2 \epsilon^{-1} \delta^{-1}) \rceil,
\]
$A_{g_j}$ contains all $g_i$ with $(\LL_{\overline{S},\overline{S}}^{-1} \ee^{(g_j)})_{g_i} \geq -U^{2 \sqrt{n} \log(n^2 \epsilon^{-1} \delta^{-1}) - 1}$. Therefore $A_{g_j}$ contains all $g_i$ with $d^{G(\overline{S})}(i, j) \leq 2^k$. Moreover by \Cref{thm:entrywise-RDDL-solver}, for all such $i$, $(\LL_{\overline{S},\overline{S}}^{-1} \ee^{(g_j)})_{g_i} \approx_{\epsilon} \widetilde{\xx}^{(g_j)}_{g_i}$, and $\widetilde{\xx}^{(g_j)}$ is obtained using $\Otil((m+n^{1+o(1)}) \cdot \sqrt{n} \log^{O(1)}(U\kappa\epsilon^{-1} \delta^{-1}))$ bit operations.

To obtain the whole $\YY$, we need to solve $|\overline{S}|$ linear systems using \Cref{thm:entrywise-RDDL-solver} with the stated $T$ and probability $1-\delta/n$. Therefore our algorithm succeeds in computing the desired output with probability $1-\delta$ using
\[
\Otil((m+n^{1+o(1)}) \cdot n^{1.5} \log^{O(1)}(U\kappa\epsilon^{-1} \delta^{-1}))
\]
bit operations. The result then follows by noting that $m \geq n$ and therefore $m n^{1.5} + n^{2.5 + o(1)} = O(m n^{1.5+o(1)})$.
\end{proof}

We are now equipped to prove \Cref{thm:subcube-inverse-of-sparse}. Since the correctness of \callalg{ShortcutterFramework} is established in \Cref{thm:shortcutter-correctness}, the proof only requires bounding the probability of success and the number of bit operations needed for implementing Steps \ref{alg-item:subcube-algo-set-N}, \ref{alg-item:subcube-algo-near-pairs}, \ref{alg-item:subcube-algo-absorption}, \ref{alg-item:subcube-escape-probs}, and \ref{alg-item:subcube-final-output}.

\subcubeTheoremSparse*

\begin{proof}
We employ the \callalg{ShortcutterFramework} with probability parameter $\widehat{\delta} = \delta / 4$. The correctness of the output is established in \Cref{thm:shortcutter-correctness} and occurs with probability at least $1 - \delta/2$.

Let $G=([n+1],E,w)$ be the graph associated with $\LL$.
To implement Step~\ref{alg-item:subcube-algo-set-N} and compute set $N$, we use Dijkstra's algorithm for each $i \in \overline{S}$ as source in $G(\overline{S})$. This takes $\Otil(mn \log U)$ bit operations.

For Step~\ref{alg-item:subcube-algo-near-pairs}, we use 
the approach of \Cref{lemma:sparse-near-pair-compute}, with $\widehat{\delta}=\delta/4$, which uses the \callalg{ThresholdDecay} framework and can be performed with $\Otil(m n^{1.5+o(1)}\log^{O(1)}(U \kappa\delta^{-1}\epsilon^{-1}))$ bit operations to provide the required approximation guarantee.

As shown in the proofs of \Cref{lemma:from-all-to-shortcutters} and \Cref{lemma:from-shortcutters-to-all},
computation of Steps~\ref{alg-item:subcube-algo-absorption} and~\ref{alg-item:subcube-escape-probs} only requires solving $|S|$ linear systems for each. Therefore since $|S|=O(\sqrt{n})$, by using \Cref{thm:entrywise-RDDL-solver} with $T=n$ and $\widehat{\delta} = \delta/4$, we can compute matrices $\UU$ and $\VV$ with $\Otil(m n^{1.5+o(1)}\log^{O(1)}(U \kappa\delta^{-1}\epsilon^{-1}))$ bit operations. Moreover with a similar argument to the proof of \Cref{lemma:compute-shortcutter-contribution}, the product $\UU \VV^\top$ can be computed approximately with $\Otil(n^{2.5} \log^{O(1)}(U \kappa\delta^{-1}\epsilon^{-1})))$

Finally, given the above, Step~\ref{alg-item:subcube-final-output} can be implemented with $\Otil(n^{2.5} \log^{O(1)}(U \kappa\delta^{-1}\epsilon^{-1})))$ bit operations. Therefore, the algorithm requires $\Otil(m n^{1.5+o(1)}\log^{O(1)}(U \kappa\delta^{-1}\epsilon^{-1}))$ bit operations and succeeds with probability at least $1 - \delta$.
    
\end{proof}

\section{Floating Point Edge Weights}
\label{sec:recursive-SC}

In this section, we consider invertible RDDL matrices with entries given in $L$-bit floating-point representation, i.e., $L$ bits for the mantissa and $L$ bits for the exponent. Our main algorithmic result in this section is \Cref{thm:main-recursion}, which states that for any such RDDL matrix $\MM$, we can compute a matrix $\ZZ$ such that $\MM^{-1} \approxbar_{\epsilon} \ZZ$ using $\Otil\left(n^3 \cdot \left(L + \log \frac{1}{\epsilon}\right)\right)$ bit operations. Although this running time might seem high, we show that under the All-Pairs Shortest Path (APSP) conjecture \cite{WW10:journal}, our running time is optimal—see \Cref{thm:RDDLInversionToAPSP}.

Our algorithm (\callalg{RecursiveInversion}) is based on a recursive Schur complement procedure that halves the size of the matrix in each iteration and then computes the inverse using the inverse of a princiapl submatrix and inverse of the corresponding Schur complement---see \eqref{eq:sc-inversion}. The required matrix products for computing the inverse in \callalg{RecursiveInversion} are performed via cubic time matrix multiplication. 

A main ingredient of our analysis is that the entries of the inverse of an RDDL matrix are all nonnegative. This is implied by \Cref{lemma:entries-of-inverse} which is proved using the matrix-tree theorem (\Cref{thm:matrix-tree}): each entry of the matrix is the ratio of the determinant of two matrices and by matrix-tree theorem, these determinants are the weighted number of oriented spanning trees in some graphs. The nonnegativity of the entries of the inverses appearing in the algorithm allows us to implement the algorithm using only addition and multiplication on nonnegative numbers, i.e., there is no subtraction. This allows us to achieve entrywise approximation guarantees with only a small number of bits for the numbers in the algorithm.

Another main component of our analysis is that for two invertible RDDL matrices $\MM$ and $\NN$, if $\MM \approxbar_{\epsilon} \NN$ and $\MM \boldsymbol{1} \approxbar_{\epsilon} \NN \boldsymbol{1}$, then $\MM^{-1} \approxbar_{2\epsilon n} \NN^{-1}$ (see \Cref{lemma:ApproxInvert}). This is proven using \Cref{lemma:entries-of-inverse} and allows us to compute only an entrywise approximation for the intermediary matrices (such as the Schur complement) in the algorithm.

\begin{theorem}[Matrix-Tree Theorem \cite{C82}]
\label{thm:matrix-tree}
Let $\MM \in \R^{n\times n}$ be an RDDL matrix such that for all $i\in [n]$, $\sum_{j\in [n]} \MM_{ij} = 0$. Let $i \in [n]$ and $\NN$ be the matrix obtained by removing row $i$ and column $i$ from $\MM$. Then $\det(\NN) = \sum_{F \in \mathcal{F}} \prod_{e\in F} w_e$, where $\mathcal{F}$ is the set of all spanning trees oriented towards vertex $i$ in the graph corresponding to matrix $\MM$. $\prod_{e\in F} w_e$ is the product of the weights of the edges of spanning tree $F$.
\end{theorem}

We colloquially refer to the summation $\sum_{F \in \mathcal{F}} \prod_{e\in F} w_e$ as the \emph{weighted number of spanning trees oriented toward vertex $i$}. 
This is because, in the case of unweighted graphs, the summation is equal to the number of spanning trees. 
The following is a direct consequence of the matrix-tree theorem.

\begin{lemma}
\label{lemma:det-of-rddl}
    Let $\MM \in \R^{n\times n}$ be an RDDL matrix. Then $\det(\MM)$ is equal to the weighted number of spanning trees oriented towards the dummy vertex $n+1$ in the associated graph of $\MM$.
\end{lemma}
\begin{proof}
Let $\YY$ be the RDDL matrix corresponding to the graph associated with $\MM$. $\YY$ is a directed Laplacian and for all $i \in [n+1]$, $\sum_{j \in [n+1]} \YY_{ij} = 0$. Now note that $\MM$ is obtained from $\YY$ by removing row $n+1$ and column $n+1$. Thererfore by \Cref{thm:matrix-tree}, $\det(\MM)$ is equal to the weighted number of spanning trees oriented towards vertex $n+1$ in the associated graph.
\end{proof}

    Note that in \Cref{lemma:det-of-rddl}, if $\MM$ is not connected to the dummy vertex, then there is no spanning tree oriented towards the dummy vertex, and $\det(\MM) = 0$; that is, $\MM$ is singular.
The next lemma provides a formula for the entries of an RDDL matrix as a ratio of oriented weighted spanning trees in two graphs.
This lemma is used to prove \Cref{lemma:ApproxInvert}, which establishes necessary conditions for the inverses of two diagonally dominant matrices to be entrywise approximations of each other.

\begin{lemma}
\label{lemma:entries-of-inverse}
Let $\MM \in \R^{n\times n}$ be an RDDL matrix and $G=([n+1],E,w)$ be the graph associated with $\MM$.
For $i,j \in [n]$ with $i\neq j$, let $G^{(i,j)}$ be the graph with vertex set $[n+1] \setminus \{j\}$ obtained from $G$ as the following: (1) for all $k \in [n] \setminus \{i,j\}$, the weight of the edge from $k$ to $i$ is replaced by the weight of the edge from $k$ to $j$ plus the weight of the edge from $k$ to $n+1$; (2) the weight of the edge from $i$ to $n+1$ is replaced with the weight of the edge from $j$ to $i$;
(3) for all $k \in [n] \setminus \{i,j\}$, the weight of the edge from $i$ to $k$ is replaced by the weight of the edges from $j$ to $k$;
(4) for all $k \in [n] \setminus \{i,j\}$, the weight of the edge from $k$ to $n+1$ is replaced by the weight of the edge from $k$ to $i$;

Moreover for $i \in [n]$, let $G^{(i,i)}$ be the graph with vertex set $[n+1] \setminus \{i\}$ obtained from $G$ as the following: for all $k \in [n] \setminus \{i\}$, add the weight of the edge from $k$ to $i$ to the weight of the edge from $k$ to $n+1$.

Then if $\MM$ is invertible, for $i,j \in [n]$,
\[
\MM^{-1}_{ji} = \frac{\text{weighted number of spanning trees in $G^{(i,j)}$ oriented towards vertex $n+1$}}{\text{weighted number of spanning trees in $G$ oriented towards vertex $n+1$}}\,.
\]
\end{lemma}
\begin{proof}
    
    Note that the solution to $\MM \xx = \ee^{(i)}$ gives the $i$th column of $\MM^{-1}$, where $\ee^{(i)}$ is the $i$th standard basis vector. Let $\MM^{(i,j)}$ be the matrix obtained from $\MM$ by replacing the $j$th column with $\ee^{(i)}$. Then by Cramer's rule 
    \begin{align}
    \label{eq:entries-of-inverse}
        \MM^{-1}_{ji} = \frac{\det\left(\MM^{(i,j)}\right)}{\det\left(\MM\right)}.
    \end{align}

By \Cref{lemma:det-of-rddl}, $\det(\MM)$ is the weighted number of spanning trees oriented vertex $n+1$ in graph $G$.
    For $\det(\MM^{(i,j)})$ note that the determinant is equal to
    \begin{align*}
        \det\left(\MM^{(i,j)}\right) = \sum_{\sigma\in S_n} \sign(\sigma) \prod_{k=1}^n \MM^{(i,j)}_{k,\sigma(k)}.
    \end{align*}
    Note that if $\sigma(k)=j$ and $k\neq i$, then $\prod_{k=1}^n \MM^{(i,j)}_{k,\sigma(k)} = 0$. This is because we have a term $\MM_{i,\ell}$ in the product with $\ell\neq j$ and any such term is equal to zero.
    Now let $\widehat{\MM}^{(i,j)}$ to be the matrix obtained from $\MM^{(i,j)}$ by removing the $i$th row and $j$th column. Then $\det(\MM^{(i,j)}) = (-1)^{i+j}\det(\widehat{\MM}^{(i,j)})$. Note that $\widehat{\MM}^{(i,j)}$ is an $(n-1)$-by-$(n-1)$ matrix.
    
    There are three cases based on how $i$ and $j$ compare.
    
    \textbf{Case 1:} $i = j$.
    $\widehat{\MM}^{(i,i)}$ is a principal submatrix of $\MM$ and therefore it is invertible and RDDL. Therefore \Cref{lemma:det-of-rddl} immediately gives the result.

    \textbf{Case 2:} $i > j$.
    We first push rows $j+1,\ldots,i-1$ of $\widehat{\MM}^{(i,i)}$ up and move row $j$ to become row $i-1$ to obtain the matrix $\overline{\MM}^{(i,j)}$. This requires first switching rows $j$ and $j+1$, then switching rows $j+1$ and $j+2$, and so on. 
    Since we have $i-1-j$ row switches, then
    \[
    \det \left(\widehat{\MM}^{(i,j)} \right)
    = (-1)^{i-1-j} \det \left(\overline{\MM}^{(i,j)} \right).
    \]
    Therefore, 
    \[
    \det\left(\MM^{(i,j)} \right) = - \det\left(\overline{\MM}^{(i,j)} \right).
    \]
    Note that all the diagonal entries of $\overline{\MM}^{(i,j)}$ are positive, except for $\overline{\MM}^{(i,j)}_{i-1,i-1}$.
    We now construct $\widetilde{\MM}^{(i,j)}$ from $\overline{\MM}^{(i,j)}$ via the following two steps:
    \begin{enumerate}
    \item We first take the sum of all columns of $\overline{\MM}^{(i,j)}$ except column $i-1$ and add it to column $i-1$;
    \item We negate the resulting column $i-1$.
    \end{enumerate}
    First, note that the first operation does not change the determinant, and the second operation changes the sign of the determinant.
    Therefore
    \begin{align}
    \label{eq:m-tilde-ij-det}
    \det \left(\MM^{(i,j)} \right) = \det \left(\widetilde{\MM}^{(i,j)} \right).
    \end{align}
    Since for row $i-1$ of $\overline{\MM}^{(i,j)}$, all the entries are non-positive,
    \[
    \widetilde{\MM}^{(i,j)}_{i-1,i-1} \geq 0.
    \]
    Moreover, by construction 
    \[
    \sum_{k\in[n-1]\setminus\{i-1\}} |\widetilde{\MM}^{(i,j)}_{i-1,k}| \leq |\widetilde{\MM}^{(i,j)}_{i-1,i-1}|.
    \]
    For $\ell\neq i-1$, 
    \[
    \widetilde{\MM}^{(i,j)}_{\ell,i-1} = - \overline{\MM}^{(i,j)}_{\ell,\ell}  - \sum_{k\in[n-1]\setminus \{\ell\}} \overline{\MM}^{(i,j)}_{\ell,k} 
    \]
    Therefore since for $k\in [n-1]\setminus \{\ell\}$, $\overline{\MM}^{(i,j)}_{\ell,k} \leq 0$, and 
    \[
    \sum_{k\in[n-1]\setminus \{\ell\}} |\overline{\MM}^{(i,j)}_{\ell,k}| \leq \overline{\MM}^{(i,j)}_{\ell,\ell},
    \]
    we have,
    \[
    \sum_{k\in[n-1]\setminus \{\ell\}} \widetilde{\MM}^{(i,j)}_{\ell,k} = \overline{\MM}^{(i,j)}_{\ell,\ell} - \left( \sum_{k\in[n-1]\setminus \{\ell\}} |\overline{\MM}^{(i,j)}_{\ell,k}| \right) + \left( \sum_{k\in[n-1]\setminus \{\ell\}} |\overline{\MM}^{(i,j)}_{\ell,k}| \right) = \widetilde{\MM}^{(i,j)}_{\ell,\ell}.
    \]
    Therefore $\widetilde{\MM}^{(i,j)}$ is RDDL. 
    Moreover $G^{(i,j)}$ is the matrix associated with $\widetilde{\MM}^{(i,j)}$ up to relabeling of vertices.
    Now applying \Cref{lemma:det-of-rddl} to $\widetilde{\MM}^{(i,j)}$ proves the result for this case.

    \paragraph{Case 3:} $i < j$.
    This is very similar to the $i > j$ case above and therefore we omit its proof. The only difference between this case and Case 2 is that we need a different row switching to obtain the matrix $\overline{\MM}^{(i,j)}$. More specifically, we need to push rows $i,\ldots, j-2$ of $\widehat{\MM}^{(i,j)}$ down and move row $j-1$ to row $i$. Moreover, we need to obtain $\widetilde{\MM}^{(i,j)}$ from $\overline{\MM}^{(i,j)}$ by changing column $i$ instead of column $i-1$. Again, $G^{(i,j)}$ is the matrix associated with $\widetilde{\MM}^{(i,j)}$ up to relabeling of vertices.
\end{proof}

We are now equipped to prove that if two invertible RDDL matrices are entrywise approximations of each other, and their row sums are also entrywise approximations of each other, then their inverses are entrywise approximations of each other as well.
We need this lemma to show that \callalg{RecursiveInversion} is stable and outputs a vector with the approximation guarantees stated in \Cref{thm:main-recursion}.

\begin{lemma}
\label{lemma:ApproxInvert}
Let $\MM,\NN \in \R^{n \times n}$ be invertible RDDL matrices such that
\[
\MM \approxbar_{\epsilon} \NN,
\]
and for all $i\in[n]$,
\[
\sum_{j\in \left[n\right]} \MM_{ij}
\approx_{\epsilon}
\sum_{j\in \left[n\right]} \NN_{ij}.
\]
Then $\MM^{-1} \approxbar_{2\epsilon n} \NN^{-1}$.
\end{lemma}

Note that just the first condition in \Cref{lemma:ApproxInvert} alone is not sufficient for the inverses to be close to each other. For example, consider matrices
\[
\MM=
\left[
\begin{matrix}
1 + \epsilon & -1\\
-1 & 1
\end{matrix}
\right],
\text{ and }
\NN=
\left[
\begin{matrix}
1 + \epsilon/M &  -1\\
-1 & 1
\end{matrix}
\right],
\]
for $M \geq 1$.
Note that $\MM \approxbar_{\epsilon} \NN$ for any $\MM \geq 1$ but 
\[
(\MM \boldsymbol{1})_1 = \epsilon  = M \cdot \frac{\epsilon}{M} = M \cdot (\NN \boldsymbol{1})_1.
\]
Moreover
\[
\MM^{-1}=
\left[
\begin{matrix}
1/\epsilon & 1/\epsilon\\
1/\epsilon & 1 + 1/\epsilon
\end{matrix}
\right],
\text{ and }
\NN^{-1}=
\left[
\begin{matrix}
M/\epsilon & M/\epsilon\\
M/\epsilon & 1 + M/\epsilon
\end{matrix}
\right],
\]
Therefore entries of $\MM^{-1}$ and $\NN^{-1}$ are bounded away from each other by a factor of $M$.

\begin{proof}[Proof of \Cref{lemma:ApproxInvert}]
    We define $\widetilde{\NN}^{(i,j)}$ similar to the definition of $\widetilde{\MM}^{(i,j)}$ in the proof of \Cref{lemma:entries-of-inverse}.
    By \eqref{eq:entries-of-inverse} and \eqref{eq:m-tilde-ij-det}, we have
    \[
    \MM^{-1}_{ij} = \frac{\det\left(\widetilde{\MM}^{(i,j)}\right)}{\det\left(\MM\right)}
    \]
    and
    \[
    \NN^{-1}_{ij} = \frac{\det\left(\widetilde{\NN}^{(i,j)}\right)}{\det\left(\NN\right)}
    \]
    First note that since $\det\left(\MM\right)$ and $\det\left(\NN\right)$ are the sum of the product of the weight of the edges of spanning trees and the edge weights in graphs corresponding to $\MM$ and $\NN$ are within a factor of $e^{\epsilon}$ from each other, $\det\left(\NN\right) \approxbar_{\epsilon n} \det\left(\MM\right)$. A similar argument applies to Case 1 in the proof of \Cref{lemma:entries-of-inverse}, which gives 
    \[
    \det\left(\widehat{\MM}^{(i,i)}\right) \approxbar_{\epsilon n} \det\left(\widehat{\NN}^{(i,i)}\right).
    \]
    
    Now consider Case 2 in the proof of \Cref{lemma:entries-of-inverse}. Trivially the edge weights in $\widetilde{\MM}^{(i,j)}$ and $\widetilde{\NN}^{(i,j)}$ outside column $i-1$ are within a factor of $e^{\epsilon}$ of each other. For edge weights in column $i-1$, note that for $\ell\neq i-1$,
    \[
    \widetilde{\MM}^{(i,j)}_{\ell,i-1} = - \overline{\MM}^{(i,j)}_{\ell,\ell}  - \sum_{k\in[n-1]\setminus \{\ell\}} \overline{\MM}^{(i,j)}_{\ell,k} 
    \]
    and 
    \[
    \widetilde{\NN}^{(i,j)}_{\ell,i-1} = - \overline{\NN}^{(i,j)}_{\ell,\ell}  - \sum_{k\in[n-1]\setminus \{\ell\}} \overline{\NN}^{(i,j)}_{\ell,k}. 
    \]
    Let $\ell'$ be the index of the row in $\MM$ corresponding to row $\ell$ in $\widetilde{\MM}^{(i,j)}$. Then
    \[
    \widetilde{\MM}^{(i,j)}_{\ell,i-1}
    =
    -\sum_{k \in [n] \setminus\{j\}} \MM_{\ell',k}
    =
    \left(-\sum_{k \in [n]} \MM_{\ell',k}\right) + \MM_{\ell',j}.
    \]
    Similarly, we have 
    \[
    \widetilde{\NN}^{(i,j)}_{\ell,i-1}
    =
    \left(-\sum_{k \in [n]} \NN_{\ell',k}\right) + \NN_{\ell',j}.
    \]
    Since $\MM$ and $\NN$ are RDDL matrices, $(-\sum_{k \in [n]} \MM_{\ell',k})$, $(-\sum_{k \in [n]} \NN_{\ell',k})$, $\MM_{\ell',j}$, $\NN_{\ell',j}$ are nonpositive numbers. Therefore, since $\sum_{k\in \left[n\right]} \MM_{\ell',k}
\approx_{\epsilon}
\sum_{k\in \left[n\right]} \NN_{\ell',k}$ and $\MM_{\ell',j} \approx_{\epsilon} \NN_{\ell',j}$, we have 
\[
\widetilde{\MM}^{(i,j)}_{\ell,i-1} \approx_{\epsilon n} \widetilde{\NN}^{(i,j)}_{\ell,i-1}.
\]
    Thus, 
    \[
    \det\left(\widetilde{\NN}^{(i,j)}\right) \approxbar_{\epsilon n} \det\left(\widetilde{\MM}^{(i,j)}\right).
    \]
    Therefore, we have
    \[
    \MM^{-1}_{ij} \approxbar_{2\epsilon n} \NN^{-1}_{ij}.
    \]
    The result follows similarly for Case 3 in the proof of \Cref{lemma:entries-of-inverse} as well.
\end{proof}

\subsection{Approximate Inversion Using Excess Vector}
\label{sec:float-inversion}
In this section, we prove that the \callalg{RecursiveInversion} algorithm computes an approximate inverse of an invertible RDDL matrix and can be implemented using $\widetilde{\mathcal{O}}\left(n^3 \cdot \left(L + \log \frac{1}{\epsilon}\right)\right)$ bit operations. \callalg{RecursiveInversion} is based on computing blocks of the inverse using Equation~\eqref{eq:sc-inversion}:
\begin{align*}
\MM^{-1} = \begin{bmatrix}
    \AA^{-1} + \AA^{-1} \BB \SS^{-1} \CC \AA^{-1} & -\AA^{-1} \BB \SS^{-1} \\
    - \SS^{-1} \CC \AA^{-1} & \SS^{-1}
\end{bmatrix}.
\end{align*}

The main difference between the \callalg{RecursiveInversion} algorithm and a naive application of Equation~\eqref{eq:sc-inversion} is the excess vector $\vv$, which is provided as input to the algorithm. Let $G = ([n+1], E, w)$ be the graph associated with the RDDL matrix $\MM$. The vector $\vv$ represents the weights of the edges from vertices in $[n]$ to the dummy vertex $n+1$. 

Note that under exact arithmetic, each entry $\vv_i$ can be computed as $\vv_i = -(\MM_{ii} + \sum_{j \in [n] \setminus \{i\}} \MM_{ij})$, where $\MM_{ii}$ is a positive number and $\MM_{ij}$'s are nonpositive. However, under floating-point arithmetic, this \emph{subtraction} can introduce significant error: the value of $\vv_i$ computed this way might not even be an entrywise approximation of $-(\MM_{ii} + \sum_{j \in [n] \setminus \{i\}} \MM_{ij})$. For example, if the true difference is very small, floating-point subtraction might yield a negative number, even though the exact result is positive.

This issue is further complicated by the recursive nature of \callalg{RecursiveInversion}, which involves Schur complements. Specifically, the diagonal entries of the Schur complement are smaller than those of the original matrix. If we attempt to compute these using subtraction, floating-point arithmetic may prevent us from obtaining an entrywise approximation using a small number of bits.

To address this, we \emph{simulate} the update of diagonal entries by updating the excess vector $\vv$. For instance, Step~\ref{algo-item:float-first-excess-vector-update} updates the excess vector for the principal submatrix, which is then passed recursively to \callalg{RecursiveInversion} in Step~\ref{ln:ff-induction}. Additionally, Step~\ref{algo-item:float-second-excess-vector-update} updates the excess vector for the Schur complement.

In summary, throughout the execution of \callalg{RecursiveInversion}, for an RDDL matrix $\XX$—which may be the original matrix, a principal submatrix, or a Schur complement—and its corresponding excess vector $\vv$, the expression $-(\vv_i + \sum_{j \in [n] \setminus \{i\}} \XX_{ij})$ approximates $\XX_{ii}$, even when computed under floating-point arithmetic.

\begin{figure}[t]
\begin{algbox}
\textbf{\uline{\textsc{RecursiveInversion}}}:
        \\
        \uline{Input}: $\MM$: an $n \times n$ invertible RDDL matrix with entries represented in floating-points, \\
            $\vv$: a nonpositive vector,
            \\
            $\epsilon$: the target accuracy. \\
        \uline{Output:} Matrix $\ZZ$ such that $\ZZ \approxbar_{\epsilon} \MM^{-1}$.
\begin{enumerate}
\item If $\MM$ is a $1$-by-$1$ matrix, return $1/\vv_1$ with $\epsilon$ precision. \label{ln:base-induction}
\item Let $F$ be the set of indices of the first $\lceil \frac{n}{2} \rceil$ rows of $\MM$ and $C$ be the set of indices of the rest of the rows.
\item \label{algo-item:float-first-excess-vector-update}
$\widetilde{\pp} \leftarrow \vv_{F} + \MM_{FC} \boldsymbol{1}$ with $\exp(\epsilon/(160n^8))$ approximation.
\item $\ZZ(FF) \leftarrow \callalg{RecursiveInversion}(\MM_{F, F}, \widetilde{\pp} ,\epsilon / (80n^7))$. \label{ln:ff-induction}
\item \label{algo-item:float-second-excess-vector-update}
Form $\SS \leftarrow \MM_{CC} - \MM_{CF} \ZZ(FF) \MM_{FC}$ and $\widetilde{\uu} \leftarrow \vv_{C} - \MM_{CF} \ZZ(FF) \vv_{F}$ with $\exp(\epsilon/(40 n^7))$ approximation.
\item $\ZZ(CC) \leftarrow  \callalg{RecursiveInversion}(\SS, \widetilde{\uu}, \epsilon / (20n^7))$. \label{ln:cc-induction}
\item Compute and return with $\exp(\epsilon/(5n))$ approximation \label{ln:last-line}
\[
\ZZ \leftarrow 
\begin{bmatrix}
\ZZ(FF) + \ZZ(FF) \MM_{FC} \ZZ(CC) \MM_{CF} \ZZ(FF)
&-\ZZ(FF) \MM_{FC} \ZZ(CC)\\
-\ZZ(CC) \MM_{FC} \ZZ(FF) & \ZZ(CC)
\end{bmatrix}.
\]
\end{enumerate}
\end{algbox}
\caption{Psuedocode for recursive inversion
with floating-point input.}
\labelalg{RecursiveInversion}
\end{figure}

\recursionTheorem*
\begin{proof}
We prove the theorem by induction. The base case for a $1$-by-$1$ matrix trivially follows from the construction of Line \ref{ln:base-induction} of the algorithm.
    We now prove $\ZZ(FF) \approxbar_{\epsilon/(40n^7)} \left(\NN_{FF}\right)^{-1}$. Let $\pp = \vv_{F} + \MM_{FC} \boldsymbol{1}$ and $\widetilde{\pp}$ be the floating point approximation of $\pp$. We have $\pp \approxbar_{\epsilon/(160 n^8)} \widetilde{\pp}$, since each entry of $\pp$ is a sum of nonpositive values.
    
    Let $\widetilde{\NN}_{FF}$ be the matrix such that
    \[
    \widetilde{\NN}_{ij}
    =
    \begin{cases}
    \NN_{ij}, & \text{for $i\neq j$},\\
    \widetilde{\pp}_i + \sum_{\widehat{j} \in F\setminus\{i\}} \MM_{i\widehat{j}}, & \text{for $i = j$}.
    \end{cases}
    \]
    
    First note that $\widetilde{\NN}_{FF} \approxbar_{\epsilon/(160n^8)} \NN_{FF}$ and also row sums approximate each other with the same approximation parameter.
    Therefore by \Cref{lemma:ApproxInvert}, $(\widetilde{\NN}_{FF})^{-1} \approxbar_{\epsilon/(80n^7)} \left(\NN_{FF}\right)^{-1}$.
    Moreover by induction
    on Line \ref{ln:ff-induction} of \textsc{RecInvert}, we have
    \[
    \ZZ(FF) \approxbar_{\epsilon/(80n^8)} \left(\widetilde{\NN}_{FF}\right)^{-1}.
    \]
    Therefore, we have 
    \[
    \ZZ(FF) \approxbar_{\epsilon/(40 n^7)} \left(\NN_{FF}\right)^{-1}.
    \]
    Now we show 
    \[
    \ZZ(CC) \approxbar_{\epsilon/(5n^5)} \sc(\NN,C)^{-1}.
    \]
    Let $\TT$ be a matrix given by
    \[
    \TT_{ij}
    =
    \begin{cases}
    \SS_{ij} \qquad \text{for $i\neq j$}\\
    -\left(\widetilde{\uu}_{i} + \sum_{k
    \in C \setminus \left\{i\right\}} \SS_{ik}\right)
    \qquad \text{for $i = j$}.
    \end{cases}
    \]
    Then by induction on Line \ref{ln:cc-induction} of \textsc{RecInvert}, we have
    $\ZZ(CC) \approxbar_{\epsilon/(20 n^7)} \TT^{-1}$. Moreover for $i\neq j$, $\TT_{ij} \approx_{\epsilon/(20 n^7)} \sc(\NN,C)_{ij}$.
        Also for each $i\in C$, we have
    \begin{align*}
    \sc(\NN,C)_{ii}
    & = \NN_{ii} - \NN_{i F} \NN_{FF}^{-1} \NN_{F i} 
    \\
    & =
    -\left(\vv_i + \sum_{j\in[n]\setminus\{i\}} \MM_{ij}\right)
    - \NN_{i F} \NN_{FF}^{-1} \NN_{F i}
    \\ & =
    -\left(\vv_i + \sum_{j\in[n]\setminus\{i\}} \NN_{ij}\right)
    - \NN_{i F} \NN_{FF}^{-1} \NN_{F i},
    \end{align*}
    where the approximation we compute satisfies
    \begin{align*}
        \TT_{ii} & = -\left(\widetilde{\uu}_{i} + \sum_{j\in C \setminus \{i\}} \SS_{ij}\right)
        \\ &
        \approx_{\epsilon/(40n^7)} -\left(\vv_i - \MM_{i F} \ZZ(FF) \vv_{F} + \sum_{j \in C \setminus \{i\}} \left(\MM_{i j} - \MM_{i F} \ZZ(FF) \MM_{F j}\right)\right)
        \\ & \approx_{\epsilon/(40n^7)} 
        - \left(\vv_i - \NN_{i F} \NN_{FF}^{-1} \vv_{F} + \sum_{j \in C \setminus \{i\}} \left(\NN_{i j} - \NN_{i F} \NN_{FF}^{-1} \NN_{F j}\right)\right)
        \\ & = 
        - \left(\vv_i - \NN_{i F} \NN_{FF}^{-1} \left(-\NN_{F [n]}\boldsymbol{1}\right) + \sum_{j \in C \setminus \{i\}} \left(\NN_{i j} - \NN_{i F} \NN_{FF}^{-1} \NN_{F j}\right)\right)
        \\ & =
        - \left(\vv_i + \NN_{i F} \NN_{FF}^{-1} \NN_{F F}\boldsymbol{1} + \NN_{i F} \NN_{FF}^{-1} \NN_{F i} + \sum_{j \in C \setminus \{i\}} \NN_{i j}\right)
        \\ & = 
        - \left(\vv_i + \NN_{i F}\boldsymbol{1} + \NN_{i F} \NN_{FF}^{-1} \NN_{F i} + \sum_{j \in C \setminus \{i\}} \NN_{i j}\right)
        \\ & = 
        -\left(\vv_i + \sum_{j\in[n]\setminus\{i\}} \NN_{ij}\right) - \NN_{i F} \NN_{FF}^{-1} \NN_{F i}.
    \end{align*}
Therefore,
\[
\sc(\NN,C)_{ii} \approx_{\epsilon/(20n^7)} \TT_{ii}.
\]
Furthermore,
\begin{align*}
\TT_{ii} + \sum_{j\in C \setminus \left\{i\right\}} \SS_{ij}
& = - \widetilde{\uu}_i
\\
& \approx_{\epsilon/(40n^7)}
-\left(\vv_i - \MM_{i F} \ZZ(FF) \vv_{F}\right)
\\
& \approx_{\epsilon/(40n^7)}
    -\left(\vv_i - \NN_{i F} \NN_{FF}^{-1} \vv_{F}\right)
\\
& =
    -\left(\vv_i - \NN_{i F} \NN_{FF}^{-1} \left(-\NN_{F \left[n\right]}\boldsymbol{1}\right)\right)
\\
& =
-\left(\vv_i + \NN_{i F}\boldsymbol{1}
+ \NN_{i F} \NN_{FF}^{-1} \NN_{F C}\boldsymbol{1}\right)
\\
& =
-\left(\left(-\NN_{i \left[n\right]} \boldsymbol{1}\right)
+ \NN_{i F}\boldsymbol{1}
+ \NN_{i F} \NN_{FF}^{-1} \NN_{F C}\boldsymbol{1}\right)
\\
& = 
-\left(-\NN_{i C} \boldsymbol{1}
+ \NN_{i F} \NN_{FF}^{-1} \NN_{F C}\boldsymbol{1}\right)
\\
& = 
\NN_{i C} \boldsymbol{1}
- \NN_{i F} \NN_{FF}^{-1} \NN_{F C}\boldsymbol{1}.
\end{align*}
Also,
\[
\sc\left(\NN,C\right)_{ii}
+
\sum_{j\in C \setminus \left\{i\right\}}
\sc\left(\NN,C\right)_{ij}
=
\sum_{j\in C} \NN_{ij}
-
\NN_{i F} \NN_{FF}^{-1} \NN_{F j},
\]
so the row sums also approximate each other with a factor of $e^{\epsilon/(20 n^7)}$.
Therefore by \Cref{lemma:ApproxInvert},
\[
\TT^{-1} \approxbar_{\epsilon/(10n^6)} \sc\left(\NN,C\right)^{-1},
\]
and thus $\ZZ(CC) \approxbar_{\epsilon/(5n^6)} \sc(\NN,C)^{-1}$.
Taking these into account for the computation of $\ZZ$ in Line \ref{ln:last-line}, by Equation \ref{eq:sc-inversion}, we have $\ZZ \approxbar_{\epsilon / n} \NN^{-1}$. 

Note that since our recursion goes for at most $O(\log n)$ iterations, the required accuracy at the lowest level is $\frac{\epsilon}{n^{O(\log n)}}$. Therefore it is enough to work with $O(\log(\frac{1}{\epsilon}) + \log^2 n)$-bit floating points. The number of arithmetic operations for computing matrix multiplications is 
\begin{align*}
    O(n^3) + 2 \cdot O( (n/2)^3) + 4 \cdot O( (n/4)^3) + \dots = O(n^3).
\end{align*}
Therefore, taking the bit complexity of the input into account, the total number of bit operations is $\Otil(n^3 \cdot (L + \log(\frac{1}{\epsilon})))$.
\end{proof}

\subsection{Fine-Grained Hardness of Inverting RDDL with Floating-Point Entries}
\label{sec:hardness}
In this section, we prove that the cubic running time of \Cref{thm:main-recursion} is tight under the fine-grained complexity conjecture that there is no truly subcubic-time algorithm for the all-pairs-shortest-path (APSP) problem. Note that this is not in contrast to \Cref{thm:subcube-inverse-of-dense,thm:quadratic-sddm-inversion} since in these theorems the input entries are restricted to be representable with $O(\log U)$ bits in \emph{fixed-point} representation. While in \Cref{thm:main-recursion}, we only restrict the bit lengths of the entries in the \emph{floating-point} representation. 

Our reduction encodes the edge weights of the APSP problem as exponents in the entries of the diagonally dominant matrix. Then the exponent of the probability of a random walk from vertex $i$ to $j$ captures the length of the walk $i \to j$. Since the $(i,j)$ entry of the inverse of the diagonally dominant matrix represents the sum of the probabilities of the random walks from $i$ to $j$, if the probability of the shortest path dominates the probabilities of longer walks, we can recover the length of the shortest path from the exponent of entry $(i,j)$. We make the probability of the shortest path dominant by adding a polynomial factor of $nM$ to the exponent of all entries. 

Before presenting our reduction (\Cref{thm:RDDLInversionToAPSP}), we define the APSP problems.

\begin{definition}[Directed and Undirected APSP Problems]
Given a directed graph $G=(V,E,w)$ with integer edge weights in $[-M,M]$ and no negative cycles, the directed APSP problem requires computing the length of the shortest path 
from $u$ to $v$ for all pairs $u,v\in V$. The undirected APSP problem asks for the distance of all pairs of vertices in an undirected graph with integer edge weights in $[0,M]$.
\end{definition}

It is conjectured that there is no algorithm with running time $\Otil(n^{3-\delta} \log^c M)$ 
with constant $c,\delta>0$ for solving the directed and undirected APSP problems \cite{WW10:journal}. 
In fact, it has been shown that these problems are equivalent in the sense that either there is an algorithm with a \emph{truly subcubic time} for both of them or there is no such algorithm for either \cite{WW10:journal}.
We prove that directed APSP reduces to inverting an RDDL matrix with a constant entrywise approximation, where the entries are represented in floating-point with $O(\log M + \log n)$ bits.
Therefore, if the APSP conjecture holds, then our algorithm for inverting RDDL matrices in cubic time is tight up to polylogarithmic factors. 
If the graph is undirected, then our reduction gives a symmetric matrix, so it reduces the undirected APSP problem to inverting an SDDM matrix.

Without loss of generality, for the directed APSP problem, we can assume that the edge lengths are all non-negative because a price function $p:V \rightarrow \Z$ can be (deterministically) computed in time $O(\sqrt{n} m \log M)$ such that for all $(i,j) \in E$, $\widehat{w}_{ij} := w_{ij} + p(i) - p(j) \geq 0$---for details of such an algorithm, see \cite{G95}. Then the distance $d_w(i,j)$ from vertex $i$ to $j$ on $G=(V,E,w)$ can be computed by
\[
d_{w}(i,j) = d_{\widehat{w}}(i,j) - p_i + p_j.
\]
Moreover, we can assume that all edge weights are positive for both directed and undirected APSP by the following construction. For $w\geq 0$, let $\widehat{w}_{ij} = n w_{ij}$ if $w_{ij} > 0$, and let $\widehat{w}_{ij}=1$ if $w_{ij} = 0$. Now note that for any pair $i,j \in V$, there is a shortest path from $i$ to $j$ in $G$ that uses at most $n-1$ edges of zero weight. Therefore, $n d_w(i,j) \leq d_{\widehat{w}}(i,j) < n d_w(i,j) + n$. Thus,
\[
d_w(i,j) = \left\lfloor \frac{d_{\widehat{w}}(i,j)}{n}  \right\rfloor.
\]

We are now equipped to provide our reduction from the APSP problem to approximately inverting a diagonally dominant matrix.

\APSPreduction*

\begin{proof}
Note that $n^2 M w_{ij}$ is a non-negative integer and since $w_{ij} \leq M$, it can be represented by $O(\log n + \log M)$ bits. 
Moreover, for all $i\in [n]$,
\[
\sum_{j \in [n]} \MM_{ij} \geq 1 - n \cdot 2^{-n^2 M} \geq 0.5,
\]
Therefore $\MM$ is RDDL, and all vertices are connected to the dummy vertex in the associated graph (see \Cref{def:rddl-graph}). Thus, by \Cref{lemma:rddl-invertible}, it is invertible. 
Let $\AA=\II - \MM$.
Then
\[
\MM^{-1} = \left(\II - \AA\right)^{-1}.
\]
The spectral radius of $\AA$ is less than one and the inverse of $\II-\AA$ is given by
\[
\left(\II-\AA\right) = \II + \AA + \AA^{2} + \cdots.
\]
Therefore, for all pairs $i,j \in V$ with $i\neq j$, $\left(\MM^{-1}\right)_{ij}$ equals the sum of the probabilities of random walks starting at $i$ and ending at $j$ in the graph corresponding to $\MM$ without ever hitting the dummy vertex. 
Let $d(i,j)$ be the distance from $i$ to $j$. Then we have a lower bound
\begin{align}
\label{apsp:lower-bound}
  \left(\MM^{-1}\right)_{ij} \geq 2^{-n^2 M d(i,j)}.
\end{align}
Moreover, since the number of edges on the shortest path from $i$ to $j$ is less than $n$, the number of shortest paths from $i$ to $j$ is at most $n^{n}$. For all $k\geq 1$, the number of paths of length $d(i,j)+k$ from $i$ to $j$ is at most $n^{d(i,j)+k}$. Therefore, we obtain an upper bound
\begin{align}
    \left(\MM^{-1}\right)_{ij} 
    & \leq 
    2^{-n^2 M d(i,j)} n^{n} + \sum_{k=1}^{\infty} 2^{-n^2 M (d(i,j)+k)} n^{d(i,j)+k}
    \\ & \leq 
    2^{-n^2 M d(i,j)} n^{n} + \sum_{k=1}^{\infty} 2^{-n^2 M (d(i,j)+1)} n^{d(i,j)+1} \cdot 2^{-k}
    \\ & \leq 
    2^{-n^2 M d(i,j)} n^{n} + 2^{-n^2 M (d(i,j)+1)} n^{d(i,j) + 1}
    \\ & \leq
    2^{-n^2 M d(i,j)} n^{n} + 2^{-n^2 M d(i,j)}
    \\ & <
    \frac{1}{4} \cdot 2^{-n^2 M (d(i,j)-1)}, \label{apsp:upper-bound}
\end{align}
where the second-to-last inequality follows from $d(i,j) \leq (n-1) M$, and the last inequality holds for $n\geq 3$.

Combining \Cref{apsp:lower-bound} and \Cref{apsp:upper-bound} gives
\begin{align*}
    2^{-n^2 M d(i,j)} \le (\MM^{-1})_{ij} < \frac{1}{4} \cdot 2^{-n^2 M (d(i,j)-1)}. 
\end{align*}
Let $\ZZ$ be an entrywise $0.7$-approximation of $\MM^{-1}$, then
\begin{align*}
    \frac{1}{2} \cdot 2^{-n^2 M d(i,j)} \le \ZZ_{ij} < \frac{1}{2} \cdot 2^{-n^2 M (d(i,j)-1)}
\end{align*}
since $e^{0.7} \le 2$.
Therefore, we can recover $d(i,j)$ from $\ZZ_{ij}$ in $O(\log n + \log M)$ time. More precisely,
\[
d(i,j) = - \left \lfloor \frac{\log_2(2\ZZ_{ij})}{n^2 M} \right\rfloor.
\]
\end{proof}

\section{Applications and Further Connections to Graphs}
\label{sec:probabilities}

In this section, we further characterize the connection between RDDL matrices and their inverses through quantities arising from random walks on graphs. Additionally, we study applications of our results to computing hitting times and all-pairs escape probabilities (APEP). We demonstrate that the latter can be computed within the time required for inverting an RDDL (or SDDM) matrix; thus, any of our inversion results can be applied for this purpose.
We begin by proving two lemmas stated in \Cref{sec:prelim}. Then, we discuss a Schur complement perspective on escape probabilities. Finally, we outline our applications.

\invertibleRDDL*
\begin{proof}
    If the associated graph is connected the dummy vertex, we want to show $\MM$ invertible. Let $\DD = \diag(\MM)$.
    First note that $\MM$ does not have a row of all zeros since otherwise the corresponding vertex is an isolated vertex and the graph corresponding to $\MM$ is not connected to the dummy vertex. 
    Therefore all $\DD_{ii}$'s are nonzero and $\DD$ is invertible.
    Let $\NN = \II - \DD^{-1} \MM$ and $k\in\N$. 
    Then $(\NN^{k})_{ij}$ is the probability that a random walk of size $k$ that started at vertex $i$ ends at vertex $j$ without hitting the dummy vertex. 
    Since by assumption, for all vertices, there is a path of positive edge weights to the dummy vertex, as $k\to \infty$, $(\NN^{k})_{ij} \to 0$. 
    Therefore the spectral radius of $\NN$ is strictly less than one. 
    Thus $\DD^{-1} \MM = \II - \NN$ does not have a $0$ eigenvalue and is invertible. 
    Therefore, $\MM$ is also invertible.

    For the other direction, we will show that $\MM$ is not invertible if the associated graph is not connected to the dummy vertex. Let $u \in [n]$ be a vertex that is not connected to the dummy vertex.
    We denote $S$ as the set of all vertices that can be reached by $u$, then the vertices in $S$ are not connected to the dummy vertex as well.
    Observe that $\MM_{S, S}$ is a Laplacian matrix, and thus is not invertible.
    Therefore, $\MM$ is not invertible since it contains $\MM_{S,S}$ as a principal submatrix.
\end{proof}

The following lemma establishes the connection between escape probabilities and the entries of the inverse of the RDDL matrix. The proof relies on the Sherman–Morrison formula.

\inverseRDDLentries*
\begin{proof}
    We fix $t \in [n]$.  Let $\pp_s$ denote the $(s, t, n+1)$-escape probability. Then $\pp$ can be solved from the unique solution to the following system.  
    \begin{align*} 
    \begin{cases} 
       \pp_i  =  \sum_{j\in[n+1], j \neq i} \frac{w(i,j)}{\MM_{ii}} \pp_j & i \in [n]/\{t\} \\
       \pp_t = 1 \\
       \pp_{n+1} = 0
    \end{cases}
    \end{align*}
    Note that $\MM_{ii} = \sum_{j\in[n+1], j \neq i} w(i,j)$ is the sum of the weights of all the outgoing edges from $i$, so $ \frac{w(i,j)}{\MM_{ii}}$ is the transition probability from $i$ to $j$. We can simplify the system to
    \begin{align*} 
    \begin{cases} 
       \pp_i  =  \sum_{j\in[n], j \neq i} \frac{w(i,j)}{\MM_{ii}} \pp_j & i \in [n]/\{t\} \\
       \pp_t = 1 \\
    \end{cases}
    \end{align*}
    where we think of $\pp$ as an $n$-dimensional vector. The system can be reformulated into matrix form
    \begin{align*}
        \MM^{(t)} \pp = \ee^{(t)} \cdot \MM_{tt}
    \end{align*}
    where $\MM^{(t)}$ is the matrix obtained by zeroing out the $t$-th row except for the diagonal entry $\MM_{tt}$. Let $\vv^{(t)}$ be the only nonzero row of $\MM^{(t)}-\MM$ as a column vector, such that  
    \[
    \MM^{(t)} = \MM + \ee^{(t)} (\vv^{(t)})^\top.
    \]
    Since $\MM$ is an $L$-matrix, $\vv^{(t)} \geq 0$.
    Then by Sherman-Morrison identity, we have
    \begin{align}
    \nonumber
    (\MM^{(t)})^{-1} \ee^{(t)}  &  = \MM^{-1} \ee^{(t)} - \frac{\MM^{-1}\ee^{(t)}(\vv^{(t)})^\top \MM^{-1}\ee^{(t)} }{1+ (\vv^{(t)})^\top \MM^{-1} \ee^{(t)}}
    \\ & =
    \nonumber
    \MM^{-1} \ee^{(t)} \cdot \left( 1- \frac{(\vv^{(t)})^\top \MM^{-1} \ee^{(t)}}{1+ (\vv^{(t)})^\top \MM^{-1} \ee^{(t)}} \right)
    \\ & =
    \label{eq:all-pairs-formula}
    \MM^{-1} \ee^{(t)} \cdot \left( \frac{1}{1+ (\vv^{(t)})^\top \MM^{-1} \ee^{(t)}} \right)
    \end{align}
    We plug in $ (\vv^{(t)})^\top = (\ee^{(t)})^\top (\MM^{(t)} - \MM)$, then
    \begin{align*}
        1+ (\vv^{(t)})^\top \MM^{-1} \ee^{(t)}
        =
        (\ee^{(t)})^\top \MM^{(t)} \MM^{(-1)} \ee^t = \MM_{tt} (\MM^{-1})_{tt}.
    \end{align*}
    Since $\vv^{(t)},\ee^{(t)}, \MM^{-1} \geq 0$, we have $\MM_{tt} > 0$ and $\MM^{(-1)}_{tt} > 0$. Therefore,
    \begin{align*}
        \mathbb{P}(s, t, n+1) = (\ee^{(s)})^\top (\MM^{(t)})^{-1} \ee^{(t)} \cdot \MM_{tt} = \frac{(\ee^{(s)})^\top \MM^{-1} \ee^{(t)}}{\MM_{tt}(\MM^{-1})_{tt}} = \frac{(\MM^{-1})_{st}}{\MM_{tt}(\MM^{-1})_{tt}}. 
    \end{align*}
\end{proof}

\paragraph{Schur complement view.}
Here we take a more careful look at the columns corresponding to $t$ and $p$ in the inverse of $\II-\AA$ with the help of Schur complement.
Without loss of generality suppose $t=n-1$ and $p=n$. Then matrix $\II-\AA$ is as the following.
\begin{align*}
\II - \AA = \begin{bmatrix}
\MM & \BB \\
\CC & \DD
\end{bmatrix},
\end{align*}
where 
\begin{align*}
\DD = \begin{bmatrix}
    1 & 0 \\
    0 & 1
\end{bmatrix},
\end{align*}
and $\CC$ is a $2$-by-$(n-2)$ matrix of all zeros. Then by \eqref{eq:sc-inversion}, we have
\begin{align*}
\left(\II-\AA\right)^{-1}_{:,n-1:n} = \begin{bmatrix}
- \left(\MM - \BB \DD^{-1} \CC\right)^{-1} \BB \DD^{-1}
\\
\DD^{-1} + \DD^{-1} \CC \left(\MM - \BB \DD^{-1} \CC\right)^{-1} \BB\DD^{-1}
\end{bmatrix}
\end{align*}
Therefore noting that $\DD$ is the identity matrix and $\CC$ is an all zero matrix, we have
\begin{align}
\label{eq:escape-formula}
\left(\II-\AA\right)^{-1}_{:,(n-1):n} = 
\begin{bmatrix}
- \MM^{-1} \BB
\\
\DD
\end{bmatrix}
\end{align}
Therefore computing $- \MM^{-1} \BB$, which corresponds to solving two linear systems, gives both escape probabilities $\P(s,t,p)$ and $\P(s,p,t)$ for $s\in V \setminus \{t,p\}$. Also, note that $\P(t,t,p) = \P(p,p,t) = 1$, and $\P(p,t,p) = \P(t,p,t) = 0$.

\paragraph{All-pairs escape probabilities.} Given a (directed or undirected) graph $G=(V,E,w)$, and $p \in V$ our goal is to compute the escape probabilities $\P(s,t,p)$ for all $s,t \in V$.
Let $\LL$ be the (directed or undirected) Laplacian corresponding to $G$ and $\LL^{(p)}$ be the matrix obtained from $\LL$ by removing row and column $p$ from $\LL$.

By \Cref{lemma:RDDL-escape}, we have
\[
\mathbb{P}(s, t, p) = \frac{((\LL^{(p)})^{-1})_{st}}{\LL^{(p)}_{tt}((\LL^{(p)})^{-1})_{tt}}
\]
Thus, denoting $\DD = \diag((\LL^{(p)})^{-1})$, the $(s,t)$ entry of $(\LL^{(p)})^{-1} \DD^{-1}$ gives the escape probabilities $\P(s,t,p)$. Consequently, to obtain all-pairs escape probabilities, it suffices to invert $\LL^{(p)}$ and subsequently multiply the result by a diagonal matrix.

Therefore, by \Cref{thm:quadratic-sddm-inversion}, with probability $1-\delta$, we can compute $\exp(\eps)$-approximations of APEP for undirected graphs with integer weights in $[0,U]$ using $\Otil\bigl(m n \log^2(U \eps^{-1} \delta^{-1})\bigr)$ bit operations. For directed matrices, we can compute $\exp(\eps)$-approximations of APEP with either $\Otil\bigl(m n^{1.5 + o(1)} \log^{O(1)}(U \kappa \delta^{-1} \epsilon^{-1})\bigr)$ or $\Otil\bigl(n^{\omega + 0.5} \log(U \epsilon^{-1}) \log(\epsilon^{-1} \delta^{-1})\bigr)$ bit operations, by \Cref{thm:subcube-inverse-of-sparse} and \Cref{thm:subcube-inverse-of-dense}, respectively. Finally, by \Cref{thm:main-recursion}, for directed graphs with weights represented in $L$-bit floating-point format, we can compute $\exp(\eps)$-approximations of APEP with $\Otil\left(n^3 \cdot \left(L + \log \frac{1}{\epsilon}\right)\right)$ bit operations.

\paragraph{Hitting times.}
Observe that the hitting time to $t$,
$\HH_{:t}$, satisfy the following recurrences:
\begin{align}
\label{eq:hitting-time-recurrence}
\HH_{it}
=
\begin{cases}
0 & \text{if $i = t$},\\
1 + \sum_{j \sim i} \AA_{ij} \HH_{jt} & \text{if $i \neq t$}.
\end{cases}
\end{align}
Both the equation and the variable that is $0$ can be
dropped, giving the equation
\[
\LL_{\neq t, \neq t} \HH_{\neq t, t}
=
\boldsymbol{1}.
\]
Thus, for undirected graphs with integer weights in $[0,U]$, by \Cref{thm:entrywise-SDDM-solver-fast}, we can compute $\exp(\epsilon)$-approximations of all hitting times to vertex $t$ with probability $1-\delta$ using $\Otil\bigl(m \sqrt{n} \log^2(U \epsilon^{-1} \delta^{-1})\bigr)$ bit operations. For directed graphs, by \Cref{thm:entrywise-RDDL-solver}, we can compute such approximations using $\Otil\bigl(m n^{1+o(1)} \log^{O(1)}(U \kappa \epsilon^{-1} \delta^{-1})\bigr)$ bit operations. Finally, for directed graphs with weights represented in $L$-bit floating-point format, by \Cref{thm:main-recursion}, we can achieve this with $\Otil\left(n^3 \cdot \left(L + \log \frac{1}{\epsilon}\right)\right)$ bit operations.

\section{Conclusion}
\label{sec:conclusion}

We presented several algorithms for computing entrywise approximations of the inverse of diagonally dominant matrices and for solving linear systems involving such matrices.

Our reduction from the APSP problem to the inversion of RDDL matrices with floating-point entries suggests a close connection between these problems. Therefore, a better understanding of solutions to RDDL systems with floating-point inputs may lead to improved algorithms for single-source shortest paths (SSSP). We believe this connection motivates further investigation.

For fixed-point inputs, there are three immediate open problems of interest in this direction. Our $\Otil(mn)$-time algorithm for sparse matrices (\Cref{thm:quadratic-sddm-inversion}) requires symmetry—that is, the corresponding graph must be undirected. It is an open question whether an $\Otil(mn)$-time algorithm is possible for RDDL matrices, i.e., directed graphs.

Another interesting direction is to develop algorithms with running time better than $m \sqrt{n}$ (perhaps near-linear time) for entrywise approximate solutions to \emph{single linear systems}, for both SDDM and RDDL matrices.

It is also of interest to investigate whether algorithms with running time better than $mn^{1.5+o(1)}$ or $n^{\omega+0.5}$ can be developed for inverting RDDL matrices.

\section*{Acknowledgement}
We thank Richard Peng, Jonathan A. Kelner, and Jingbang Chen for helpful discussions.

\bibliographystyle{alpha}
\bibliography{main}

\appendix

\section{Repeated Squaring}
\label{sec:repeated-square}

As mentioned in \Cref{sec:prelim}, if $\II - \AA$ is invertible, the inverse can be computed by the means of the power series $\II + \AA + \AA^{2} + \cdots$. 
However, to design an algorithm from this, we need to cut the power series 
 and take the summation for a finite number of terms in the power series. 
 In other words, we need to output $\II + \AA + \AA^{2} + \cdots + \AA^{k}$ (for some $k\in\N$) as an approximate inverse. 
 Then such an approximate inverse can be computed efficiently by utilizing a repeated squaring approach. 
 
 In this section, we give bounds for the required number of terms to produce appropriate approximations. 
 A disadvantage of this approach compared to the recursion approach that we discuss in the next section is that the running time would have a dependence on the logarithm of hitting times to $t$ and $p$.

 We start by formally defining the set of random walks that hit a certain vertex only in their last step and use that to define hitting times. We then show that such hitting times characterize how fast the terms in the power series decay. Finally, we use that to bound the number of terms needed to approximate the inverse well.

\begin{definition}
\label{def:hitting-time}
    For a (directed or undirected and weighted or unweighted) graph $G=(V,E)$ and $s,t\in V$, we define the hitting time of $s$ to $t$, as the expected number of steps for a random walk starting from $s$ to reach $t$ for the first time. 
    More formally, let 
    \begin{align}
        W_{st}=\{w=(v_1,\ldots,v_k): k\in\N, \forall i\in [k], \forall j\in[k-1], v_i\in V, (v_j,v_{j+1}) \in E,v_j\neq t, v_1=s, v_k = t\},
    \end{align}
    be the set of all possible walks from $s$ to $t$ that visit $t$ only once. We set the probability of the random walk $w=(v_1,\ldots,v_k)$ equal to $\pr(w):= \pr(v_1,v_2) \pr(v_2,v_3) \cdots \pr(v_{k-1},v_k)$, where $\pr(u,v)$ is the probability of going from vertex $u$ to $v$ in one step. 
    We also denote the size of $w$ with $|w|$ which is the number of edges traversed in $w$, i.e., for $w=(v_1,\ldots,v_k)$, $|w|=k-1$. Then the hitting time of $s$ to $t$ is
    \begin{align}
        H(s,t):=\sum_{w\in W_{st}} \pr(w) \cdot |w|.
    \end{align}
    If there is no path from $s$ to $t$, i.e., $W_{st}=\emptyset$, then $H(s,t)=\infty$.
\end{definition}

The way \Cref{def:hitting-time} defines the hitting time, it only works for one vertex. 
To define the hitting time for hitting either $t$ or $p$, we can consider the graph $\overline{G}$ obtained from $G$ by contracting $t$ and $p$. 
We denote the vertex corresponding to the contraction of $t$ and $p$ with $q$.
Then the probability of going from $s$ to $q$ (in one step) in $\overline{G}$ denoted by $\overline{\pr}(s,q)$ is equal to $\pr(s,t)+\pr(s,p)$. 
Then the hitting time of $s$ to $q$ in $\overline{G}$ can be defined in the same manner as \Cref{def:hitting-time} and we denote it with $\overline{H}(s,q)$. 
The following lemma bounds the decay of the terms of the power series using this hitting time.

\begin{lemma}
\label{lemma:2t}
For a (directed or undirected and weighted or unweighted) graph $G=(V,E)$ and $s,t,p\in V$, let $\overline{G} =(\overline{V},\overline{E})$ be the graph obtained by contracting (identifying) $t$ and $p$. 
Let $q$ be the vertex corresponding to $t$ and $p$ in $\overline{G}$.

Let $\AA$ be the random walk matrix associated with graph $G=(V,E)$ in which the rows corresponding to $t$ and $p$ are zeroed out.
Then for any $h$ and any $k \ge 2h$ such that
\[
h
\geq
1+\max_{s \in \overline{V}} \overline{H}\left(s,q)\right.
\]
we have 
\[
\left\| \left(\AA^{k}\right)_{u:} \right\|_1
\leq
\frac{1}{2}
\qquad
\forall u \in V . 
\]
That is, any row of $\AA^{k}$ sums up to less than $0.5$, and thus $\|\AA^k\|_\infty \le 0.5.$
\end{lemma}

\begin{proof}
    By definition, we have $\overline{H}(s,q) \leq h$, for any $s\in \overline{V}$. Let $w$ be a random walk from $s$ to $q$ in $\overline{G}$. We have $\E(|w|)< h$. Therefore, by the Markov inequality, $\P[|w|\leq k] \geq \P[|w|\leq 2h] \geq \frac{1}{2}$. Note that entry $u \in \overline{V}\setminus \{q\}$ in row $s$ of $\AA^{k}$ denotes the probability that a random walk of size $k$ in $\overline{G}$ that starts at $s$ ends at $u$. We need to clarify here that a random walk that reaches $t$ (or $p$) in $2h$ steps for the first time stays at $t$ (or $p$). However, the probability of a random walk that reaches $t$ (or $p$) in the $k$'th step and stays there does not get added to $\AA^{k+1}$ (or any of the terms after that) because row $t$ (and $p$) of $\AA$ are zero and therefore the last term in the probability of the random walk is zero. This argument clarifies that we do not ``double count'' the probability of random walks in this calculation. 
    
    Note that $\P[|w|\leq k] \geq \frac{1}{2}$ implies that any random walk of size $k$ ends up at $q$ with probability at least $0.5$ since the random walk does not exit $q$ when it reaches it. Therefore the total probability for all the other vertices in $\overline{V}\setminus \{q\}$ is at most $0.5$. Therefore $\| (\AA^{k})_{u:} \|_1 \leq \frac{1}{2}$.
\end{proof}

Here, we recall the property of matrix infinity norm. For $k \ge 2h$, multiplication by $\AA^{k}$ makes any vector smaller. 

\begin{corollary} \label{coro:infinity-norm}
    For any vector $\vv$,  we have $\| \AA^{k} \vv \|_{\infty} \leq  \| \AA^{k} \|_{\infty} \| \vv \|_\infty \le \frac{1}{2}\|\vv\|_{\infty}$.
\end{corollary}

We need the following lemma for the maximum hitting time for unweighted undirected graphs to bound the number of terms required in the power series.

\begin{lemma}[\hspace{-0.05pt}\cite{F17}]
\label{lemma:m2}
    For an undirected unweighted graph $G$, the maximum hitting time is at most $m^2$.
\end{lemma}

We are now prepared to prove the main result of this section which gives an algorithm for computing escape probabilities in graphs with bounded polynomial weights.

\begin{theorem}
\label{lemma:2tg}
    Let $G=(V,E)$ be a undirected graph with integer weights in $[1, n^c]$, for $c\in\N$. Given $t,p\in V$, we can compute the escape probability $\P(s,t,p)$ for all $s\in V$ in $\Otil(n^3 c \log(c) \log\frac{1}{\epsilon})$ bit operations within a $e^{\epsilon}$ multiplicative factor.
\end{theorem}

\begin{proof}
Let $\AA$ and $h$ be as defined in \Cref{lemma:2t}. Without loss of generality, we assume the graph is connected. There are three cases that need to be handled first:
\begin{itemize}
    \item All paths between $s$ and $p$ contain $t$: This indicates $\P(s,t,p)=1$. This can be identified by removing $t$ and checking if $s$ and $p$ are disconnected.
    \item All paths between $s$ and $t$ contain $p$: This indicates $\P(s,t,p)=0$. This can be identified by removing $p$ and checking if $s$ and $t$ are disconnected.
    \item There is no path from $s$ to $t$ and $p$. In this case, the escape probability is undefined.
\end{itemize}
The above cases can be handled by running standard search algorithms on the graph (e.g., breadth-first search) and removing the vertices $s$ with such escape probabilities. The running time of such a procedure is smaller than the running time stated in the theorem.
Ruling out these cases, we can assume all vertices $s\in V\setminus\{t,p\}$ have paths with positive edge weights to both $t$ and $p$, $\II-\AA$ is invertible, and $\sum_{i=1}^{\infty} \AA^{i}$ is convergent.

Let $k = (2h+1)(2n(c+1)\lceil \log \frac{(2h+1)n}{\epsilon} \rceil + 1) - 1$, where $h$ is defined as in \Cref{lemma:2t}.
We show
\begin{align*}
\left( \II - \AA \right)^{-1}
\approxbar_{\epsilon}
\sum_{i = 0}^{k}
\AA^{i}.
\end{align*}

First note that 
    \begin{align*}
        \left(\II-\AA\right)^{-1} = \sum_{i = 0}^{\infty} \AA^i.
    \end{align*} 
    Let $\BB=\sum_{i = k +1}^{\infty} \AA^i$. Since $\AA^i \geq 0$, for all $i\geq 0$, $\BB\geq 0$. Therefore, we have
    \begin{align*}
        \left(\II-\AA\right)^{-1} \geq \sum_{i = 0}^{k}\AA^i
    \end{align*}.
    Moreover, we have
    \begin{align}
        \left(\II-\AA\right)^{-1} = (\sum_{j=0}^{\infty} \AA^{(2h+1)j}) (\sum_{i = 0}^{2h}\AA^i).
    \end{align}
    Note that by \Cref{coro:infinity-norm} and induction,
    \begin{align}
        \norm{ \AA^{(2h+1)j} \vv }_{\infty} & =
        \norm{ \AA^{2h+1}(\AA^{(2h+1)(j-1)} \vv)}_{\infty} \nonumber
        \\ & \leq
        \frac{1}{2}  \norm{\AA^{(2h+1)(j-1)} \vv}_{\infty} \nonumber
        \\ & \leq
        \frac{1}{2^j}  \norm{\vv}_{\infty} \label{eq:things_get_small}
    \end{align}

    Note that the $\ell_1$ norm of each column of $\AA^{i}$ is bounded by $n$ because all the entries are in $[0,1]$. Therefore, again because all entries are between $[0,1]$ (i.e., they are non-negative), the $\ell_1$ norm of each column of $\sum_{i = 0}^{2h}\AA^i$ is bounded by $(2h+1)n$. Therefore by setting $\vv$ to be a column of $\sum_{i = 0}^{2h}\AA^i$ and using \eqref{eq:things_get_small} and triangle inequality, we have that 
    \begin{align}
        \norm{(\sum_{j=2n(c+1)\lceil\log((2h+1)n/\epsilon)\rceil+1}^{\infty} \AA^{(2h+1)j}) \vv}_{\infty} 
        & \leq \sum_{j=2n(c+1)\lceil\log((2h+1)n/\epsilon)\rceil+1}^{\infty} \norm{ \AA^{(2h+1)j} \vv }_{\infty} \nonumber
        \\ & \leq 
    \sum_{j=2n(c+1)\lceil\log((2h+1)n/\epsilon) \rceil+1}^{\infty} \frac{1}{2^j}\norm{ \vv }_{\infty} \nonumber
    \\ & \leq 
    \frac{1}{2^{2n(c+1)\lceil\log((2h+1)n/\epsilon) \rceil}} \norm{ \vv }_{\infty} \nonumber
    \\ & \leq
    \frac{\epsilon}{n^{n(c+1)} \cdot (2h+1) n} \norm{ \vv }_{\infty} \nonumber
    \\ & \leq
    \frac{\epsilon}{n^{n(c+1)}}. \label{eq:very-small}
    \end{align}

    Note that any escape probability is at least $n^{-n(c+1)}$ since any $\pr(x,y)$ is at least $1/n^{c+1}$. 
    Therefore by \eqref{eq:very-small}, not considering the terms of the power series with an exponent larger than $(2h+1)(2n(c+1)\lceil\log((2h+1)n/\epsilon)\rceil+1)$ can only cause a multiplicative error of $O(e^{\epsilon})$ in the computation of the escape probability. Therefore our algorithm is to compute the matrix
    \[
    \XX = \sum_{i=0}^{k-1} \AA^{i},
    \]
    and return $\XX_{:t}$,
    where $k$ is the smallest power of two larger than $(2h+1)(2n(c+1)\lceil \log \frac{(2h+1)n}{\epsilon} \rceil + 1)$. To compute $\XX$ we use the recursive approach of repeated squaring. Namely, let $k=2^r$. Then

\begin{align*}
\XX
= \sum_{i=0}^{2^r-1} \AA^i
= \left(\sum_{i=0}^{2^{r-1}-1} \AA^i\right)
\left(\II + \AA^{2^{r-1}}\right)
=
\cdots
=
\left(\II+\AA\right)
\left(\II+\AA^2\right)
\left(\II+\AA^4\right)
\cdots
\left(\II+\AA^{2^{r-1}}\right).
\end{align*}
Therefore, $\XX$ can be computed with $O(\log(k))$ matrix multiplications.

In addition, all of $\AA^{2^{\widehat{r}}}$ ($\widehat{r}\in[r-1]$) can be computed with $\log(k)$ matrix multiplications, by observing that $\AA^{2^{\widehat{r}}} = \AA^{2^{\widehat{r}-1}}\AA^{2^{\widehat{r}-1}}$.
Now, suppose each floating-point operation we perform incurs an error of $e^{\widehat{\epsilon}}$. Therefore if we use an algorithm with $O(n^3)$ number of arithmetic operations to compute the matrix multiplications, then we incur a multiplicative error of at most $e^{2n\widehat{\epsilon}}$. Then the error of computation of $\AA^{2^{\widehat{r}}}$ is at most $e^{2n r\widehat{\epsilon}}$ and the total error of computing $\XX$ is at most $e^{(2n r+1)r\widehat{\epsilon}}$. Setting $\widehat{\epsilon} = \frac{\epsilon}{(2nr+1)r}$, the total error will be at most $e^{\epsilon}$. Note that to achieve this we need to work with floating-point numbers with $O(\log\frac{(2n\log k)\log k}{\epsilon})$ bits. This shows that our algorithm requires only
\[
\Otil(n^3 (\log k)(\log\frac{(2n\log k)\log k}{\epsilon}))
\]
bit operations. To finish the proof, we need to bound $k$. To bound $k$, we need to bound $h$. By \Cref{lemma:m2}, the maximum hitting time for undirected unweighted graphs is $m^2$. Note that this gives a bound for our hitting time to either of $t$ or $p$ as well in the unweighted case since after handling the pathological cases at the beginning of the proof, we assumed that for every vertex $s\in V \setminus \{t,p\}$, there exist paths to both $t$ and $p$. Moreover introducing polynomial weights in the range of $[1,n^c]$ can only increase the hitting time by a factor of $n^c$. Therefore $h=O(n^{c+4})$. Then $k = \Otil(n^{c+5} c \log(\frac{n^{c+5}}{\epsilon}))$. Therefore $\log(k) = \Otil (c \cdot \log c)$. Therefore
\[
n^3 (\log k)(\log\frac{(2n\log k)\log k}{\epsilon}) = \Otil(n^3 \cdot c \cdot \log(c) \cdot \log(\frac{1}{\epsilon})).
\]
\end{proof}

\end{document}